\documentclass[11pt]{article}

\usepackage[margin=1in]{geometry}
\usepackage{setspace}
\newif\ifersubmission
\ersubmissionfalse
\ifersubmission\else\fi

\usepackage{amsmath,amssymb,amsthm}
\usepackage{mathtools}
\usepackage{graphicx}
\usepackage{xcolor}
\usepackage{booktabs}
\usepackage{threeparttable}
\usepackage{array}
\usepackage[inline]{enumitem}
\usepackage[colorlinks,citecolor=blue,linkcolor=blue,urlcolor=blue]{hyperref}
\usepackage{natbib}
\usepackage{docmute}
\usepackage{placeins}

\usepackage[linesnumbered,ruled,vlined]{algorithm2e}

\usepackage{subcaption}

\newcolumntype{C}[1]{>{\centering\arraybackslash}p{#1}}
\allowdisplaybreaks
\usepackage{chngcntr}

\theoremstyle{plain}
\newtheorem{theorem}{Theorem}
\newtheorem{lemma}[theorem]{Lemma}
\newtheorem{proposition}[theorem]{Proposition}
\newtheorem{corollary}[theorem]{Corollary}

\theoremstyle{definition}
\newtheorem{definition}{Definition}
\newtheorem{remark}{Remark}

\newcounter{assumpG}

\newtheorem{assumptionG}[assumpG]{Assumption}
\newcounter{assumpS}

\newtheorem{assumptionS}[assumpS]{Assumption}
\newcounter{assumpSS}

\newtheorem{assumptionSS}[assumpSS]{Assumption}

\DeclareMathOperator*{\argmin}{arg\,min}



\begin{document}

\title{Testing Clustered Equal Predictive Ability with Unknown Clusters}

\author{O\u{g}uzhan Akg\"{u}n\thanks{Corresponding author. LEDi UR 7467, Universit\'{e} Bourgogne Europe, France. Email: oguzhan.akgun@u-bourgogne.fr}
\and Alain Pirotte\thanks{CRED, Paris-Panth\'{e}on-Assas University, France. Email: alain.pirotte@assas-universite.fr}
\and Giovanni Urga\thanks{Bayes Business School (formerly Cass), London, United Kingdom. Email: g.urga@city.ac.uk}
\and Zhenlin Yang\thanks{School of Economics, Singapore Management University, Singapore. Email: zlyang@smu.edu.sg}}

\date{\today}

\maketitle

\begin{abstract}
We develop tests of clustered equal predictive ability (C-EPA) in panels where the clusters are unknown and estimated by Panel Kmeans. Since the hypotheses depend on data-driven clusters, we adopt a selective conditional inference framework. We show that the square root of a Wald statistic for the equality of two estimated cluster centers is asymptotically a truncated $\chi$ variable, conditionally on the selected clusters and their number, with multiple random initializations, and that the truncation set is a finite union of intervals that can be computed exactly. Unlike existing methods of selective inference for Kmeans, this holds under serial and cross-sectional dependence and heteroskedasticity. For the C-EPA hypothesis, we combine the pairwise $p$-values with a test of overall EPA. In Monte Carlo experiments, the proposed tests control size where naive tests severely over-reject. An application to exchange rate forecasting with time series and machine learning methods illustrates the procedure.



\end{abstract}

\newpage

\section{Introduction}\label{sec:intro}

Since the seminal work of \cite{diebold95}, a large and ever-growing literature\footnote{See \cite{giacomini2011}, \cite{clark13}, and \cite{rossi_forecasting_2021} for reviews of the early and more recent contributions to the area.} on testing equal predictive ability (EPA) using time series data has emerged. Despite this vast time series literature, testing EPA with panels has attracted the attention of econometricians only recently.
To the best of our knowledge, the only contributions are those of Akgun, Pirotte, Urga and Yang (\citeyear{akgun24}, APUY hereafter) and \citet{qu23b}.
Both papers focus on two EPA hypotheses: the overall EPA (O-EPA) hypothesis and the clustered EPA (C-EPA) hypothesis, but they assume that the relevant clusters are known to the researcher.

This assumption is restrictive in many empirical forecasting environments. Predictive performance may vary systematically across countries, firms, sectors, forecasters, or assets, but the dimensions that generate such heterogeneity are rarely known ex ante. Observable classifications such as income groups, geographical regions, institutional arrangements, or political alignments may be informative, but they are often coarse, overlapping, unstable over time, or unrelated to the particular loss differential under study. For example, \cite{dreher2008political} show that IMF forecast quality differs with countries' IMF-program status and political alignment with major donors, suggesting that forecast accuracy may depend on institutional and political features not captured by standard country classifications. More generally, the relevant grouping for EPA testing is the grouping induced by relative predictive performance itself, not necessarily by an externally imposed taxonomy. Imposing predetermined clusters may therefore mask important heterogeneity, while testing after estimating clusters without accounting for their data-driven nature leads to invalid inference. This creates the need for C-EPA procedures that both learn the clusters from the data and conduct valid inference conditional on that clustering step.

The main contribution of this paper is to develop conditional C-EPA tests for panel data when the relevant cluster structure is unknown and must be learned from the data.
Our framework advances beyond APUY, \citet{qu23b} and the selective inference literature in five main directions. First, inspired by \cite{giacomini06}, we formulate C-EPA testing in a conditional predictive ability setting, allowing loss differentials to be interacted with general testing functions. Second, we estimate latent clusters using Panel Kmeans and show that, for our purpose, Panel Kmeans is equivalent to standard Kmeans applied to unit-level time averages, thereby reducing the selective-inference problem from $NTP$ to $NP$ dimensions. Third, we develop a selective conditional inference procedure that delivers valid post-clustering inference: for pairwise equality of cluster centers, the square root of a Wald-type statistic has an asymptotic truncated $\chi$ distribution conditional on the selected clustering outcome. Fourth, we make selective inference after Kmeans valid under serial correlation, cross-sectional dependence (CSD) and heteroskedasticity, the typical features of forecast losses. The procedures of \cite{gao24} and \cite{chen23} assume independent observations with a common covariance matrix, and their Euclidean decomposition of the data requires the tested difference to be an eigenvector of the covariance matrix; \cite{gonzalez2023post} show that, even for separable covariance matrices, extending them to dependent observations requires compound symmetry across observations. We decompose the data into the tested difference and a residual in the metric of the long-run covariance matrix, which makes the two independent in the Gaussian limit for any positive definite covariance matrix, so that units may load heterogeneously on common factors and have different long-run variances. The asymptotic argument conditions only on selection events of positive probability.

Fifth, we construct tests of panel-wide homogeneity and of the full C-EPA null by combining selective pairwise $p$-values with an O-EPA $p$-value through $p$-merging functions valid under dependence among the components \citep{vovk20,vovk2022admissible,gasparin2024combining}. The conditioning event is the selection of the clusters by the procedure used in practice, including multiple random initializations and the choice of the number of clusters by an information criterion. Since the initializations are independent of the data, the tests are valid conditionally on them and condition neither on the iterations of the algorithm nor on the initializations that are not retained; the truncation sets are computed exactly. The same tests apply to the joint EPA (J-EPA) hypothesis of APUY, that EPA holds for every unit. Section~\ref{sec:extensions} collects these extensions, together with selective tests of each cluster center which, combined by the step-down procedure of \cite{holm1979simple}, identify the clusters for which EPA fails once the C-EPA hypothesis is rejected.

The central inferential difficulty is that the clusters used in the C-EPA null are estimated from the same loss differentials on which the test is based. Panel Kmeans is a natural tool for uncovering latent heterogeneity in predictive performance, and it is widely used in econometric applications with grouped heterogeneity \citep{lin12,bonhomme15a,sarafidis15,bonhomme2022discretizing,patton23}. However, once clusters are data-driven, conventional tests that treat them as fixed suffer from a post-selection problem: the clustering step mechanically separates units in ways that can create spurious evidence of heterogeneity. This is the familiar problem of `double dipping' \citep{kriegeskorte09}, here arising because the same panel is used both to estimate the clusters and to conduct inference on their centers.

Inference on the selected clusters is of interest for the same reason as inference on the variables selected by the Lasso, or on the best of several policies \citep{lee_exact_2016,andrews2024inference}: practitioners report the output of the algorithm, and statements about it require a valid measure of uncertainty. When the expected loss differentials have a unique Kmeans partition, which allows for heterogeneity within clusters, the estimated partition coincides with it with probability approaching one, and the selective tests become tests about population clusters (Section~\ref{sec:selective}). Under J-EPA, which implies C-EPA for every partition, the selective tests remain valid and the clustering serves to direct power (Section~\ref{sec:unitlevel}).

Sample splitting offers a simple way to avoid reusing the same observations for clustering and testing, and has been used in related panel settings \citep{patton23}. Its simplicity is attractive, but it comes at a cost. The split between training and testing samples is arbitrary, the testing sample may be substantially smaller, and the procedure can be sensitive to structural breaks or other forms of temporal instability. These limitations motivate our selective conditional approach, which uses the full sample while conditioning on the clustering outcome.

Monte Carlo simulations show that naive post-clustering tests are severely oversized, that the selective and Split Sample tests control size when the cluster structure is stable over time, in which case sample splitting is more powerful, a difference that a shrinkage estimator of the long-run covariances narrows considerably, and that only the selective tests remain valid under a break in relative forecast performance.

We illustrate the empirical relevance of the methodology through an application to exchange rate forecasting. Using a panel of bilateral exchange rates against the U.S. dollar, we compare traditional time series models and machine learning methods with an AR(1) benchmark. The results show that accounting for latent clusters can reveal heterogeneous predictive gains that are missed by aggregate forecast-comparison tests.

Section~\ref{sec:preliminaries} presents the hypotheses, three motivating examples, a generalized test of C-EPA with predetermined clusters and the Panel Kmeans estimator. Section~\ref{sec:naive} states the assumptions and a preliminary lemma. Section~\ref{sec:tests} introduces the selective tests of homogeneity and C-EPA and establishes their asymptotic properties, and Section~\ref{sec:extensions} covers the estimation of the number of clusters, the J-EPA hypothesis and the identification of the clusters that violate EPA. Sections~\ref{sec:monte} and~\ref{sec:app} report the Monte Carlo results and the empirical application, and Section~\ref{sec:conc} concludes. The supplement contains the derivations of the examples, the proofs, the Split Sample test, details on the application, the Panel Kmeans algorithm, the computation of the truncation sets, primitive conditions for Assumptions~\ref{ass:clt} and~\ref{ass:stability}, and additional simulation results.

{\em Notation.}\; Random variables are denoted by upper-case letters and their realizations by the corresponding lower-case letters, e.g., $w(\cdot)$ denotes a realization of the test statistic $W(\cdot)$. $\| \cdot \|$ denotes the Euclidean norm, $\mathbf{1} \{ \cdot \}$ the indicator function, $\mathrm{tr}( \cdot )$ the trace of a square matrix, $| \cdot |$ the cardinality of a set, $\lfloor \cdot \rfloor$ the largest integer not exceeding its argument, $\lceil \cdot \rceil$ the smallest integer not less than its argument and $\otimes$ the Kronecker product, and $\mathrm{dir}(v)=v/\lVert v\rVert$ is the direction of a nonzero vector $v$.

\section{Setup and Preliminaries}\label{sec:preliminaries}

This section introduces the hypotheses, three motivating examples, a conditional C-EPA test with predetermined clusters that generalizes APUY, and the Panel Kmeans estimator of the clusters.

\subsection{Testing framework and hypotheses}\label{sec:frame}

Let $\widehat{Y}_{a,it}$ be the $\tau$-steps ahead, $\tau \geq 1$, forecast of agents $a=1,2$ for the target $Y_{it}$ made at time $t-\tau$, $t=1,2,\dots,T$, for unit $i=1,2,\dots,N$. Here, $a$ represents a forecasting agent such as the IMF or the OECD as in APUY and \citet{qu23b}, or a forecasting model \citep[see][among others]{clark2001tests,clark13,clark2014tests,clark2015nested,giacomini06}. A generic loss function is denoted by $L(\cdot,\cdot)$. This can be a quadratic loss, an absolute loss or another loss function that is not necessarily in the forecast error form. Define the loss differentials of the two forecasts as $\Delta L_{it} = L(\widehat{Y}_{1,it},Y_{it}) - L(\widehat{Y}_{2,it},Y_{it})$ which are defined on a complete probability space $(\mathit{\Omega},\mathcal{E},\Pr)$.

Following \cite{giacomini06}, when the forecasts are produced by estimated models, we assume that each forecast made at time $t-\tau$ is a measurable function of the data observed in a rolling window of fixed length $R$ ending at $t-\tau$. The object of comparison is then the forecasting method, including its estimation procedure, and parameter estimation error is part of the forecasts rather than a nuisance; no relation between $R$ and $T$ is required, and nested and non-nested models are treated alike. Under this scheme, the assumptions of Section~\ref{sec:naive} on the loss differentials follow from assumptions on the underlying data (Section~\ref{sec:primitive_forecasts} of the supplement). Expanding windows, and hypotheses on the population parameters of the forecasting models as in \cite{west1996asymptotic}, would require corrections for parameter estimation error that depend on the relative size of $R$ and $T$, and are beyond the scope of this paper.

The null of interest is the generalized C-EPA hypothesis. By `generalized' we mean that it allows conditioning variables, unlike the unconditional null considered by APUY and \citet{qu23b}. It is stated as
\begin{equation}\label{eq:null}
\mathcal{H}_{0}: \frac{1}{|\mathcal{C}_k|} \sum_{i \in \mathcal{C}_k} \mathbb{E}(\Delta L_{it} \mid \mathcal{F}_{t-\tau}) = 0, \, \text{a.s.,} \, \text{ for all } k=1,2,\dots,K\text{ and } t=1,\dots,T,
\end{equation}
where $\mathcal{F}_{t} \subseteq \mathcal{E}$ is a conditioning set (see examples below), and $\left\{\mathcal{C}_k\right\}_{k=1}^K$ are the sets of panel unit indexes that form a partition of the set of all units $\{1,2,\dots,N\}$. More concretely, $\mathcal{C}_k = \{i : k_i = k \}$, where $k_i \in \{1,2,\dots,K\}$ is the cluster membership indicator of the unit $i$. 

The alternative hypothesis is
\begin{equation}\label{eq:alt}
\mathcal{H}_{1}: \frac{1}{|\mathcal{C}_k|} \sum_{i \in \mathcal{C}_k} \mathbb{E}(\Delta L_{it} \mid \mathcal{F}_{t-\tau}) \neq 0 \text{ with positive probability for at least one } k \text{ and } t.
\end{equation}

The null hypothesis requires the conditional expectations to vanish at every $t$ and does not require them to be time invariant. Under the alternative they may vary over time, and the unconditional moments below are therefore defined as averages over the evaluation period, which makes explicit that the tests target average relative performance (see the experiment with structural instability in Section~\ref{sec:supp_breaks} of the supplement).

The two cases covered in the null and alternative hypotheses are important. The first case is the unconditional C-EPA hypothesis which is obtained when $\mathcal{F}_{t} = \{ \emptyset, \mathit{\Omega} \}$. For predetermined clusters, the tests for this null hypothesis were developed by APUY and \citet{qu23b} under different assumptions on the autocorrelation and CSD properties of the loss differentials. The second case is the conditional C-EPA hypothesis, for which two sub-cases are particularly useful. In the first, $\mathcal{F}_{t} = \sigma(\{ W_{is} \}^N_{i=1}, s \leq t)$ is the $\sigma$-field generated by the present and past of the measurable-$\mathcal{E}$ random variables $W_{it} = (Y_{it},X'_{it})'$, with $X_{it}$ a vector of external predictors used to make the predictions $\widehat{Y}_{a,it}$. In the second, a researcher is interested in conditional EPA given the realizations of a vector of measurable-$\mathcal{E}$ common factors $F_t$, so that $\mathcal{F}_{t} = \sigma(F_{s}, s \leq t)$. Dummy variables for, e.g., the global financial crisis or the COVID-19 period make it possible to focus on local differences in the predictive abilities of the two forecasters. Section~\ref{sec:supp_framework} of the supplement discusses combinations of these schemes and the comparison of many forecasts with a common benchmark.

When the testing functions are common to all units of a cluster, that is, $\widetilde{H}_{i,t-\tau}=\widetilde{H}_{k,t-\tau}$ for $i\in\mathcal{C}_k$, with $\widetilde H_{k,t-\tau}$ a measurable-$\mathcal{F}_{t-\tau}$ vector, the law of iterated expectations shows that $\mathcal{H}_{0}$ implies $|\mathcal{C}_k|^{-1}\sum_{i\in\mathcal{C}_k}\mathbb{E}(\widetilde{H}_{i,t-\tau}\Delta L_{it}) = \mathbb{E}[\widetilde H_{k,t-\tau}|\mathcal{C}_k|^{-1}\sum_{i\in\mathcal{C}_k}\mathbb{E}(\Delta L_{it}\mid\mathcal F_{t-\tau})]=0$ \citep{giacomini06}. This is the case for the common factors and dummy variables discussed above. The implication fails when the testing functions differ across units of the same cluster. For example, with $\mathcal{C}_k=\{1,2\}$, $\mathbb{E}(\Delta L_{1t}\mid\mathcal F_{t-\tau})=1$, $\mathbb{E}(\Delta L_{2t}\mid\mathcal F_{t-\tau})=-1$, $\widetilde H_{1,t-\tau}=1$ and $\widetilde H_{2,t-\tau}=2$, the cluster average of the conditional expectations is zero while the average of the unconditional moments is $-1/2$. With unit-specific testing functions, such as the lagged target used in Sections~\ref{sec:monte} and~\ref{sec:app}, we therefore replace $\mathcal H_0$ by the stronger joint EPA (J-EPA) null of APUY
\begin{equation}\label{eq:nullU}
\mathcal{H}^{jepa}_{0}: \mathbb{E}(\Delta L_{it}\mid\mathcal{F}_{t-\tau})=0, \, \text{a.s.,} \, \text{ for all } i=1,\dots,N \text{ and } t=1,\dots,T,
\end{equation}
which implies $\mathbb{E}(\widetilde H_{i,t-\tau}\Delta L_{it})=0$ for every unit and every measurable-$\mathcal F_{t-\tau}$ vector $\widetilde H_{i,t-\tau}$, and hence the cluster-level restrictions below for every partition. In both cases, taking unconditional expectations yields moment conditions that do not involve $\mathcal F_{t-\tau}$.

Let $H_{it}$ be such a $P \times 1$ vector, called a `testing function' by \cite{giacomini06}, where $P$ is the number of testing functions, and $Z_{it} = H_{i,t-\tau}\Delta L_{it}$ with $\mu_{i}^0 = T^{-1}\sum_{t=1}^T\mathbb{E} ( Z_{it} )$, the average expected value over the evaluation period. Define also $\theta^0_{k}(\mathcal{C}) =|\mathcal{C}_k|^{-1} \sum_{i \in \mathcal{C}_k} \mu_i^0$ where $\mathcal{C} = \{ \mathcal{C}_1,\dots,\mathcal{C}_K \}$.
Then, $\mathcal{H}_{0}$ (with common testing functions) or $\mathcal{H}^{jepa}_0$ implies
\begin{equation}\label{eq:h0rewritten}
  \mathcal{H}'_{0}: \theta^0_{k}(\mathcal{C}) = 0, \text{ for all } k=1,2,\dots,K.
\end{equation}
This transformation from conditional to unconditional moments is standard in forecast evaluation and avoids modeling $\mathcal{F}_{t-\tau}$; it does not retain all the information in the conditional restriction, but it preserves enough structure for testing when the testing function is informative. Integrated conditional moment approaches \citep{bierens1982consistent,dominguez2004consistent} avoid this loss of information by using a continuum of testing functions. Within our framework, an intermediate option is a finite grid of generically comprehensively revealing functions, $H_{i,t-\tau}=(1,\exp(s_1'\psi(W_{i,t-\tau})),\dots,\exp(s_J'\psi(W_{i,t-\tau})))'$ for a bounded one-to-one transformation $\psi$ and fixed points $s_1,\dots,s_J$, at the cost of a larger $P$. A test that is consistent against all conditional alternatives is beyond the scope of this paper. In practice, the choice of $H_{i,t-\tau}$, e.g., the lagged loss differentials, regressors, or cluster-specific moments, alters the power properties and interpretation of the resulting test.

\subsection{Examples}\label{sec:examples}

Conditional EPA tests have been widely used since \cite{giacomini06} \citep[see][for a review]{clark13}. The following examples show why the clusters relevant for C-EPA are often unknown.

\paragraph*{Example 1: Time series forecasting}
A benchmark autoregression is often compared with more flexible specifications across many series \citep{diebold95,marcellino2006comparison}. Suppose that the series $\{Y_{it},X_{it}\}_{t=0}^R$, $i=1,\dots,N$, belong to one of two latent clusters,
\begin{equation*}
Y_{it}=\alpha_i\mathbf 1\{i\in\mathcal C_1\}+\beta_iX_{i,t-1}+U_{it},\quad U_{it}\sim iid(0,\sigma^2),
\end{equation*}
with fixed predictors, and that Forecaster 1 predicts $\widehat Y^{(1)}_{i,R+1}=\hat\alpha_i+\hat\beta_iX_{i,R}$ and Forecaster 2 $\widehat Y^{(2)}_{i,R+1}=\tilde\beta_iX_{i,R}$, with least squares estimates from a rolling window of length $R$. Forecaster 2 omits the intercept of the units in $\mathcal C_1$, while Forecaster 1 estimates a superfluous intercept for those in $\mathcal C_2$, so that each forecaster is more accurate on one cluster. The expected loss differentials, derived conditionally on the predictors in Section~\ref{sec:examples_derivations} of the supplement, differ between the two clusters, which are not observed.

\paragraph*{Example 2: Panel data forecasting}
The trade-off between the bias of pooled estimators and the variance of heterogeneous ones is central to panel forecasting \citep{pesaran26}, and latent group structures are now common in panel data analysis \citep{bonhomme15a,ando2017clustering,lumsdaine2023estimation}. Let $Y_{it}=\alpha_i+\beta'_{k_i}X_{i,t-1}+U_{it}$, $U_{it}\sim iid(0,\sigma^2)$, $k_i\in\{1,\dots,K\}$, with fixed predictors and errors independent of the regressors and estimators. Forecaster 1 predicts $\widehat Y^{\text{pooled}}_{i,R+1}=\hat\alpha_i+\hat\beta'X_{i,R}$, with $\hat\beta$ the pooled fixed effects (within) estimator, and Forecaster 2 predicts $\widehat Y^{\text{het}}_{i,R+1}=\tilde\alpha_i+\hat\beta_i'X_{i,R}$, with unit-by-unit least squares estimates, all computed over the estimation window. Pooling is biased for the units with $\beta_{k_i}$ different from the probability limit $\beta$ of $\hat\beta$, while unit-by-unit estimation is unbiased but noisier. The loss differentials (Section~\ref{sec:examples_derivations} of the supplement) therefore differ across units according to how much they gain from pooling: the clusters are defined by relative performance rather than by observable labels, and must be estimated from the loss differentials.

\paragraph*{Example 3: Forecasting with machine learning methods}
Machine learning methods are increasingly used in economics \citep{athey2018impact,haghighi2025machine}, and \cite{goulet2022machine} find that their main advantage in macroeconomic forecasting is to capture nonlinearities associated with uncertainty, financial stress and housing bubbles. When a linear method, such as the Lasso, is compared with a nonlinear one, such as random forests, and only some units have nonlinear dynamics, the average loss differential over the panel can be misleading, and clustering the loss differentials before testing C-EPA identifies the groups of units for which each method dominates.

\subsection{Generalized clustered EPA tests with predetermined clusters}\label{sec:known}

The tests for the unconditional C-EPA hypothesis have been developed by APUY and \citet{qu23b} for predetermined clusters. With $\mathcal{F}_{t} = \{ \emptyset, \mathit{\Omega} \}$, the C-EPA null reduces to $|\mathcal{C}_k|^{-1} \sum_{i \in \mathcal{C}_k} \mathbb{E}(\Delta L_{it}) = 0$, $k=1,\dots,K$. APUY suggested several test statistics under different assumptions on the dependence structure of the loss differentials; here we generalize their methodology to $\mathcal{F}_{t} \neq \{ \emptyset, \mathit{\Omega} \}$ and add a small sample adjustment.

Consider the following test statistic for \eqref{eq:h0rewritten}:
\begin{equation}\label{eq:f_con}
W(\mathcal{C}) = \frac{B-KP+1}{KPB} T \hat{\theta}'(\mathcal{C}) \widehat{\Omega}^{-1}(\mathcal{C}) \hat{\theta}(\mathcal{C}),
\end{equation}
where $\hat{\theta}(\mathcal{C}) = [\hat{\theta}_1'(\mathcal{C}),\dots,\hat{\theta}'_K(\mathcal{C})]'$ with $\hat{\theta}_k(\mathcal{C}) = (|\mathcal{C}_k|T)^{-1} \sum_{i \in \mathcal{C}_k} \sum_{t=1}^T Z_{it}$ and
$\widehat{{\Omega}}(\mathcal{C})$ is an orthonormal series (OS) variance-covariance estimator defined as follows
\begin{align}\label{eq:var}
\widehat{\Omega}(\mathcal{C})
&=
\frac{1}{B}
\sum_{j=1}^B
\widehat{\Lambda}_{j}(\mathcal{C})
\widehat{\Lambda}_{j}'(\mathcal{C}),  \\
\widehat{\Lambda}_{j}(\mathcal{C})
&=
\sqrt{\frac{2}{T}}
\sum_{t=1}^T
[
\bar{Z}_{t}(\mathcal{C})
-
\hat{\theta}(\mathcal{C})
]
\cos
\left[
\pi j
\left(
\frac{t-1/2}{T}
\right)
\right],
\end{align}
with $\bar{Z}_{t}(\mathcal{C}) = [\bar{Z}'_{1,t}(\mathcal{C}),\dots,\bar{Z}'_{K,t}(\mathcal{C})]'$, $\bar{Z}_{k,t}(\mathcal{C}) = |\mathcal{C}_k|^{-1} \sum_{i \in \mathcal{C}_k} Z_{it}$, and $B$ is the number of orthonormal basis functions used in its estimation. 
The first factor in \eqref{eq:f_con}, $(B-KP+1)/(KPB)$, is a small sample correction obtained through the connection between Hotelling's $T^2$ distribution and the $F$-distribution, and using the asymptotic property of the proposed variance estimator.

Orthonormal series (OS) estimators of a long-run variance (LRV) were proposed by \cite{phillips05} and developed by \cite{muller07} and \cite{sun11,sun13,sun14a}, among others. For fixed $B$, Lemma~\ref{lemma:clt} and the arguments of \cite{sun13} give $W(\mathcal{C}) \overset{d}{\longrightarrow} \mathbb{F}_{KP,B-KP+1}$ under the null, where $\mathbb{F}_{v_1,v_2}$ denotes the $F$-distribution with numerator and denominator degrees of freedom of $v_1$ and $v_2$, respectively. When $B \longrightarrow \infty$, the usual results of APUY generalize to $KPB(B-KP+1)^{-1}W(\mathcal{C}) \overset{d}{\longrightarrow} \chi^2_{\scriptscriptstyle KP}$. When $B$ is not too large, the $\mathbb{F}_{KP,B-KP+1}$ critical values yield better size than the (scaled) $\chi^2_{\scriptscriptstyle KP}$ ones \citep{sun13}.

With some abuse of notation, let $p [ w(\mathcal{C}) ] = \Pr \left[\, \mathbb{F}_{KP,B-KP+1} \geq w(\mathcal{C}) \,\right]$ be the $p$-value associated with $w(\mathcal{C})$. Here and below, the reference distribution of $p[\ \cdot\ ]$ is suppressed when it is clear from the context. A level-$\alpha$ test rejects the null if $p [ w(\mathcal{C}) ] \leq \alpha$, where $\alpha \in (0,1)$ is the predetermined Type I error rate.

\subsection{Panel Kmeans estimator}

If there is no a priori information on the clusters $\mathcal{C}_k$, $k = 1,\dots,K$, one may use the Panel Kmeans estimator applied to the stacked panel $Z = (Z'_{11},Z'_{12},\dots,Z_{NT}')'$, denoted $\mathcal{C}(Z)$, to learn these clusters from the data.
For a given $K$, the Panel Kmeans estimators of the cluster membership sets and cluster centers are defined respectively as:
\begin{align}\label{eq:kmeans}
(\widehat{\mathcal{C}}_1,\dots,\widehat{\mathcal{C}}_K)
&=
\argmin_{({\mathcal{C}}_1,\dots,{\mathcal{C}}_K)}
\sum_{k=1}^K
\sum_{i \in \mathcal{C}_k}
\sum_{t=1}^T
\left\lVert
Z_{it}
-
\frac{1}{|\mathcal{C}_k|T}
\sum_{j \in \mathcal{C}_k}
\sum_{s=1}^T
Z_{js}
\right\rVert^2
\\
\hat{\theta}_{k}(\widehat{\mathcal{C}})
&=
\frac{1}{|\widehat{\mathcal{C}}_k|T}
\sum_{i \in \widehat{\mathcal{C}}_k}
\sum_{t=1}^T
Z_{it}.
\end{align}
Problem \eqref{eq:kmeans} is solved by Algorithm~\ref{algo} in Section~\ref{sec:supp_kmeans} of the supplement, a version of Lloyd's algorithm \citep{lloyd82} that is initialized, as in \cite{chen23}, with the time averages of $K$ units drawn at random, independently of the data. With time-invariant cluster centers, its assignment and center-update steps are those of Lloyd's algorithm applied to the unit-specific time averages $\bar Z_1,\dots,\bar Z_N$ (Lemma~\ref{lemma:average_equivalence} in the supplement).

We denote by $M$ the number of iterations of Algorithm~\ref{algo} until convergence, which depends on the initialization, and by $\hat k_i(Z)=k_i^{(M)}(Z)$ the cluster assigned to unit $i$ at convergence, so that $\widehat{\mathcal{C}}_k=\{i:\hat k_i(Z)=k\}$. Since the criterion in \eqref{eq:kmeans} is non-convex, Algorithm~\ref{algo} is run in practice from $S$ initializations drawn independently of the data, and the partition with the smallest value of the criterion is retained, with ties broken in favor of the smallest index; when the number of clusters is selected by the information criterion of Section~\ref{sec:bic}, this is done for every $K$ in a range $\mathcal{K}=\{2,\dots,K_{max}\}$. Initializations that produce an empty cluster are discarded. We call the selected number of clusters and the selected partition, $\hat s(Z)=(\widehat K,\widehat{\mathcal C})$, computed with the initializations held fixed, the selection outcome. By Lemma~\ref{lemma:average_equivalence} in the supplement, $\hat s(Z)$ depends on the data only through the time averages $\bar Z_i=T^{-1}\sum_{t=1}^TZ_{it}$, $i=1,\dots,N$, and, through the information criterion, through the within-unit deviations $Z_{it}-\bar Z_i$.

\section{Assumptions and a Useful Lemma}\label{sec:naive}

We now present the assumptions and a preliminary result that will be instrumental in developing the test statistics of Section \ref{sec:tests} and establishing their asymptotic properties. This requires some additional notation. 
Throughout, $C$ denotes a generic positive constant, the number of units $N$ is fixed, and limits are taken as $T\to\infty$. Define $V_{it} = Z_{it} - \mu_i^0$, where $V_{p,it}$, for $p = 1, \dots, P$, denotes the $p$th element of the vector $V_{it}$ and the demeaned time average
$
\bar{V}_i = \frac{1}{T}\sum_{t=1}^T V_{it} \in \mathbb{R}^P
$,
and let $\bar{V} = (\bar{V}_1',\dots,\bar{V}_N')' \in \mathbb{R}^{NP}$ denote the stacked vector of demeaned time averages across all units with
$\Xi=\lim_{T\to\infty}\mathrm{Var}(\sqrt T\bar V)\in\mathbb R^{NP\times NP}$ the long-run covariance matrix of the innovations (Assumption~\ref{ass:clt}). For any partition $\mathcal C=(\mathcal C_1,\dots,\mathcal C_K)$ and any pair $k\neq g$, define $\delta^N_{k,g}\in\mathbb R^N$ by $(\delta^N_{k,g})_i =
\mathbf 1\{i\in\mathcal C_k\} / |\mathcal C_k| - \mathbf 1\{i\in\mathcal C_g\} / |\mathcal C_g|$ and let $\Pi^N_{k,g} = I_N - \delta^N_{k,g}(\delta^N_{k,g})'
/ \lVert\delta^N_{k,g}\rVert^2 $, and let $\Sigma_{k,g}(\mathcal C)=(\delta^N_{k,g}\otimes I_P)'\Xi(\delta^N_{k,g}\otimes I_P)$ be the long-run covariance matrix of the pairwise difference of the cluster means, whose estimator $\widehat\Sigma_{k,g}(\mathcal C)$ is defined in Section~\ref{sec:selective}.
Let $\mathfrak C_{N,K}$ denote the finite set of partitions of $\{1,\dots,N\}$ into $K$ non-empty clusters, and let $V_t=(V_{1t}',\dots,V_{Nt}')'\in\mathbb R^{NP}$.

The following assumptions will be referred to throughout the paper. They are grouped into two classes: the generic ones are labeled as G$\#$, which are required for both size and power properties of the tests, and the other two are specific assumptions, labeled as S$\#$, required only for the consistency of the information criterion of Section~\ref{sec:bic}.

\begin{assumptionG}\normalfont\label{ass:alternative}
\begin{enumerate*}[label=(\alph*)]
  \item\label{ass:compact} $\lVert \mu_i^0 \rVert \leq C$,
  \item\label{ass:meanzero} $\mathbb{E} \lVert V_{it} \rVert^4 \leq C$,
  \item\label{ass:weakdepts} $T^{-1} \sum_{t,s=1}^T \mathbb{E} \lVert V_{it} V'_{is} \rVert \leq C$.
\end{enumerate*}
\end{assumptionG}

\begin{assumptionG}\normalfont\label{ass:clusternumbers}
  $N$ is fixed and $|\mathcal{C}_k|\geq1$ for each $k=1,\dots,K$.
\end{assumptionG}

\begin{assumptionG}\normalfont\label{ass:clt}
As $T\to\infty$, $T^{-1/2}\sum_{t=1}^{\lfloor rT\rfloor}V_t\Rightarrow\Xi^{1/2}\mathbb W_{NP}(r)$ on $[0,1]$, where $\mathbb W_{NP}$ is an $NP$-dimensional standard Brownian motion and $\Xi=\lim_{T\to\infty}\mathrm{Var}(\sqrt T\bar V)$ is positive definite.
Moreover, as $B\to\infty$ and $B/T\to0$,
\[
\max_{\mathcal C\in\mathfrak C_{N,K}}
\max_{k\neq g}
\left\|
\Sigma_{k,g}(\mathcal C)^{-1/2}
\widehat\Sigma_{k,g}(\mathcal C)
\Sigma_{k,g}(\mathcal C)^{-1/2}
-
I_P
\right\|
=
o_p(1)
\]
and $\max_{\mathcal C\in\mathfrak C_{N,K}}\max_{k\neq g}\max_{1\leq i\leq N}\lVert[(\widehat\Xi-\Xi)(\delta^N_{k,g}\otimes I_P)]_i\rVert=o_p(1)$, where $\widehat\Xi$ is the OS estimator of $\Xi$ defined in Section~\ref{sec:selective} and $[\cdot]_i$ denotes the $i$th block of $P$ rows.
\end{assumptionG}

Assumptions~\ref{ass:alternative}\ref{ass:compact}--\ref{ass:meanzero} ensure that the cluster centers and their variances are finite and can be estimated consistently, and Assumption~\ref{ass:alternative}\ref{ass:weakdepts} limits the time series dependence, $\sum_{t\neq s}\mathbb E\|V_{it}V_{is}'\| = O(T)$. No restriction is placed on the CSD (see the discussion following Lemma~\ref{lemma:clt}). 
Assumption~\ref{ass:clusternumbers} states that the asymptotic analysis is for a fixed number of units, the setting of the selective results of Section~\ref{sec:selective}, and that the clusters are non-empty. Assumption~\ref{ass:clt} is a functional central limit theorem for the $NP$-dimensional vector of innovations; it gives the Gaussian limit of $\sqrt T\bar V$ and the limits of the OS variance estimators with fixed $B$. It accommodates general CSD, including strong factor dependence, since $\Xi$ is unrestricted. The consistency conditions are imposed over all partitions because the partition used in the selective test is data dependent. No restriction is placed on $\Xi$ beyond positive definiteness.
 In particular, $\delta^N_{k,g}$ need not be an eigenvector of $\Xi$, which, for $P=1$, is what $(\delta^N_{k,g}\otimes I_P)'\Xi(\Pi^N_{k,g}\otimes I_P)=0$ requires, and which would be needed for the Euclidean residual used by \cite{chen23} to be independent of the tested pairwise difference. This condition holds under compound symmetry \citep[see][]{gonzalez2023post} but fails, for instance, when units load heterogeneously on common factors or when long-run variances differ across clusters, and it is no longer needed because the decomposition in Section~\ref{sec:selective} is taken in the metric of $\Xi$. 
Assumption~\ref{ass:clt} holds under the mixing and moment conditions of Section~\ref{sec:hdclt_primitives} of the supplement.

We now state a useful lemma for the theoretical analysis of the test statistics developed below, which requires some new notation.
Let $\theta^{0}_{k}(\mathcal{C}) = |\mathcal{C}_k|^{-1} \sum_{i \in \mathcal{C}_k} \mu_i^0$ be the true center of the $k$-th cluster implied by the partition $\mathcal{C}$.
Define the $KP \times 1$ vector of true cluster centers: $\theta^0(\mathcal{C}) = [\theta^{0\prime}_{1}(\mathcal{C}),\dots,\theta^{0\prime}_{K}(\mathcal{C})]'$ and let \( \Omega(\mathcal{C}) \in \mathbb{R}^{KP \times KP} \) denote the variance-covariance matrix of the vector \( \hat{\theta}(\mathcal{C}) \) after scaling, that is,
$
\Omega(\mathcal{C}) = \mathbb{V} \{ \sqrt{T} [\hat{\theta}(\mathcal{C}) - \theta^0(\mathcal{C})]\}.
$
The following result gives the usual properties of the sample means for a fixed $\mathcal{C}$; it is useful because the selective tests condition on the estimated clusters.
\begin{lemma}\normalfont\label{lemma:clt}
Let $\mathcal{C}$ be a fixed partition of $\{1,\dots,N\}$ into $K$ clusters.
Under Assumptions~\ref{ass:alternative}--\ref{ass:clt}, the following results hold as $T \to \infty$:
\begin{enumerate}[label=(\alph*)]
    \item\label{lemma:clt:lln} $\hat{\theta}_k(\mathcal{C}) - \theta^0_k(\mathcal{C}) = o_p(1)$ for each $k \in \{1,\dots,K\}$.
        \item\label{lemma:clt:clt} $
    \widehat{\Sigma}_{k,g}(\mathcal{C})^{-1/2}
    \sqrt{T}\,
    \bigl[\hat{\theta}_k(\mathcal{C}) - \hat{\theta}_g(\mathcal{C}) - \theta^0_k(\mathcal{C}) + \theta^0_g(\mathcal{C})\bigr]
    \overset{d}{\longrightarrow}
    \mathcal{N}(0, I_P)
    $.
    \item\label{lemma:clt:joint} For any fixed $q\times KP$ matrix $\mathsf R$ of full row rank, $\sqrt T\,[\mathsf R\Omega(\mathcal C)\mathsf R']^{-1/2}\mathsf R[\hat\theta(\mathcal C)-\theta^0(\mathcal C)]\overset{d}{\longrightarrow}\mathcal N(0,I_q)$.
\end{enumerate}
\end{lemma}

Part~\ref{lemma:clt:lln} is a law of large numbers showing that Assumptions~\ref{ass:alternative}--\ref{ass:clt} suffice for the consistency of the sample cluster centers under any fixed partition $\mathcal{C}$, and Part~\ref{lemma:clt:clt} is the corresponding CLT for pairwise differences. Part~\ref{lemma:clt:joint} is the joint CLT used for the Wald statistics with fixed partitions, such as $W(\mathcal C)$ in Section~\ref{sec:known}. Parts~\ref{lemma:clt:clt} and~\ref{lemma:clt:joint} follow from Assumption~\ref{ass:clt} because $\sqrt{T}[\hat{\theta}_k(\mathcal{C})-\hat{\theta}_g(\mathcal{C})-\theta^0_k(\mathcal{C})+\theta^0_g(\mathcal{C})] = (\delta^N_{k,g}\otimes I_P)' \sqrt{T}\bar V$ is a linear functional of $\sqrt T\bar V$, and the positive definiteness of $\Xi$
ensures that $\Sigma_{k,g}(\mathcal C)$ is positive definite. CSD of any form, including strong factor dependence, is accommodated because $\Xi$ need not be block diagonal.

\begin{remark}\normalfont
Lemma~\ref{lemma:clt} applies to a fixed partition $\mathcal{C}$ and does not directly give the asymptotic distribution of cluster mean estimators when $\widehat{\mathcal{C}}$ is estimated from the data. Conditioning on the realized partition does not by itself make Lemma~\ref{lemma:clt} applicable, because the event $\{\widehat{\mathcal{C}}=\mathcal{C}\}$ depends on the same data. Section~\ref{sec:selective} therefore characterizes the conditional distribution exactly in the Gaussian limit experiment and conditions only on selection events of positive probability.
\end{remark}

The following assumptions are used only for the consistency of the information criterion of Section~\ref{sec:bic} (Lemma~\ref{lemma:ic_consistent} in the supplement).

\begin{assumptionS}\normalfont\label{ass:clustercenters} $\mu_i^0 =\theta^0_{k}$ for all $i \in \mathcal{C}^0_k$ and $k = 1,\dots,K^0$, where $\theta^0_{k}$ is the true cluster center of the $k$th cluster and $\mathcal{C}^0_k$ is the set of units belonging to the true $k$th cluster.
\end{assumptionS}

\begin{assumptionS}\normalfont\label{ass:separation} Let $K^0 \geq 2$. Then for all $k,g \in \{1,\dots,K^0\}$, $k \neq g$, there exists $C_{k,g} > 0$ such that $\lVert \theta^0_{k} - \theta^0_{g} \rVert^2 \geq C_{k,g}$.
\end{assumptionS}

Assumption~\ref{ass:clustercenters} states that the units are homogeneous within clusters and Assumption~\ref{ass:separation} that the clusters are separated. Together they imply that C-EPA fails, but they are not necessary for power: even if $K^0 = 1$, the tests have power when the overall mean differs from zero.

Based on Lemma~\ref{lemma:clt}, a naive attempt to test the C-EPA null would estimate the unknown clusters using the Panel Kmeans estimator and then use these estimates to construct a Wald test statistic. The test that rejects if $p [ w(\widehat{\mathcal{C}}) ] \leq \alpha$, with $W(\widehat{\mathcal{C}})$ the Wald statistic computed with the Panel Kmeans clusters, uses the same data to estimate the clusters and to test the null. Testing homogeneity after Kmeans or hierarchical clustering is known to be severely anti-conservative \citep{gao24,patton23,chen23}: the algorithm separates the cluster means even when the null holds. Since homogeneity is a sub-hypothesis of our null (Section~\ref{sec:tests}), the naive C-EPA tests suffer from the same problem (Section~\ref{sec:monte}).

\section{Tests of Generalized C-EPA with Unknown Clusters}\label{sec:tests}
This section develops valid tests of the C-EPA null when the clusters are estimated by Panel Kmeans, using selective conditional inference; a sample splitting approach is given in Section~\ref{sec:split} of the supplement. We first break down the C-EPA hypothesis into its sub-hypotheses of homogeneity and O-EPA. The implication \eqref{eq:h0rewritten} of the null of interest \eqref{eq:null} can be written as $\mathcal{H}'_{0}: \mathcal{H}^{homo}_{0} \bigcap \mathcal{H}^{oepa}_{0}$,
with
\begin{equation}\label{eq:homonull}
\mathcal{H}^{homo}_{0}: \{ \theta^0_k(\mathcal{C}) = \theta^0_{g}(\mathcal{C}) \} \text{ for all } k, g \in \{1, \dots, K\}, \; k \neq g,
\end{equation}
being the homogeneity hypothesis and
\begin{equation}\label{eq:oepanull}
\mathcal{H}^{oepa}_{0}: \frac{1}{N} \sum_{k=1}^K |\mathcal{C}_k| \theta^0_k(\mathcal{C}) = 0,
\end{equation}
the O-EPA hypothesis of APUY, in which the overall predictive performance is a weighted average of the cluster means. The parameter of the O-EPA hypothesis does not depend on the partition. Both hypotheses are of empirical interest: O-EPA is studied by APUY with known clusters, and the relevance of testing homogeneity goes beyond EPA testing \citep[see the applications of][]{patton23}. Section~\ref{sec:selective} develops selective tests of the homogeneity of a pair of selected clusters and of $\mathcal{H}^{homo}_{0}$, and Section~\ref{sec:maintests} presents the O-EPA test and the test of $\mathcal{H}_{0}$.

\subsection{Testing the null of homogeneity}\label{sec:selective}

We develop a test for \eqref{eq:homonull}: we first develop tests of each pairwise equality sub-hypothesis and their properties, then a homogeneity test that merges their $p$-values.

{\em Testing pairwise equality.}\; The homogeneity null $\mathcal{H}^{homo}_{0}$ is the intersection of $K(K-1)/2$ unique pairwise equality hypotheses. For each of these nulls, we define the test statistic $D_{k,g}(\widehat{\mathcal{C}})$ as the square root of the associated Wald test statistic:
\begin{equation}
D^2_{k,g}(\widehat{\mathcal{C}}) = T [ \hat{\theta}_{k}(\widehat{\mathcal{C}}) - \hat{\theta}_{g}(\widehat{\mathcal{C}}) ]' \widehat{\Sigma}^{-1}_{k,g}(\widehat{\mathcal{C}}) [ \hat{\theta}_{k}(\widehat{\mathcal{C}}) - \hat{\theta}_{g}(\widehat{\mathcal{C}}) ],
\end{equation}
where
$\widehat{\Sigma}_{k,g}(\widehat{\mathcal{C}}) = \widehat{\omega}_{k,k}(\widehat{\mathcal{C}}) + \widehat{\omega}_{g,g}(\widehat{\mathcal{C}}) - \widehat{\omega}_{k,g}(\widehat{\mathcal{C}}) - \widehat{\omega}_{g,k}(\widehat{\mathcal{C}})
$
with $\widehat{\omega}_{k,g}(\widehat{\mathcal{C}})$ being the $\{k,g\}$th $P \times P$ block of $\widehat{{\Omega}}(\widehat{\mathcal{C}})$. For a fixed partition $\mathcal C$ for which $\theta^0_k(\mathcal C)=\theta^0_g(\mathcal C)$, Lemma~\ref{lemma:clt} gives $D_{k,g}(\mathcal{C}) \overset{d}{\longrightarrow} \chi_{\scriptscriptstyle P}$ as $T \longrightarrow \infty$, where $\chi_{\scriptscriptstyle P}$ is a $\chi$ variable with $P$ degrees of freedom. As discussed in the previous section, however, the associated critical values lose their validity when used with estimated clusters.

The following definition of asymptotic selective Type I error control is the basis for valid C-EPA testing with unknown clusters.
\begin{definition}\normalfont\label{definition:selectivepvalue}
For a pair of clusters $k, g \in \{1, \dots, K\}$, $k \neq g$, a test of
$
\mathcal{H}^{k,g}_0 : \{ \theta^0_k(\mathcal{C}) = \theta^0_{g}(\mathcal{C}) \}
$
controls the selective Type I error rate asymptotically as $T \to \infty$
at level $\alpha \in (0,1)$ if
\begin{equation}\label{eq:selective_pvalue_def}
    \limsup_{T \to \infty} \Pr_{\mathcal{H}^{k,g}_0} \left[ \text{Reject } \mathcal{H}^{k,g}_0 \text{ at level } \alpha \;\middle|\; \hat s(Z)=\hat s(z)
 \right] \leq \alpha,
\end{equation}
where $\hat s(Z)$ is the selection outcome defined in Section~\ref{sec:preliminaries}, that is, the selected number of clusters and the selected partition, computed with the initializations of Algorithm~\ref{algo} held fixed, and $\hat s(z)$ is its realization in the observed sample. Since the initializations are drawn independently of the data, a test that satisfies \eqref{eq:selective_pvalue_def} for every set of initializations satisfies it uniformly over them, and also unconditionally on them.
\end{definition}
A valid test of $\mathcal{H}^{k,g}_0$ thus rejects only if the probability of a statistic at least as large as the observed one, among the realizations of $Z$ that yield the same selection outcome as $z$, does not exceed $\alpha$.

As noted by \cite{chen23}, characterizing this condition is not trivial. They condition instead on the clusters estimated at every step of the algorithm, which yields a smaller event that is simpler to characterize. We condition on the selection outcome itself, which can be computed exactly along the one-dimensional perturbation path introduced below (Section~\ref{sec:calc_of_p} of the supplement); conditioning in addition on every step of the algorithm, and on every initialization, is also valid but discards information and lowers power. Two further terms to condition on emerge from decomposing the time averages into a term associated with $D_{k,g}(\mathcal{C})$ and a residual that is uncorrelated with it in the metric of the long-run covariance matrix. Let $\bar Z=(\bar Z_1',\dots,\bar Z_N')'$ and $L_{k,g}=\delta^N_{k,g}\otimes I_P$, with $\delta^N_{k,g}$ the pairwise difference weight vector of Section~\ref{sec:naive} for the partition $\mathcal{C}$, so that $\sqrt TL_{k,g}'\bar Z=\sqrt T[\hat\theta_k(\mathcal C)-\hat\theta_g(\mathcal C)]$ and $D_{k,g}(\mathcal C)=\lVert\widehat\Sigma^{-1/2}_{k,g}(\mathcal C)\sqrt TL_{k,g}'\bar Z\rVert$. Then
\begin{equation}\label{eq:zrewritten_inthetext}
\begin{aligned}
\sqrt T\bar Z &= \widehat U_{k,g} + D_{k,g}(\mathcal{C})\, \widehat\Xi L_{k,g}\widehat{\Sigma}^{-1/2}_{k,g}(\mathcal{C})\, \mathrm{dir}\left[\widehat{\Sigma}^{-1/2}_{k,g}(\mathcal{C})\sqrt TL_{k,g}'\bar Z\right],\\
\widehat U_{k,g} &= \sqrt T\bar Z - \widehat\Xi L_{k,g}\widehat{\Sigma}^{-1}_{k,g}(\mathcal{C})\sqrt TL_{k,g}'\bar Z,
\end{aligned}
\end{equation}
where $\widehat\Xi=B^{-1}\sum_{j=1}^B\widehat\Lambda_j\widehat\Lambda_j'$ with $\widehat\Lambda_j=\sqrt{2/T}\sum_{t=1}^T(Z_t-\bar Z)\cos[\pi j(t-1/2)/T]$ and $Z_t=(Z_{1t}',\dots,Z_{Nt}')'$, each unit being demeaned by its own time average. This equality holds at every realization and is derived in Equation~\eqref{eq:zrewritten} of Section~\ref{sec:proofs} of the supplement. The residual $\widehat U_{k,g}$ is the regression residual of $\sqrt T\bar Z$ on the tested difference with coefficient $\widehat\Xi L_{k,g}\widehat{\Sigma}^{-1}_{k,g}(\mathcal{C})$, and in the Gaussian limit it is independent of $D_{k,g}(\mathcal C)$ and of the direction for any positive definite $\Xi$ (Lemma~\ref{lemma:independence}). The construction of \cite{chen23}, adapted to panels, uses instead the Euclidean projection $\Pi_{k,g}=I-\hat\nu_{k,g}\hat\nu_{k,g}'/\lVert\hat\nu_{k,g}\rVert^2$ onto the subspace orthogonal to $\hat\nu_{k,g}=(\hat\nu'_{k,g,1},\dots,\hat\nu'_{k,g,N})'$, with $\hat\nu_{k,g,i}=\iota_T\hat\delta_{k,g,i}$, $\iota_T$ a $T\times1$ vector of ones and $\hat{\delta}_{k,g,i} = \mathbf{1} \{ \hat{k}_i = k \}/|\widehat{\mathcal{C}}_k| - \mathbf{1} \{ \hat{k}_i = g \}/|\widehat{\mathcal{C}}_g|$. The two residuals coincide only when the column space of $L_{k,g}$ is invariant under $\Xi$, which for $P=1$ requires $\delta^N_{k,g}$ to be an eigenvector of $\Xi$. The only additional estimates required are the $NP\times P$ long-run covariances $\widehat\Xi L_{k,g}=B^{-1}\sum_{j=1}^B\widehat\Lambda_j(L_{k,g}'\widehat\Lambda_j)'$ between the unit-level series and the tested difference; the $L_{k,g}'\widehat\Lambda_j$ are the cosine transforms already used in $\widehat\Sigma_{k,g}(\mathcal C)$, so the $NP\times NP$ matrix $\widehat\Xi$ is never formed. Any consistent estimator of $\Xi L_{k,g}$ preserves the asymptotic results. When $B$ is small relative to $NP$, $\widehat\Xi L_{k,g}$ is noisy, and we also consider the shrinkage estimator $(1-\hat\omega)\widehat\Xi L_{k,g}+\hat\omega L_{k,g}(L_{k,g}'L_{k,g})^{-1}\widehat\Sigma_{k,g}(\mathcal C)$, whose target is the covariance implied by the Euclidean decomposition and whose intensity $\hat\omega$ is the ratio of the estimated variance of $\widehat\Xi L_{k,g}$ to its squared distance from the target, truncated to $[0,1]$, as in \cite{ledoit2004well}. Since this variance is of order $B^{-1}$, $\hat\omega$ converges to zero unless the target is correct, in which case both estimators are consistent. The shrinkage estimator substantially increases the power of the selective tests in finite samples (Section~\ref{sec:monte}). Since Panel Kmeans depends on the data only through the time averages and the within-unit deviations (Lemma~\ref{lemma:average_equivalence} in the supplement), the decomposition translates into the panel by shifting all observations of unit $i$ by the same amount, and the conditional distribution of $D_{k,g}(\widehat{\mathcal{C}})$ given $\hat s$ is based on it.

Let $G\sim\mathcal N(0,\Xi)$, and denote by $D^G_{k,g}$, $U^G_{k,g}$ and $\hat s^G$ the counterparts of $D_{k,g}(\widehat{\mathcal C})$, $\widehat U_{k,g}$ and $\hat s$ in the Gaussian limit experiment, in which $\sqrt T(\bar Z-\mu^0)$ is replaced by $G$ and $\widehat\Xi$ by $\Xi$. We define the asymptotic $p$-value as the conditional tail probability in this experiment,
\begin{equation}\label{eq:pairwise_p}
p_{\infty} [ d_{k,g}(\widehat{\mathcal{C}}) ] = \Pr_{\mathcal{H}^{k,g}_0} \left[ D^G_{k,g} \geq d_{k,g}(\widehat{\mathcal{C}}) \;\middle|\; \mathcal{A}^G \right],
\end{equation}
for $k,g \in \{1,\dots,K\}$, where
\begin{equation}\label{eq:conditions}
\mathcal{A}^G
=
\left\lbrace
\hat s^G = \hat s(z),
\,
U^G_{k,g} = \hat u_{k,g},
\,
\mathrm{dir}\left[\Sigma^{-1/2}_{k,g}(\widehat{\mathcal{C}})L_{k,g}'(\sqrt T\mu^0+G)\right]
=
\mathrm{dir}\left[\widehat{\sigma}^{-1/2}_{k,g}(\widehat{\mathcal{C}})\sqrt TL_{k,g}'\bar z\right]
\right\rbrace,
\end{equation}
with $\hat u_{k,g}$ and $\widehat{\sigma}_{k,g}(\mathcal{C})$ being the realizations of $\widehat U_{k,g}$ and $\widehat{\Sigma}_{k,g}(\widehat{\mathcal{C}})$ associated with the realization $z$ of $Z$. Conditioning on $U^G_{k,g}$ and on the direction, which are events of probability zero, is legitimate in the Gaussian experiment, where the three components are independent. In finite samples, we do not condition on such events: joint weak convergence of a statistic and of nuisance components with independent limits does not imply convergence of the conditional distributions given the nuisance components. Proposition~\ref{proposition:pairwise_p} below therefore establishes validity conditional on the selection outcome only, which has positive probability.

The first condition in \eqref{eq:conditions} is the selection event of Definition~\ref{definition:selectivepvalue}: the selection outcome in the limit experiment, that is, the number of clusters and the cluster assigned to each unit $i$, must match the one obtained under the realization $z$. The other two conditions remove the nuisance parameters $U^G_{k,g}$ and the direction of the tested difference in \eqref{eq:zrewritten_inthetext}, without which the conditional distribution of $D_{k,g}(\widehat{\mathcal{C}})$ given $\widehat{\mathcal{C}}$ is not tractable.
These conditions are standard in the selective inference literature \citep[see][]{chen23}.

The asymptotic $p$-value $p_{\infty} [ d_{k,g}(\widehat{\mathcal{C}}) ]$ builds on \cite{chen23} but generalizes it in several ways. Here the random variables $Z_{it}$ are double-indexed, $i=1,\dots,N$, $t=1,\dots,T$. \cite{chen23} allow dependence only across the elements of the same observation, between $Z_{p,it}$ and $Z_{c,it}$, and not between $Z_{it}$ and $Z_{js}$ for $i \neq j$ or $t \neq s$, whereas we allow arbitrary autocorrelation and CSD as well as dependence between the elements of $Z_{it}$. Their method also relies on the normality of the data and on a covariance structure under which the Euclidean residual is independent of the tested difference, whereas we allow an arbitrary positive definite long-run covariance matrix through the decomposition \eqref{eq:zrewritten_inthetext} and exploit the time series dimension through the CLT of Lemma~\ref{lemma:clt}.

The following lemma shows how to calculate a $p$-value in observed samples following this definition.

\begin{lemma}\normalfont\label{lemma:perturbation}
    Let $k,g \in \{1,\dots,K\}$ with $k\neq g$ and $K \geq 2$ given, and $B \to \infty$ as $T \to \infty$
such that $B/T \to 0$.
    Under Assumptions~\ref{ass:alternative}--\ref{ass:clt} and $\mathcal{H}^{k,g}_0$, a $p$-value following the asymptotic principle \eqref{eq:pairwise_p} can be calculated using
    \begin{equation}
        p[d_{k,g}(\widehat{\mathcal{C}})] = 1 - F_{\chi_{\scriptscriptstyle P}} [\, d_{k,g}(\widehat{\mathcal{C}});\mathcal{T} \,],
    \end{equation}
    where $F_{\chi_{\scriptscriptstyle P}}( \ \cdot \ ;\mathcal{T})$ denotes the cumulative distribution function of a $\chi_{\scriptscriptstyle P}$ random variable truncated to the set $\mathcal{T}$ with
    \begin{equation}\label{eq:set_s}
        \mathcal{T} = \left\lbrace \phi \in \mathbb{R}_{\geq 0} : {\hat s[z(\phi)] = \hat s(z)} \right\rbrace,
    \end{equation}
    and
    \begin{equation}\label{eq:perturbed}
[z(\phi)]_{it} = z_{it}-\bar z_i+\frac{1}{\sqrt T}\left[\hat u_{k,g}+\phi\,\hat\xi L_{k,g}\widehat{\sigma}^{-1/2}_{k,g}(\widehat{\mathcal{C}})\,\mathrm{dir}\left(\widehat{\sigma}^{-1/2}_{k,g}(\widehat{\mathcal{C}})\sqrt TL_{k,g}'\bar z\right)\right]_i,
    \end{equation}
    where $[\cdot]_i$ is the $i$th block of $P$ elements, $\hat\xi$ is the realization of $\widehat\Xi$, and $L_{k,g}$ is computed from $\widehat{\mathcal C}$.
\end{lemma}

Equation \eqref{eq:perturbed} defines a perturbation $z(\phi)$ of the data $z$ in which the clusters $k$ and $g$ are pushed together or pulled apart along the direction $\hat\xi L_{k,g}\widehat{\sigma}^{-1/2}_{k,g}(\widehat{\mathcal{C}})\mathrm{dir}(\cdot)$, in which unit $i$ moves by its own amount, reflecting the long-run covariance between its series and the tested difference: $z(\phi)=z$ when $\phi=d_{k,g}(\widehat{\mathcal{C}})$, the clusters are pulled apart when $\phi > d_{k,g}(\widehat{\mathcal{C}})$ and pushed together when $\phi < d_{k,g}(\widehat{\mathcal{C}})$, their centers coinciding in the extreme case $\phi = 0$. Thus $\phi$ measures the degree of perturbation \citep[see Figure~2 of][]{chen23}. The truncation sets $\mathcal{T}$ for Panel Kmeans are derived in Section~\ref{sec:calc_of_p} of the supplement.

For the asymptotic validity of the selective $p$-values, we consider $N$ fixed and expected loss differentials of the form $\mu^0_{i}=\mu^*_i+h_i/\sqrt T$. The component $\mu^*=(\mu^{*\prime}_1,\dots,\mu^{*\prime}_N)'$ describes the heterogeneity that persists in large samples, and $h=(h_1',\dots,h_N')'$ the local heterogeneity. The selection outcome evaluated at time averages $\mu^*+y/\sqrt T$ must settle down as $T$ grows.

\begin{assumptionG}\normalfont\label{ass:stability}
$N$ is fixed, {$\Xi$ is positive definite,} $Z_{it}=\mu^0_{i}+V_{it}$ with $\mu^0_i=\mu^*_i+h_i/\sqrt T$ and $\{V_t\}$ not depending on $T$, and there is a measurable map $\hat s_\infty$ such that, for Lebesgue-almost every $y\in\mathbb R^{NP}$ and all sequences $y_T\to y$, the selection outcome computed from the time averages $\mu^*+y_T/\sqrt T$, with within-unit second moments converging to their probability limit, equals $\hat s_\infty(y)$ for all $T$ large enough.
\end{assumptionG}

Section~\ref{sec:stabilization} of the supplement shows that Assumption~\ref{ass:stability} holds when $\mu^*=0$, which covers local alternatives and the J-EPA null, and when $\mu^*$ is a regular point, in the sense that the selection outcome is constant in a neighborhood of $\mu^*$, which covers expected loss differentials in general position, including population clusters with heterogeneity within clusters. It also holds for grouped expected loss differentials under a mild condition on the information criterion. The following result establishes the validity of $p[ d_{k,g}(\widehat{\mathcal{C}}) ]$ for the pairwise null $\mathcal{H}^{k,g}_0$ of Definition~\ref{definition:selectivepvalue}.

\begin{proposition}\normalfont\label{proposition:pairwise_p}
Let $k,g \in \{1,\dots,K\}$ with $k\neq g$ and $K \geq 2$ given, and $B \to \infty$ as $T \to \infty$ such that $B/T \to 0$.
\begin{enumerate}[label=(\alph*)]
    \item\label{proposition:pairwise_p:partb} Under Assumptions~\ref{ass:alternative}--\ref{ass:clt} and~\ref{ass:stability}, let $s$ be a selection outcome with $\Pr\{\hat s_\infty(h+G)=s\}>0$ such that $\mathcal{H}^{k,g}_0$ holds for the partition $\mathcal C(s)$ for every $T$. Then
    \begin{equation*}
    \lim\limits_{{T} \to \infty} \Pr \{\, p [ d_{k,g}(\widehat{\mathcal{C}}) ] \leq \alpha {\mid \hat s=s} \,\} = \alpha, \; \forall \alpha \in (0,1).
    \end{equation*}
    \item\label{proposition:pairwise_p:partc} Suppose now that
$K\geq2$ is fixed rather than selected, $N$ is fixed, $h=0$, and $\mu^*$ is a regular point of the selection outcome with $\hat s_\infty\equiv s$ (Section~\ref{sec:stabilization}). Under Assumptions~\ref{ass:alternative} and~\ref{ass:clt}, and if $\mathcal{H}^{k,g}_0$ fails for $\mathcal C(s)$,
    \begin{equation*}
    \lim_{T \to \infty} \Pr \{\, p [ d_{k,g}(\widehat{\mathcal{C}}) ] \leq \alpha \,\} = 1, \; \forall \alpha \in (0,1).
    \end{equation*}
\end{enumerate}
\end{proposition}

Part~\ref{proposition:pairwise_p:partb} states that, under the pairwise null and conditionally on the selection outcome, $p [ D_{k,g}(\widehat{\mathcal{C}}) ]$ is asymptotically a $p$-variable in the sense of \cite{vovk20}; following common practice, we call both this random variable and its realization $p [ d_{k,g}(\widehat{\mathcal{C}}) ]$ $p$-values. Part~\ref{proposition:pairwise_p:partc} shows that the test based on $D_{k,g}(\widehat{\mathcal{C}})$ is consistent whenever $\mathcal{H}^{k,g}_0$ fails, provided $\mu^*$ is a regular point (Remark~\ref{rem:fixedN}); $K$ need not equal the number of clusters in the population, and Section~\ref{sec:bic} proposes an information criterion to choose it.

\begin{remark}\normalfont\label{rem:fixedN}
Proposition~\ref{proposition:pairwise_p} is stated for fixed $N$. Extending it to $N\to\infty$ would require a high-dimensional Gaussian approximation over the selection events of Kmeans, which are defined by differences of quadratic forms in all $NP$ coordinates of $\sqrt T\bar Z$ and are not convex; we leave this for future work, and
 Section~\ref{sec:monte} examines panels with $N$ comparable to $T$. The result also holds along each sequence of the form $\mu^*+h/\sqrt T$ and for selection outcomes of positive limiting probability. Uniform validity over sets of expected loss differentials would require selection probabilities bounded away from zero, or randomization of the selection step \citep{tian2017asymptotics,tibshirani2018uniform,andrews2024inference}. Part~\ref{proposition:pairwise_p:partc} requires a regular point because the truncation set must contain an interval of length proportional to $\sqrt T$ below the observed statistic; with grouped expected loss differentials and a number of clusters that splits a true cluster, the selected partition is not constant near $\mu^*$ and this may fail, and power is then studied by simulation.
\end{remark}

\begin{remark}\normalfont\label{rem:joint}
For a $q\times KP$ matrix $\mathsf{R}$ of full row rank, the same construction with $L=A(\widehat{\mathcal C})\mathsf R'$ in place of $L_{k,g}$, where $A(\mathcal C)=[a_1\otimes I_P,\dots,a_K\otimes I_P]$ and $a_k\in\mathbb R^N$ has elements $(a_k)_i=\mathbf 1\{i\in\mathcal C_k\}/|\mathcal C_k|$, gives selective tests of $\mathsf R\theta^0(\widehat{\mathcal C})=0$ with a $\chi_q$ pivot truncated to the corresponding set. With $\mathsf R$ stacking the $K-1$ differences $(e_k-e_K)'\otimes I_P$, this is a joint test of the homogeneity hypothesis that does not require merging $p$-values, and with $\mathsf R=I_{KP}$ it tests the C-EPA hypothesis for the selected clusters directly, without combining it with an O-EPA test. These joint tests are valid conditional on the selection outcome by the arguments of Proposition~\ref{proposition:pairwise_p}, and Section~\ref{sec:monte} compares them with the merged tests below. For simultaneous inference on several pairs, see \cite{yun2024selective}.
\end{remark}

{\em Testing homogeneity.}\; We construct a $p$-value combination test for the homogeneity null \eqref{eq:homonull} by aggregating the $n_p$ selective $p$-values $p[d_{k,g}(\widehat{\mathcal{C}})]$ from all unique pairwise equality tests. Define
\[
\mathcal P_K=\{(k,g):1\leq k<g\leq K\}, 
\quad 
n_p=|\mathcal P_K|=\frac{K(K-1)}{2}.
\]
Following the recent studies of \citet{vovk20} and \citet{vovk2022admissible} on the M-family of merging functions, our proposed test is based on the generalized mean of order $r \in \mathbb{R} \setminus \{0\}$ defined as:
\begin{equation*}
F_{r} = b_{r,n_p} \left\lbrace \frac{1}{n_p} \sum_{(k,g)\in\mathcal P_K} \{ p[d_{k,g}(\widehat{\mathcal{C}})] \}^r \right\rbrace^{1/r} \wedge 1
\end{equation*}
where $b_{r,n_p}$ is a calibration constant chosen to ensure that $F_{r}$ is a valid $p$-value under arbitrary dependence among the $p$-values.

The M-family nests classical combination rules as special cases{:} $r=1$, $r \to 0$ and $r = -1$ {give the} arithmetic{,} geometric and harmonic {means}, respectively{, and} $r \to -\infty$ {gives the} Bonferroni $p$-merging function \citep{vovk20}. Not all of these{, however,} preserve the merging or precision properties under arbitrary dependence, especially for small numbers of $p$-values. Since the selective $p$-values are dependent, we choose $r \in [-\infty,-1)$, for which the M-mean is a valid $p$-merging function under arbitrary dependence and which gave the best finite-sample accuracy in our simulations among the choices considered by \citet{vovk20} and \citet{vovk2022admissible}. In the simulations and the application we set $r=-20$, and Section~\ref{sec:supp_mc} of the supplement reports the sensitivity of size and power to $r\in\{-2,-5,-10,-\infty\}$.
Following Proposition~5 of \citet{vovk20}, we set $b_{r,n_p} = [r/(r+1)] n_p^{1+1/r}$ for $r$ in this range.
The resulting homogeneity test statistic is given by:
\begin{equation}\label{eq:Mhomo}
F_{homo,r} = \frac{r}{r+1} n_p^{1+1/r} \left\lbrace \frac{1}{n_p} \sum_{(k,g)\in\mathcal P_K} \{ p[d_{k,g}(\widehat{\mathcal{C}})] \}^r \right\rbrace^{1/r} \wedge 1
\end{equation}
with $r \in [-\infty,-1)$.

This merging function is precise, monotone and sharp under arbitrary dependence of the inputs, and the normalization $[r/(r+1)] n_p^{1+1/r}$ makes \eqref{eq:Mhomo} a valid $p$-value under the global null \citep[Theorem 2 of][and Theorem 3 of \citealp{vovk2022admissible}]{vovk20}, so that the test based on $F_{homo,r}$ controls the Type I error of the homogeneity null under any dependence between the pairwise $p$-values. Section~\ref{sec:supp_framework} of the supplement compares it with Fisher's method and with the merging function of \cite{spreng23}.

Now we state the asymptotic properties of the test statistic $F_{homo,r}$.

\begin{theorem}\normalfont\label{theorem:homo}
Let $K \geq 2$ be given, and $B \to \infty$ as $T \to \infty$ such that $B/T \to 0$.
\begin{enumerate}[label=(\alph*)]
    \item\label{theorem:homo:parta} Under Assumptions~\ref{ass:alternative}--\ref{ass:clt} and~\ref{ass:stability}, let $s$ be a selection outcome with $\Pr\{\hat s_\infty(h+G)=s\}>0$ such that $\mathcal{H}^{homo}_{0}$ holds for $\mathcal C(s)$ for every $T$. Then
    \begin{equation*}
    \limsup\limits_{T \to \infty} \Pr[\,F_{homo,r} \leq \alpha \mid \hat s=s\,]\leq\alpha, \; \forall \alpha \in (0,1).
    \end{equation*}
    \item\label{theorem:homo:partb} Suppose now that
$K\geq2$ is fixed, $N$ is fixed, $h=0$, and $\mu^*$ is a regular point of the selection outcome with $\hat s_\infty\equiv s$. Under Assumptions
\ref{ass:alternative} and~\ref{ass:clt}, and if $\mathcal{H}^{homo}_{0}$ fails for $\mathcal C(s)$,
    \begin{equation*}
    \lim_{{T} \to \infty} \Pr [\, F_{homo,r} \leq \alpha \,] = 1, \; \forall \alpha \in (0,1).
    \end{equation*}
\end{enumerate}
\end{theorem}

Beyond its role in C-EPA testing, $F_{homo,r}$ is an alternative to the Split Sample homogeneity test of \cite{patton23}: Part~\ref{theorem:homo:parta} shows that it controls the Type I error rate conditionally on the selection outcome, and Part~\ref{theorem:homo:partb} that it is consistent if at least one pairwise null $\mathcal{H}^{k,g}_0$ fails.

\subsection{The overall EPA test and the main result}\label{sec:maintests}

To test the O-EPA hypothesis $\mathcal{H}^{oepa}_{0}$, that the two forecasts are equally good on average, consider
\begin{equation}
W_{oepa} = \frac{B-P+1}{PB} T\bar{Z}'_{o} \widehat{\Omega}^{-1}_{o} \bar{Z}_{o},
\end{equation}
where $\bar{Z}_{o} = T^{-1} \sum_{t=1}^T \bar{Z}_{t}$, 
$\bar{Z}_{t} = N^{-1} \sum_{i=1}^N Z_{it}$, and
$\widehat{\Omega}_{o} = B^{-1}\sum_{j=1}^B \widehat{\Lambda}_{o,j}\widehat{\Lambda}_{o,j}'$, with
\begin{equation}
\widehat{\Lambda}_{o,j}
=
\sqrt{\frac{2}{T}} \sum_{t=1}^T
(\bar{Z}_{t} - \bar{Z}_{o})
\cos \left[
\pi j \left( \frac{t-1/2}{T} \right)
\right].
\end{equation}

The asymptotic properties of $W_{oepa}$ are summarized in the following proposition.
\begin{proposition}\normalfont\label{proposition:overall}
Suppose that Assumptions~\ref{ass:alternative} and~\ref{ass:clt} hold with $\mathcal{C} = (1,\dots,1)$, that is, $K=1$.
Then, for $B$ fixed as $T \to \infty$, the following results hold; if instead $B\to\infty$ with $B/T\to0$, part (a) holds with $\mathbb{F}_{P,B-P+1}$ replaced by $\chi^2_P/P$, to which it converges, so that the $\mathbb{F}$ critical values remain asymptotically valid.
\begin{enumerate}[label=(\alph*)]
    \item\label{proposition:overall:parta} Under $\mathcal{H}^{oepa}_{0}$, $W_{oepa} \overset{d}{\longrightarrow} \mathbb{F}_{P,B-P+1}$.
    \item\label{proposition:overall:partb} Suppose that $\mathcal{H}^{oepa}_{0}$ fails. Let $C>0$ be a fixed real number. Then $\Pr[\; W_{oepa}>C \;] \to 1$.
\end{enumerate}
\end{proposition}

The test rejects the null of O-EPA if $p ( w_{oepa} ) = \Pr \left[\; \mathbb{F}_{P,B-P+1} \geq w_{oepa} \;\right] \leq \alpha$ where $\alpha \in (0,1)$ is the predetermined Type I error rate.
When $B=T$ and $P=1$, the statistic becomes a Wald-type statistic robust to arbitrary CSD but not controlling for autocorrelation, a special case of the $S^{(3)}$ test of APUY with the kernel bandwidth chosen to ignore potential autocorrelation.
For all tests we choose $B=\lceil c_BT^{2/3}\rceil$, the rate of the testing-optimal rule of \cite{sun13}, with $c_B$ given by that rule under an AR(1) approximation. The selective tests require $B\to\infty$, and with this choice the O-EPA and selective tests use the same variance estimator; the $\mathbb{F}$ critical values then provide a finite-sample refinement \citep{sun13}. Section~\ref{sec:monte} reports the sensitivity to $c_B$.

We now turn to our main test statistic for the C-EPA null $\mathcal{H}_{0}$. As in the previous section, we propose a $p$-value combination statistic built from the $p$-values of the $n_p$ pairwise equality tests and the O-EPA test:
\begin{equation}\label{eq:selective_comb}
F_{SI,r} = \frac{r}{r+1} (n_p+1)^{1+1/r} \left\lbrace \frac{1}{n_p+1} \sum_{(k,g)\in\mathcal P_K} \{ p[d_{k,g}(\widehat{\mathcal{C}})] \}^r + \frac{1}{n_p+1} p ( W_{oepa} )^r \right\rbrace^{1/r} \wedge 1
\end{equation}
where $r \in [-\infty,-1)$.

The following main result summarizes the desired properties of \eqref{eq:selective_comb}.
\begin{theorem}\normalfont\label{theorem:main}
Let $K \geq 2$ be given, and $B \to \infty$ as $T \to \infty$ such that $B/T \to 0$.
\begin{enumerate}[label=(\alph*)]
    \item\label{theorem:main:parta} Under Assumptions~\ref{ass:alternative}--\ref{ass:clt} and~\ref{ass:stability},
    \begin{equation*}
    \limsup\limits_{T \to \infty} \Pr[\,F_{SI,r} \leq \alpha,\ \mathcal{H}'_{0}\text{ holds for }\widehat{\mathcal C}\,]\leq\alpha, \; \forall \alpha \in (0,1).
    \end{equation*}
    \item\label{theorem:main:partb} Suppose now that
$K\geq2$ is fixed, $N$ is fixed, $h=0$, and $\mu^*$ is a regular point of the selection outcome. Under Assumptions
\ref{ass:alternative} and~\ref{ass:clt}, and if either $\mathcal{H}^{homo}_{0}$ or $\mathcal{H}^{oepa}_{0}$ fails for $\mathcal C(s)$, where $s$ is the constant limit of the selection outcome, then
    \begin{equation*}
    \lim_{{T} \to \infty} \Pr [\, F_{SI,r} \leq \alpha \,] = 1, \; \forall \alpha \in (0,1).
    \end{equation*}
\end{enumerate}
\end{theorem}

Part~\ref{theorem:main:partb} shows that the test is consistent when either $\mathcal{H}^{homo}_{0}$ or $\mathcal{H}^{oepa}_{0}$ fails; its finite-sample properties are studied in Section~\ref{sec:monte}. Part~\ref{theorem:main:parta} controls the probability of rejecting a true C-EPA hypothesis for the selected clusters, averaged over the selection outcomes, rather than conditionally on each selection outcome. The reason is that the O-EPA $p$-value is not uniform conditional on the selection outcome, since the O-EPA statistic and the clustering are computed from the same data. The pairwise components are conditionally valid and the O-EPA hypothesis does not depend on the partition, so the merged test is valid in this sense. Under the J-EPA null $\mathcal{H}^{jepa}_0$, which implies $\mathcal{H}'_0$ for every partition, part~\ref{theorem:main:parta} is ordinary size control. A test of C-EPA for the selected clusters that is valid conditionally on the selection outcome is the joint test of Remark~\ref{rem:joint} with $\mathsf R=I_{KP}$.

\begin{remark}\normalfont\label{rem:merging_power}
When the O-EPA hypothesis holds but the clusters are heterogeneous, the O-EPA component is uninformative and, through the merging function, it inflates the merged $p$-value, which lowers the power of $F_{SI,r}$ (see Table~\ref{tab:power}). Two alternatives preserve the control of the full null. The joint selective test of Remark~\ref{rem:joint} with $\mathsf R=I_{KP}$ avoids merging altogether. Alternatively, weighted merging functions \citep{vovk20} can give more weight to the homogeneity $p$-value, for instance a weighted Bonferroni combination $\min\{p_{homo}/w,\,p(W_{oepa})/(1-w)\}\wedge1$ with $w\in(0,1)$, where $p_{homo}$ is the homogeneity $p$-value \eqref{eq:Mhomo}. Section~\ref{sec:monte} compares these procedures.
\end{remark}

\section{Extensions}\label{sec:extensions}

This section considers three extensions of the tests of Section~\ref{sec:tests}: the estimation of the number of clusters, the hypothesis that EPA holds for every unit, and the identification of the clusters for which EPA fails.

\subsection{Estimating the number of clusters under the alternative}\label{sec:bic}

When the number of clusters under the alternative is unknown, it can be estimated from the sample in hand. For this purpose, \cite{patton23} suggest a multiple testing procedure based on the Bonferroni correction. Adapting their proposal, one would compute the $p$-value of \eqref{eq:selective_comb} for $K = 2,\dots,K_{max}$, apply the usual Bonferroni correction, and reject $\mathcal{H}_{0}$ if the Bonferroni $p$-value does not exceed the predetermined Type I error rate.
As an alternative, we propose an information criterion (IC) to estimate the number of clusters.
Consider the following IC:
\begin{equation*}
IC(K) = \log \left[ \mathrm{det} \left( \frac{1}{NT} \sum_{i=1}^N \sum_{t=1}^T \widehat{V}_{it}(K) \widehat{V}'_{it}(K) \right) \right] + (KP + N) \frac{\varsigma \log(NT)}{NT},
\end{equation*}
where $\widehat{V}_{it}(K) = Z_{it} - \hat{\theta}_{K,\hat{k}_i}$ with $\hat{\theta}_{K,\hat{k}_i}$ being the solution to \eqref{eq:kmeans} with $K$ clusters, and $\varsigma$ is a tuning constant.
The IC estimate of the number of clusters is given by
\begin{equation}\label{eq:kmeans_khat}
\widehat{K}_{IC} = \argmin_{K \in \{2,\dots,K_{max}\}} IC(K).
\end{equation}

For the Split Sample test, this IC can be adapted using only the training portion of the data, and other penalty functions can be employed \citep[see, for instance,][for different penalties for estimating the number of factors in factor models]{bai02}. Our IC adapts the one used by \cite{lumsdaine2023estimation} to our multivariate framework. Lemma~\ref{lemma:ic_consistent} in the supplement shows that $\widehat{K}_{IC}$ is consistent for $K^0 \geq 2$ under Assumptions~\ref{ass:alternative}--\ref{ass:separation} when $N$ is fixed and Panel Kmeans attains the global minimum for each $K$. The penalty increment between $K$ and $K+1$ clusters, $\varsigma P\log(NT)/(NT)$, dominates the $O_p(T^{-1})$ reduction of the criterion obtained by splitting a true cluster because $\log(NT)\to\infty$ with $N$ fixed. In particular, under local alternatives with $\sqrt T\mu^0$ bounded, the fit term changes by $O_p(T^{-1})$ across $K$, so $\widehat K_{IC}=2$ with probability approaching one, and larger values are selected only when the clusters are well separated. In finite samples with $N$ large relative to $T$, however, this increment is small, and $\widehat K_{IC}$ tends to over-estimate $K^0$ unless $\varsigma$ is large enough, which is consistent with the behavior reported below for small values of $\varsigma$. Over-estimation does not affect the validity of the selective tests, which condition on $\widehat K_{IC}$ (Proposition~\ref{proposition:bic}), but it affects their power and the interpretation of the clusters.

In our simulations, values $\varsigma \in [1.5, 3]$ work well, with smaller values over-estimating the number of clusters when the signal is weak; the upper bound $\varsigma = 3$ is also suggested by \cite{lumsdaine2023estimation}. Since homogeneity testing is embedded in our framework, we set $\varsigma=1.5$, accepting some over-estimation of the true number.

The IC is computationally cheaper than a Bonferroni $p$-value over $K$: it requires only the Panel Kmeans estimates for each $K$, not the truncation sets, which dominate the computational cost.

An alternative is cross-validation (CV) \citep{li2022consistent}: the data are repeatedly split into training and validation sets, and for each $K$ the within-cluster prediction error on the validation set is evaluated using parameters estimated from the training set. The $K$ minimizing the average out-of-sample prediction error across folds is the estimate, denoted $\widehat{K}_{CV}$. CV requires no tuning parameter but is computationally more intensive and, with random folds, is not covered by the conditioning event of the selective tests; we use it only as a robustness check.

A concern with a data-dependent choice of the number of clusters is that it might invalidate the selective inference procedure by requiring further conditioning on the choice of the information criterion. For instance, in valid inference on Lasso, choosing the tuning parameter of the objective function requires extra conditioning \citep{markovic2017unifying}. The following result shows that this is not the case in our framework.

\begin{proposition}\normalfont\label{proposition:bic}
Let $\widehat{\mathcal{C}}$ be a clustering with $K$ clusters and assume that $\widehat{\mathcal{C}}$ is the partition selected with $\widehat K_{IC}$ by the procedure of Section~\ref{sec:preliminaries}, given the initializations. Then, the inference procedures that condition on the clustering assignment $\widehat{\mathcal{C}}$ implicitly condition on $\widehat{K}_{IC}$ as well.
That is,
\[
\Pr\left[D_{k,g}(\widehat{\mathcal{C}}) \in \mathcal{T} \;\middle|\; \bigcap_{i=1}^N \lbrace \hat{k}_i = k_i \rbrace \right] = \Pr\left[D_{k,g}(\widehat{\mathcal{C}}) \in \mathcal{T} \;\middle|\; \widehat{K}_{IC} = K,\; \bigcap_{i=1}^N \lbrace \hat{k}_i = k_i \rbrace \right],
\]
where $\widehat{K}_{IC}$ is given by \eqref{eq:kmeans_khat}.
\end{proposition}

\begin{remark}\normalfont\label{rem:prop_ic}
Proposition~\ref{proposition:bic} holds because the selected partition has $\widehat K_{IC}$ clusters. The event that the whole procedure selects $\widehat{\mathcal C}$ is smaller than the event that Algorithm~\ref{algo} with $K$ clusters returns it, since along the perturbation path \eqref{eq:perturbed} another $K'$ may be selected. The truncation set \eqref{eq:set_s} thus includes the comparisons of $IC(K')$, $K'\in\mathcal K$, which are polynomial inequalities in $\phi$ of degree at most $2P$, and remains a finite union of intervals computed exactly (Section~\ref{sec:calc_of_p} of the supplement). Multiple initializations are covered in the same way: they are drawn independently of the data, so that size is controlled conditionally on them, hence uniformly and unconditionally, and the initializations not retained matter only in that they do not change the selection. Conditioning also on the paths of all initializations, as in \cite{chen23}, is valid but much less powerful, because small perturbations then change intermediate assignments without changing the selected partition (Section~\ref{sec:monte}).
\end{remark}

\subsection{Testing joint EPA}\label{sec:unitlevel}

The C-EPA hypothesis concerns clusters of units. A stronger hypothesis of empirical interest is the J-EPA hypothesis that EPA holds for every unit, $\mathcal{H}^{jepa}_0$ in \eqref{eq:nullU}, which implies
\begin{equation}\label{eq:nullUprime}
\mathcal{H}^{jepa\prime}_0: \mu^0_1=\mu^0_2=\dots=\mu^0_N=0.
\end{equation}
This hypothesis does not refer to any partition, and it implies $\mathcal{H}'_0$ for every partition, in particular for the selected one. Under $\mathcal{H}^{jepa\prime}_0$, the hypotheses tested by the selective $p$-values of Section~\ref{sec:tests} hold for every selection outcome, so that Proposition~\ref{proposition:pairwise_p} and the joint tests of Remark~\ref{rem:joint} give tests of $\mathcal{H}^{jepa\prime}_0$ that control the Type I error rate conditionally on the selection outcome, and therefore also unconditionally, while Theorem~\ref{theorem:main} gives unconditional control for the merged test.

The clustering then serves to direct power against alternatives in which groups of units gain and lose. The natural statistic is the joint selective test of Remark~\ref{rem:joint} with $\mathsf R=I_{KP}$, which tests that all selected cluster means are zero and does not require merging with the O-EPA $p$-value. It has power against alternatives in which O-EPA holds but the expected loss differentials have offsetting signs, against which the O-EPA test has no power, and, unlike tests based on aggregates of unit-level statistics, it identifies the groups of units responsible for a rejection, with valid inference on the differences between them.

\subsection{Identifying the clusters that violate EPA}\label{sec:clusterspecific}

When the C-EPA hypothesis is rejected, it is of interest to know which clusters are responsible, for instance which groups of units gain or lose from a competing forecast. The tests of Section~\ref{sec:tests} do not answer this question, since the pairwise tests compare the clusters with each other and the O-EPA test compares the panel average with zero. For each selected cluster $k\in\{1,\dots,\widehat K\}$, we therefore test
\begin{equation}\label{eq:clusternull}
\mathcal{H}^{k}_0:\ \theta^0_k(\widehat{\mathcal C})=0,
\end{equation}
the hypothesis that EPA holds on average for the units of the $k$th selected cluster, so that $\mathcal H'_0$ for the selected partition is the intersection of $\mathcal H^1_0,\dots,\mathcal H^{\widehat K}_0$. The statistic is the square root of the Wald statistic $D^2_{k}(\widehat{\mathcal{C}}) = T \hat{\theta}_{k}(\widehat{\mathcal{C}})' \widehat{\omega}_{k,k}(\widehat{\mathcal{C}})^{-1}\hat{\theta}_{k}(\widehat{\mathcal{C}})$, and the decomposition \eqref{eq:zrewritten_inthetext} and the perturbation \eqref{eq:perturbed} are formed with $L_k=a_k\otimes I_P$ in place of $L_{k,g}$, with $a_k$ as in Remark~\ref{rem:joint} for $\widehat{\mathcal C}$; this is the joint test of Remark~\ref{rem:joint} with $\mathsf R=e_k'\otimes I_P$. Along the perturbation path, the time averages of the units of cluster $k$ move towards or away from zero, each by its own amount, while the residual $\widehat U_k$ is held fixed, so that the means of the other clusters, which are nuisance parameters for $\mathcal H^k_0$, are conditioned on. The selective $p$-value $p[d_{k}(\widehat{\mathcal{C}})]$ is computed as in Lemma~\ref{lemma:perturbation} with $\chi_P$ truncated to the corresponding set. By the arguments of Proposition~\ref{proposition:pairwise_p}, it is asymptotically uniform under $\mathcal H^k_0$ conditionally on the selection outcome, whatever the means of the other clusters, and it converges to zero when $\mathcal H^k_0$ fails at a regular point.

The $\widehat K$ selective $p$-values are dependent, and their joint distribution conditional on the selection outcome is not tractable. We combine them with the step-down procedure of \cite{holm1979simple}, which is valid under arbitrary dependence. Let $p_{(1)}\leq\dots\leq p_{(\widehat K)}$ be the ordered $p$-values and $\mathcal H^{(1)}_0,\dots,\mathcal H^{(\widehat K)}_0$ the corresponding hypotheses. The procedure rejects $\mathcal H^{(j)}_0$ when the adjusted $p$-value $p^{holm}_{(j)}=\max_{l\leq j}\min\{(\widehat K-l+1)p_{(l)},1\}$ does not exceed $\alpha$.

\begin{corollary}\normalfont\label{cor:holm}
Let $K\geq2$ be given, and $B \to \infty$ as $T \to \infty$ such that $B/T \to 0$.
\begin{enumerate}[label=(\alph*)]
    \item\label{cor:holm:parta} Under Assumptions~\ref{ass:alternative}--\ref{ass:clt} and~\ref{ass:stability}, let $s$ be a selection outcome with $\Pr\{\hat s_\infty(h+G)=s\}>0$, and let $\mathcal I_0(s)$ be the set of clusters $k$ for which $\mathcal H^k_0$ holds for the partition $\mathcal C(s)$ for every $T$. Then
    \begin{equation*}
    \limsup\limits_{T \to \infty} \Pr\big[\,p^{holm}_{k}\leq\alpha \text{ for some } k\in\mathcal I_0(s) \mid \hat s=s\,\big]\leq\alpha, \; \forall \alpha \in (0,1),
    \end{equation*}
    where $p^{holm}_{k}$ is the adjusted $p$-value of $\mathcal H^k_0$.
    \item\label{cor:holm:partb} Under the conditions of Proposition~\ref{proposition:pairwise_p}\ref{proposition:pairwise_p:partc}, if $\mathcal H^k_0$ fails for $\mathcal C(s)$, then $\lim_{T\to\infty}\Pr[\,p^{holm}_{k}\leq\alpha\,]=1$ for all $\alpha\in(0,1)$.
\end{enumerate}
\end{corollary}

The procedure thus controls the family-wise error rate conditionally on the selection outcome, and the signs of the estimated cluster means of the rejected hypotheses indicate whether the competing forecast is better or worse than the benchmark for these clusters.

Rejecting $\mathcal H'_0$ for the selected partition when $p^{holm}_{(1)}=\min\{\widehat K p_{(1)},1\}\leq\alpha$ also gives a test of C-EPA for the selected clusters that is valid conditionally on the selection outcome, like the joint test of Remark~\ref{rem:joint} with $\mathsf R=I_{KP}$. We do not recommend it for this purpose. Perturbing a single cluster mean moves its units towards the neighboring clusters, so that the selected partition changes after small perturbations and the truncation sets are narrow, especially when the clusters overlap. Furthermore, when all units have the same nonzero expected loss differential, the O-EPA statistic, which does not depend on the partition, is far more informative than any single cluster mean. In simulations not reported here, this test was less powerful than the combination of the pairwise and O-EPA tests in all designs considered, and had almost no power when all expected loss differentials were shifted by the same amount. We therefore apply the cluster-specific tests after a rejection by the tests of Section~\ref{sec:tests}, to locate the clusters for which EPA fails. Since the Holm procedure controls the family-wise error rate on its own, applying it only after a rejection preserves this control. Section~\ref{sec:app} reports these tests for the exchange rate application.

\section{Monte Carlo Study}\label{sec:monte}

We study the size and power of the tests in a design that extends that of \cite{hoga2023testing}, whose focus is on measurement errors in the target variable, to clustered heterogeneity, estimated forecasts and CSD. The target follows the panel AR(1) process
\begin{equation}\label{eq:sim_dgp}
Y_{it} = \alpha_i(1-\rho_{k_i}) + \rho_{k_i} Y_{i,t-1} + \gamma_iF^y_t + U_{it}, \quad U_{it} \sim iid \, \mathbb{N}(0, 1),
\end{equation}
with individual effects $\alpha_i\sim\mathbb N(1,0.25)$, a stationary AR(1) common factor $F^y_t$ with coefficient 0.5 and loadings $\gamma_i\sim U(0,1)$, all drawn once per replication, and $(\rho_1,\rho_2,\rho_3)=(0.1,0.2,0.3)$. The $N=80$ units form three latent clusters, $k^0_i=1$ for $i\leq N/4$, $k^0_i=2$ for $N/4<i\leq N/2$ and $k^0_i=3$ otherwise. Two forecasters estimate the AR(1) model by least squares on rolling windows of length $R=60$, Forecaster~1 with an intercept and Forecaster~2 without, and their one-step-ahead forecasts are perturbed by the noise $\sqrt{v_{a,it}}\,\varepsilon_{a,it}$, $a=1,2$, where $\varepsilon_{a,it}=\phi\varepsilon_{a,i,t-1}+\lambda F_{a,t}+(1-\phi^2-\lambda^2)^{1/2}\xi_{a,it}$, with $\xi_{a,it}$ and the common factors $F_{a,t}$ independent standard normal, $\phi=0.2$ and $\lambda=0.2$. The noise variances $v_{a,it}$ are functions of $Y_{i,t-1}$, calibrated unit by unit by a preliminary simulation so that
\[
\mathbb E(H_{i,t-1}\Delta L_{it})=\psi_{k_i}(1,\alpha_i)',\quad \Delta L_{it}=(Y_{it}-\widehat Y_{1,it})^2-(Y_{it}-\widehat Y_{2,it})^2,
\]
for the unconditional and conditional tests, with $H_{i,t-1}=1$ and $H_{i,t-1}=(1,Y_{i,t-1})'$; the noise is added to the forecaster whose loss must increase. Under the null, $\psi_1=\psi_2=\psi_3=0$, so that EPA holds for every unit. Under the alternative, we consider Case~1, $(\psi_1,\psi_2,\psi_3)=\psi/2\cdot(1,1,1)+\psi\cdot(-1.2,-0.8,1)$, in which O-EPA fails, and Case~2, $(\psi_1,\psi_2,\psi_3)=\psi\cdot(-1.2,-0.8,1)$, in which it holds, with $\psi\in\{0.125,0.25,0.375,0.5\}$. With $N$ fixed and strictly stationary and geometrically mixing data, the loss differentials satisfy Assumption~\ref{ass:clt} for any form of CSD (Section~\ref{sec:primitive_forecasts} of the supplement). All results are based on 1000 replications, with $T\in\{20,50,100,200\}$ under the null and $T\in\{50,200\}$ under the alternative.

{\em Tests.}\; The Predetermined and Naive tests use the statistic of Section~\ref{sec:known} with the true clusters $k^0_i$ and with the Panel Kmeans clusters $\hat k_i$, respectively. The Split Sample test (Section~\ref{sec:split} of the supplement) estimates the clusters from the first 20\% of the sample and tests on the remaining observations after a gap of $\lfloor(0.2T)^{1/2}\rfloor$ periods. The selective tests are the merged test $F_{SI,r}$ of \eqref{eq:selective_comb} with $r=-20$, the joint test of Remark~\ref{rem:joint} with $\mathsf R=I_{KP}$, and the weighted Bonferroni combination of Remark~\ref{rem:merging_power} with $w=1/2$. All tests except Predetermined use $\widehat K_{IC}$ with $K_{max}=5$ and $\varsigma=1.5$, and Algorithm~\ref{algo} with $S=10$ initializations and at most 100 iterations; the selective tests condition on the selected partition and on $\widehat K_{IC}$. All full-sample tests use the same $B=\lceil c_BT^{2/3}\rceil$ of Section~\ref{sec:maintests}, bounded below by $K_{max}P+4$.

\begin{table}[t]
\centering
\begin{threeparttable}
\caption{Rejection rates under the null hypothesis}\label{tab:size_main}
\small\setlength{\tabcolsep}{4pt}\renewcommand{\arraystretch}{0.9}
\begin{tabular}{r cccccc c cccccc}
\toprule
& \multicolumn{6}{c}{Unconditional ($H_{i,t-1}=1$)} && \multicolumn{6}{c}{Conditional ($H_{i,t-1}=(1,Y_{i,t-1})'$)}\\
\cmidrule{2-7}\cmidrule{9-14}
$T$ & Pred. & Naive & Split & $F_{SI,r}$ & Joint & W.B. && Pred. & Naive & Split & $F_{SI,r}$ & Joint & W.B.\\
\midrule
20 & .057 & .982 & .040 & .041 & .083 & .050 && .061 & .824 & .065 & .057 & .109 & .065\\
50 & .061 & .999 & .062 & .057 & .076 & .056 && .058 & .878 & .065 & .046 & .096 & .047\\
100 & .061 & .999 & .060 & .059 & .085 & .056 && .059 & .922 & .063 & .064 & .086 & .064\\
200 & .046 & 1.00 & .067 & .047 & .058 & .052 && .054 & .979 & .074 & .053 & .079 & .056\\
\bottomrule
\end{tabular}
\begin{tablenotes}\footnotesize
\item Note: Rejection rates at the 5\% level, 1000 replications, $N=80$. Pred.: Predetermined; Split: Split Sample with training proportion $\gamma=0.2$; $F_{SI,r}$: merged selective test with $r=-20$; Joint: joint selective test with $\mathsf R=I_{KP}$; W.B.: weighted Bonferroni combination with $w=1/2$.
\end{tablenotes}
\end{threeparttable}
\end{table}

\begin{table}[t]
\centering
\begin{threeparttable}
\caption{Rejection rates under the alternative hypothesis}\label{tab:power}
\small\setlength{\tabcolsep}{4pt}\renewcommand{\arraystretch}{0.9}
\begin{tabular}{rr cccccc c cccccc}
\toprule
&& \multicolumn{6}{c}{Unconditional} && \multicolumn{6}{c}{Conditional}\\
\cmidrule{3-8}\cmidrule{10-15}
$T$ & $\psi$ & Pred. & Naive & Split & $F_{SI,r}$ & Joint & W.B. && Pred. & Naive & Split & $F_{SI,r}$ & Joint & W.B.\\
\midrule
\multicolumn{15}{l}{Case 1: O-EPA fails}\\
50 & 0.125 & .477 & .998 & .077 & .059 & .103 & .067 && .379 & .931 & .076 & .082 & .094 & .102\\
200 & 0.125 & .994 & 1.00 & .257 & .145 & .109 & .164 && .977 & .999 & .156 & .170 & .123 & .190\\
50 & 0.250 & .964 & 1.00 & .174 & .112 & .127 & .152 && .903 & .978 & .124 & .140 & .140 & .173\\
200 & 0.250 & 1.00 & 1.00 & .928 & .514 & .287 & .632 && 1.00 & 1.00 & .654 & .498 & .323 & .580\\
50 & 0.375 & 1.00 & 1.00 & .382 & .234 & .176 & .316 && .998 & .984 & .197 & .261 & .216 & .329\\
200 & 0.375 & 1.00 & 1.00 & 1.00 & .849 & .399 & .915 && 1.00 & 1.00 & .966 & .854 & .523 & .906\\
50 & 0.500 & 1.00 & 1.00 & .700 & .396 & .214 & .549 && 1.00 & .992 & .369 & .415 & .288 & .564\\
200 & 0.500 & 1.00 & 1.00 & 1.00 & .977 & .510 & .992 && 1.00 & 1.00 & .999 & .978 & .613 & .991\\
\multicolumn{15}{l}{Case 2: O-EPA holds}\\
50 & 0.125 & .485 & .999 & .069 & .058 & .079 & .081 && .351 & .917 & .066 & .068 & .101 & .070\\
200 & 0.125 & .994 & 1.00 & .168 & .061 & .107 & .068 && .966 & .999 & .087 & .076 & .123 & .083\\
50 & 0.250 & .958 & 1.00 & .119 & .073 & .127 & .070 && .879 & .964 & .083 & .078 & .121 & .086\\
200 & 0.250 & 1.00 & 1.00 & .888 & .182 & .284 & .153 && 1.00 & 1.00 & .493 & .220 & .358 & .215\\
50 & 0.375 & 1.00 & 1.00 & .342 & .092 & .184 & .082 && .999 & .991 & .121 & .103 & .225 & .102\\
200 & 0.375 & 1.00 & 1.00 & 1.00 & .334 & .417 & .283 && 1.00 & 1.00 & .961 & .423 & .545 & .398\\
50 & 0.500 & 1.00 & 1.00 & .677 & .112 & .230 & .090 && 1.00 & .999 & .353 & .164 & .342 & .157\\
200 & 0.500 & 1.00 & 1.00 & 1.00 & .479 & .511 & .430 && 1.00 & 1.00 & .998 & .609 & .687 & .585\\
\bottomrule
\end{tabular}
\begin{tablenotes}\footnotesize
\item Note: See Table~\ref{tab:size_main}. Case~1: $(\psi_1,\psi_2,\psi_3)=\psi/2\cdot(1,1,1)+\psi\cdot(-1.2,-0.8,1)$; Case~2: $(\psi_1,\psi_2,\psi_3)=\psi\cdot(-1.2,-0.8,1)$.
\end{tablenotes}
\end{threeparttable}
\end{table}

{\em Size.}\; Table~\ref{tab:size_main} shows the consequences of double dipping. The Naive test, which treats the estimated clusters as known, rejects a true null hypothesis in 82\% to 100\% of the replications, and more often as $T$ increases, because Panel Kmeans separates the estimated cluster means even when the true means are equal. The Predetermined, Split Sample, $F_{SI,r}$ and weighted Bonferroni tests have rejection rates between 0.040 and 0.074 for all $T$. The joint selective test is oversized for small $T$, with rejection rates of 0.083 and 0.109 for the unconditional and conditional tests with $T=20$, falling to 0.058 and 0.079 with $T=200$.

{\em Power.}\; Table~\ref{tab:power} reports the rejection rates under the alternative. The Predetermined test, which uses the true clusters, is an infeasible benchmark. In Case~1, the power of $F_{SI,r}$ comes mostly from its O-EPA component and increases quickly with $\psi$ and $T$, reaching 0.977 with $\psi=0.5$ and $T=200$. In Case~2, the O-EPA component carries no signal and $F_{SI,r}$ is less powerful, with rejection rates of 0.182 and 0.479 for $\psi=0.25$ and $\psi=0.5$ with $T=200$. The joint selective test, which does not merge with the O-EPA $p$-value, is more powerful in this case for moderate alternatives (0.284 and 0.511), whereas the weighted Bonferroni combination does not improve on $F_{SI,r}$ (0.153 and 0.430), because the homogeneity $p$-value is rarely small enough to benefit from its larger weight. The Split Sample test is the most powerful feasible test in this design, with rejection rates of 0.928 and 0.888 in Cases~1 and~2 with $\psi=0.25$ and $T=200$. Here the clusters are stable over time and overlap: a short training sample yields a partition that is correlated with the true one, and the testing sample is independent of it, whereas the selective tests discard the information that the data carry about the separation of the selected clusters. Part of this difference is due to the estimation of $\Xi L$: with the shrinkage estimator of Section~\ref{sec:selective}, the rejection rate of $F_{SI,r}$ in Case~2 with $\psi=0.5$ and $T=200$ rises from 0.479 to 0.918, and that of the homogeneity test from 0.501 to 0.920 (Table~\ref{tab:supp_var} of the supplement). The conditional tests behave similarly.

{\em Robustness and additional experiments.}\; Section~\ref{sec:supp_mc} of the supplement reports further experiments. The merged tests are insensitive to $r\le-10$, and $r=-2$ is conservative. Conditioning on the paths of all initializations, as in \cite{chen23}, controls size but leaves the pairwise tests without power, which supports conditioning on the selection outcome only. When the expected loss differentials change sign in the middle of the sample, so that the null hypotheses hold on average over the evaluation period (Section~\ref{sec:frame}), the Predetermined and $F_{SI,r}$ tests reject in 2.0\% to 7.6\% of the replications and the joint selective test in 6.5\% to 10.1\%, with no tendency to increase with $T$, whereas the rejection rate of the Split Sample test with a training proportion of 0.5 grows with $T$, up to 94.9\% with $T=200$, because it tests the cluster means after the break (Section~\ref{sec:supp_breaks} of the supplement). The shrinkage estimator of $\Xi L$ roughly doubles the power of the selective tests with size close to the nominal level; in this design the Euclidean decomposition behaves similarly. In a second design, with strongly heterogeneous factor loadings, long-run variances and serial correlation, in which the pairwise differences are far from eigenvectors of $\Xi$, the Euclidean decomposition and the shrinkage estimator have rejection rates between 4.6\% and 8.1\% under the null hypothesis for $T\geq50$ and are again much more powerful than the metric of $\Xi$ (Table~\ref{tab:supp_euclfail} of the supplement). We therefore report the selective tests with both estimators of $\Xi L$ in the application. The $\mathbb F$ reference law makes the tests very conservative. The results are insensitive to the number $S\in\{1,10,100\}$ and type of initializations, and power increases with $c_B$ and when fewer clusters are selected. The information criterion with $\varsigma=1.5$ selects the true number of clusters more often than $\varsigma=3$ and cross-validation, except for large separations. When Panel Kmeans splits a true cluster, the selective test of equality of the two resulting clusters rejects in 4.7\% and 9.2\% of the cases for the unconditional and conditional tests, against 100\% and 99.2\% for the naive test. With heterogeneity within the true clusters, the selective tests reject C-EPA for the estimated clusters when it holds for the true ones, because they test the population Kmeans partition, whose cluster means differ. Finally, when every unit has its own expected loss differential and there is no cluster structure, unit-level tests of J-EPA are much more powerful than the selective tests, and screening on the maximum of the unit-level statistics recovers a large part of this power for the joint selective test.

\section{Empirical Application: Forecasting Exchange Rate Returns}\label{sec:app}

We compare one-month-ahead forecasts of exchange rate returns against the US dollar from thirteen methods with those of an AR(1) benchmark. The methods range from univariate time series models and panel AR(1) estimators to machine learning methods with macroeconomic predictors.

{\em Data and forecasts.}\; Monthly exchange rates are from the IMF, from January 1999 to December 2023; returns are 100 times the first differences of the log rates. The macroeconomic predictors are the FRED-MD series, transformed and imputed by the Tall-Wide method of the \texttt{fbi} package \citep{Chan2022fbi}, without the four exchange rate series. Forecasts use a rolling window of $R=60$ months. Within each window, the returns of each currency are demeaned and scaled by their window mean and standard deviation, and forecasts, realized returns and quadratic losses are expressed in these units, so that losses are comparable across currencies without using information posterior to the forecast date. Following the insanity filter of \cite{swanson1997model}, forecasts are truncated at five window standard deviations. Devaluations and changes of exchange rate regime produce realized returns of many window standard deviations and loss differentials hundreds of times larger than the others, which Panel Kmeans would isolate in singleton clusters; we therefore winsorize the standardized realized returns at $\pm5$. Currencies in a monetary union or with a hard peg to a common anchor have almost identical returns against the dollar. Identical series make the long-run covariance matrix $\Xi$ singular and tie in Panel Kmeans, so we keep one currency of each group of series whose standardized returns have an absolute correlation above 0.98, the anchor when it is in the sample. Excluding the series with missing observations or with a window of almost constant returns, and the duplicated series, leaves $N=86$ currencies and $T=238$ forecasts, from March 2004 to December 2023.

{\em Forecasting methods.}\; The univariate methods are the AR($p$) model with $p\le6$ chosen by BIC, the Elastic Net on six lags, the no-change forecast of the exchange rate, that is a zero return, and the AR(1) forecast with a coefficient shrunk towards zero by the factor $\max(0,1-1/t^2)$, with $t$ its $t$-statistic. The panel methods estimate the AR(1) coefficient by pooling all currencies, by empirical Bayes shrinkage of each coefficient towards the pooled one, and by pooling within three groups of coefficients \citep{pesaran26}. The methods with predictors are XGBoost, support vector machines and random forests on the own lag and the macroeconomic predictors; a ridge regression on the own lag, the lagged dollar factor of \cite{lustig2011common} and three principal components of the predictors \citep{stock2002macroeconomic}; and a boosting model trained on all currencies together \citep{montero2021principles}. The last method is the equal-weight combination of AR($p$), Elastic Net, XGBoost, SVM and random forest. Section~\ref{sec:supp_app} of the supplement gives the details. We test EPA against the AR(1) benchmark unconditionally and conditionally on the lagged return, a unit-specific testing function for which the null hypothesis is J-EPA (Section~\ref{sec:frame}), and report the results for the post-crisis and post-COVID dummies in the supplement.

{\em Descriptive analysis.}\; The average loss differentials are small relative to their dispersion (Section~\ref{sec:supp_desc} of the supplement). Only the shrinkage estimators of the AR(1) coefficient improve on AR(1) on average, and SVM, the no-change forecast and XGBoost are the least accurate.

\begin{table}[!ht]
\centering
\caption{$p$-values of the EPA tests of the competing methods against the AR(1) benchmark}\label{tab:app_results}
\begin{threeparttable}
\renewcommand{\arraystretch}{0.88}\small
\begin{tabular}{lccccccccc}
\toprule
&&&&& \multicolumn{3}{c}{Metric of $\Xi$} & \multicolumn{2}{c}{Shrinkage} \\ \cmidrule(lr){6-8}\cmidrule(lr){9-10}
Method & $\widehat K$ & O-EPA & Naive & Split & $F_{SI,r}$ & Joint & Homog. & $F_{SI,r}$ & Homog. \\
\midrule
\multicolumn{10}{l}{A. Unconditional ($H_{i,t-1}=1$)} \\
\midrule
\addlinespace[2pt]\multicolumn{10}{l}{\footnotesize Univariate methods} \\
\quad AR($p$) & 2 & 0.00$^{\dagger}$ & 0.00$^{\dagger}$ & 0.00$^{\dagger}$ & 0.00$^{\dagger}$ & 0.01 & 0.83 & 0.00$^{\dagger}$ & 1.00 \\
\quad Elastic Net & 2 & 0.00$^{\dagger}$ & 0.00$^{\dagger}$ & 0.05 & 0.01$^{\dagger}$ & 0.32 & 0.37 & 0.01$^{\dagger}$ & 0.44 \\
\quad No change & 5 & 0.00$^{\dagger}$ & 0.00$^{\dagger}$ & 0.00$^{\dagger}$ & 0.00$^{\dagger}$ & 0.08 & 0.66 & 0.00$^{\dagger}$ & 0.00$^{\dagger}$ \\
\quad Shrunk AR(1) & 3 & 0.00$^{\dagger}$ & 0.00$^{\dagger}$ & 0.00$^{\dagger}$ & 0.00$^{\dagger}$ & 0.59 & 1.00 & 0.00$^{\dagger}$ & 1.00 \\
\addlinespace[2pt]\multicolumn{10}{l}{\footnotesize Panel AR(1) methods} \\
\quad Pooled AR(1) & 3 & 0.69 & 0.00$^{\dagger}$ & 0.01 & 0.71 & 0.02 & 0.53 & 0.48 & 0.36 \\
\quad Empirical Bayes AR(1) & 2 & 0.00$^{\dagger}$ & 0.00$^{\dagger}$ & 0.00$^{\dagger}$ & 0.00$^{\dagger}$ & 0.05 & 0.46 & 0.00$^{\dagger}$ & 0.57 \\
\quad Grouped AR(1) & 3 & 0.04 & 0.00$^{\dagger}$ & 0.02 & 0.17 & 0.99 & 1.00 & 0.17 & 0.90 \\
\addlinespace[2pt]\multicolumn{10}{l}{\footnotesize Methods with macroeconomic or cross-sectional predictors} \\
\quad XGBoost & 2 & 0.00$^{\dagger}$ & 0.00$^{\dagger}$ & 0.00$^{\dagger}$ & 0.00$^{\dagger}$ & 0.14 & 0.27 & 0.00$^{\dagger}$ & 0.37 \\
\quad SVM & 3 & 0.00$^{\dagger}$ & 0.00$^{\dagger}$ & 0.00$^{\dagger}$ & 0.00$^{\dagger}$ & 0.14 & 0.39 & 0.00$^{\dagger}$ & 0.52 \\
\quad Random Forest & 2 & 0.01 & 0.00$^{\dagger}$ & 0.11 & 0.03 & 0.55 & 0.73 & 0.03 & 0.82 \\
\quad Factor ridge & 3 & 0.29 & 0.05$^{\dagger}$ & 0.42 & 1.00 & 0.95 & 1.00 & 1.00 & 1.00 \\
\quad Global boosting & 3 & 0.39 & 0.00$^{\dagger}$ & 0.04 & 0.28 & 0.30 & 0.21 & 1.00 & 1.00 \\
\addlinespace[2pt]
\quad Combination & 3 & 0.35 & 0.00$^{\dagger}$ & 0.07 & 0.03 & 0.29 & 0.02 & 0.03 & 0.02 \\
\midrule
\multicolumn{10}{l}{B. Conditional on the lagged return ($H_{i,t-1}=(1,Y_{i,t-1})'$)} \\
\midrule
\addlinespace[2pt]\multicolumn{10}{l}{\footnotesize Univariate methods} \\
\quad AR($p$) & 4 & 0.00$^{\dagger}$ & 0.01$^{\dagger}$ & 0.04 & 0.00$^{\dagger}$ & 0.54 & 1.00 & 0.00$^{\dagger}$ & 1.00 \\
\quad Elastic Net & 3 & 0.05 & 0.00$^{\dagger}$ & 0.15 & 0.21 & 0.29 & 0.64 & 0.21 & 1.00 \\
\quad No change & 5 & 0.00$^{\dagger}$ & 0.00$^{\dagger}$ & 0.00$^{\dagger}$ & 0.00$^{\dagger}$ & 0.03 & 0.10 & 0.00$^{\dagger}$ & 0.00$^{\dagger}$ \\
\quad Shrunk AR(1) & 2 & 0.00$^{\dagger}$ & 0.00$^{\dagger}$ & 0.01 & 0.00$^{\dagger}$ & 0.05 & 0.62 & 0.00$^{\dagger}$ & 0.68 \\
\addlinespace[2pt]\multicolumn{10}{l}{\footnotesize Panel AR(1) methods} \\
\quad Pooled AR(1) & 5 & 0.83 & 0.01$^{\dagger}$ & 0.07 & 0.47 & 0.12 & 0.43 & 0.20 & 0.18 \\
\quad Empirical Bayes AR(1) & 2 & 0.00$^{\dagger}$ & 0.00$^{\dagger}$ & 0.00$^{\dagger}$ & 0.00$^{\dagger}$ & 0.00$^{\dagger}$ & 0.06 & 0.00$^{\dagger}$ & 0.41 \\
\quad Grouped AR(1) & 2 & 0.02 & 0.03$^{\dagger}$ & 0.05 & 0.04 & 0.31 & 0.26 & 0.04 & 0.33 \\
\addlinespace[2pt]\multicolumn{10}{l}{\footnotesize Methods with macroeconomic or cross-sectional predictors} \\
\quad XGBoost & 2 & 0.00$^{\dagger}$ & 0.00$^{\dagger}$ & 0.00$^{\dagger}$ & 0.00$^{\dagger}$ & 0.00$^{\dagger}$ & 0.03 & 0.00$^{\dagger}$ & 0.09 \\
\quad SVM & 2 & 0.00$^{\dagger}$ & 0.00$^{\dagger}$ & 0.00$^{\dagger}$ & 0.00$^{\dagger}$ & 0.04 & 0.45 & 0.00$^{\dagger}$ & 0.62 \\
\quad Random Forest & 2 & 0.01 & 0.00$^{\dagger}$ & 0.11 & 0.02 & 0.11 & 0.12 & 0.02 & 0.21 \\
\quad Factor ridge & 2 & 0.52 & 0.00$^{\dagger}$ & 0.50 & 0.35 & 0.01 & 0.17 & 0.32 & 0.15 \\
\quad Global boosting & 5 & 0.36 & 0.00$^{\dagger}$ & 0.06 & 1.00 & 0.42 & 1.00 & 1.00 & 1.00 \\
\addlinespace[2pt]
\quad Combination & 3 & 0.44 & 0.01$^{\dagger}$ & 0.32 & 0.49 & 0.35 & 0.37 & 1.00 & 1.00 \\
\bottomrule
\end{tabular}
\begin{tablenotes}\footnotesize
\item Note: $p$-values for the loss differentials of each method relative to the AR(1) benchmark, $N=86$ currencies, $T=238$ months. $\widehat K$ is chosen by the information criterion with $\varsigma=1.5$ and $K_{\max}=5$; Panel Kmeans uses 100 initializations. O-EPA: overall EPA test; Naive: Wald test of C-EPA treating the estimated clusters as known; Split: Split Sample test with $\gamma=0.1$; $F_{SI,r}$: selective C-EPA test with $r=-20$; Joint: joint selective C-EPA test; Homog.: selective test of homogeneity. $^{\dagger}$: significant at 5\% after the Holm adjustment across the 13 methods. 0.00 denotes $p<0.005$.
\end{tablenotes}
\end{threeparttable}
\end{table}

{\em Test results.}\; Table~\ref{tab:app_results} reports the $p$-values, with the number of clusters chosen by the information criterion. The O-EPA test rejects for nine of the thirteen methods in both panels. After the Holm adjustment across methods, empirical Bayes and the shrunk AR(1) remain more accurate than AR(1), and AR($p$), XGBoost, SVM and the no-change forecast less accurate, as is the Elastic Net in Panel A. The Split Sample test mostly agrees and, before adjustment for multiple comparisons, also rejects in Panel A for the pooled AR(1) and global boosting. The naive test, which treats the estimated clusters as known, rejects C-EPA for all thirteen methods in both panels, also for the pooled AR(1), the factor ridge, global boosting and the combination, whose overall loss differentials are not significant ($p$-values between 0.29 and 0.69 in Panel A). The selective $F_{SI,r}$ test rejects almost exactly where the O-EPA test does: once the estimation of the clusters is accounted for, most of the clustered predictive ability found by the naive test is a difference in overall accuracy, and the selective homogeneity test rejects with the metric of $\Xi$ only for the combination in Panel A and XGBoost in Panel B, neither of which survives the Holm adjustment, and with the shrinkage estimator after the Holm adjustment only for the no-change forecast.

Two exceptions show what the aggregate tests miss. The first is the no-change forecast. AR(1) is far more accurate on average, but the estimated clusters separate 68 of the 86 currencies, including the euro, the yen, sterling and most floating currencies, for which the no-change forecast is slightly more accurate than AR(1) (mean loss differential $-0.036$), from small groups of crawling pegs and currencies with a steady depreciation for which it is much worse: Nicaragua alone ($17.9$), Costa Rica and Rwanda ($2.2$), Ghana, Honduras and Sierra Leone ($1.0$), and twelve others ($0.24$). The window mean of AR(1) captures the drift of these currencies, which the no-change forecast ignores. The heterogeneity is significant after the Holm adjustment with the shrinkage estimator in both panels, for $K=2$ and $K=3$ with the metric of $\Xi$, and with the post-COVID dummy, where the joint and homogeneity tests both have $p<0.005$ (Section~\ref{sec:supp_app}). The rejection of EPA between AR(1) and the no-change forecast is thus driven by a handful of managed currencies, in line with the difficulty of beating the random walk for floating exchange rates \citep{meese1983empirical}.

The second is the pooled AR(1), which is as accurate as AR(1) overall (O-EPA $p=0.69$) but for which the joint selective C-EPA test rejects ($p=0.02$). Its three clusters are 79 currencies for which pooling the AR(1) coefficient improves on AR(1) ($-0.029$), the six currencies Costa Rica, Ghana, Honduras, Laos, Rwanda and Sierra Leone, for which it is less accurate ($0.285$), and Nicaragua ($0.885$). Gains for most currencies and losses for a few offset each other in the overall average. With the post-COVID dummy as testing function, the same partition is selected and the joint test gives $p=0.001$, which survives the Holm adjustment; the unit-level tests of J-EPA also reject for this method. This result is less robust than that for the no-change forecast: it holds for $K=3$ only and weakens with the $\mathbb F$ reference law and with $\varsigma=3$.

For the methods and testing functions for which the C-EPA hypothesis is rejected, Table~\ref{tab:app_clusters} of the supplement reports the estimated cluster means and the Holm-adjusted selective $p$-values of the cluster-specific tests of Section~\ref{sec:clusterspecific}, which identify the groups of currencies for which the competing method is more or less accurate than AR(1).
C-EPA is rejected by the merged or the joint selective test in 37 of the 52 combinations of method and testing function, and at least one cluster is identified in 10 of them. In the others, the rejection comes from the overall difference in accuracy, against which the cluster-specific tests have little power (Section~\ref{sec:clusterspecific}). The gains of empirical Bayes are significant for a cluster of 77 currencies ($-0.013$, Holm-adjusted $p=0.04$), while the nine remaining currencies, mostly frontier markets with large devaluations, have larger but imprecise gains ($-0.119$). The losses of XGBoost conditionally on the lagged return are significant in both of its clusters, of 76 and 10 currencies, and those of SVM for a cluster of 53 currencies with the post-COVID dummy: the losses of the machine learning methods are broad based rather than concentrated in a few currencies. With the dummies as testing functions, the cluster-specific tests isolate Nicaragua for the no-change forecast.

{\em Robustness.}\; Section~\ref{sec:supp_app} of the supplement reports further results. The conclusions above hold with the post-crisis and post-COVID dummies as testing functions. Conditioning on the paths of all initializations removes almost all rejections of the selective tests, as in the simulations, while the results for the no-change forecast are robust to the decomposition, the number of initializations and local demeaning.

\section{Conclusions}\label{sec:conc}

This paper developed tests of linear hypotheses on the cluster centers of a panel after the clusters are estimated by Panel Kmeans, and applied them to conditional C-EPA testing. To address double dipping, the use of the same data to select the clusters and to test hypotheses on them, we proposed two approaches. The first is based on selective conditional inference: its $p$-values are conditional tail probabilities, computed among the realizations of the panel that yield the same selection outcome as the realization in hand. The second is a simpler Split Sample test. We established size control and consistency of both under arbitrary CSD.

The selective tests are valid under an arbitrary long-run covariance matrix, because the decomposition of the data is taken in the metric of that matrix, and they condition only on the outcome of the procedure used in practice, including multiple initializations and the choice of the number of clusters, with truncation sets computed exactly. In our simulations, naive tests that ignore the estimation of the clusters are severely oversized, whereas both approaches control size when the relative performance of the forecasts is stable, sample splitting being then more powerful. A shrinkage estimator of the long-run covariances between the unit-level series and the tested differences narrows this difference considerably. Under a break in relative performance, only the selective tests remain valid, because the Split Sample test targets the cluster means of the testing sample. The same tests are valid for the J-EPA hypothesis, for which the clustering directs power against alternatives with groups of units that gain and lose, and the cluster-specific tests identify the clusters for which EPA fails. Increasing the power of selective inference after clustering, for instance by randomizing the selection step so that less information is lost by conditioning, is an interesting direction for future research.

In an application to thirteen time series, panel and machine learning forecasts of exchange rate returns, naive tests find clustered predictive ability for every method, whereas the selective tests find it mainly where the overall accuracy differs. They also reveal heterogeneity that the overall comparison hides: the no-change forecast is at least as accurate as AR(1) for most floating currencies, and AR(1) wins because of a handful of crawling pegs, while the gains from pooling the AR(1) coefficient for most currencies are offset by losses for largely the same managed currencies.

\medskip\noindent {\textbf{Acknowledgments.} We thank Lucy L.\ Gao, Antonio Monta\~{n}es, Ryo Okui, Hashem Pesaran, and Esther Ruiz-Ortega and the participants of the Centre for Econometric Analysis Occasional Econometrics Seminar at the Bayes Business School (London, October~2023), the Annual Spatial Econometrics Association Conference (San Diego, November~2023), the 29th International Panel Data Conference (Orl\'{e}ans, July~2024), the GIAM Seminar at Galatasaray University (Istanbul, April~2025), and the SPOC Seminar at Burgundy Institute of Mathematics (Dijon, October~2025). The usual disclaimer applies.}

\medskip\noindent {\textbf{Disclosure statement.} The authors report there are no competing interests to declare.}

\medskip\noindent {\textbf{Funding.} No funding was received for this work.}

\medskip\noindent {\textbf{Declaration of generative AI use.} During the preparation of this work, the authors used Claude (Anthropic) to assist with editing the text and with writing and checking the computer code. The authors reviewed and edited all output and take full responsibility for the content of the article.}

\medskip\noindent {\textbf{Data availability statement.} The code and data that reproduce the simulations and the empirical results are available at \url{https://github.com/akoguzhan/EPA-with-Unknown-Clusters}.}

\bibliography{ref}

\clearpage
\numberwithin{theorem}{section}
\setcounter{page}{1}
\renewcommand{\thetheorem}{\thesection.\arabic{theorem}}
\counterwithin{remark}{section}
\setcounter{section}{0}\setcounter{equation}{0}\setcounter{table}{0}\setcounter{figure}{0}
\renewcommand{\thesection}{S.\arabic{section}}
\renewcommand{\thesubsection}{S.\arabic{section}.\arabic{subsection}}
\renewcommand{\theequation}{S.\arabic{equation}}
\renewcommand{\thetable}{S.\arabic{table}}
\renewcommand{\thefigure}{S.\arabic{figure}}

\clearpage
\begin{center}
{\LARGE Supplement to `Testing Clustered Equal Predictive Ability\\[0.3em] with Unknown Clusters'\par}
\vspace{1.5em}
{\large O\u{g}uzhan Akg\"{u}n$^{1}$ \qquad Alain Pirotte$^{2}$ \qquad Giovanni Urga$^{3}$ \qquad Zhenlin Yang$^{4}$\par}
\vspace{1em}
{\small $^{1}$LEDi UR 7467, Universit\'{e} Bourgogne Europe, France. Email: oguzhan.akgun@u-bourgogne.fr\\
$^{2}$CRED, Paris-Panth\'{e}on-Assas University, France. Email: alain.pirotte@assas-universite.fr\\
$^{3}$Bayes Business School (formerly Cass), London, United Kingdom. Email: g.urga@city.ac.uk\\
$^{4}$School of Economics, Singapore Management University, Singapore. Email: zlyang@smu.edu.sg\par}
\vspace{1em}
{\large \today\par}
\end{center}
\vspace{1.5em}

\noindent In this Supplement, we present detailed derivations for the examples of the paper and remarks on the framework (Section~\ref{sec:examples_derivations}), proofs of the theoretical results (Section~\ref{sec:proofs}), the Split Sample test statistic (Section~\ref{sec:split}), details on the empirical application (Section~\ref{sec:supp_app}), the Panel Kmeans algorithm (Section~\ref{sec:supp_kmeans}), the calculation of the truncation set (Section~\ref{sec:calc_of_p}), primitive conditions for Assumption~\ref{ass:clt} of the paper (Section~\ref{sec:hdclt_primitives}), sufficient conditions for the stability of the selection outcome (Section~\ref{sec:stabilization}) and additional Monte Carlo results (Section~\ref{sec:supp_mc}).

\section{Derivations for the Examples and Remarks on the Framework}\label{sec:examples_derivations}

This section derives the loss differentials of Examples~1 and~2 of Section~\ref{sec:examples} of the paper and discusses extensions of the framework of Section~\ref{sec:frame}.

\subsection{Details on Example 1 of the main paper}

We begin by showing that, under the setup described in Example 1 of the main paper, we have
\begin{align}\label{eq:loss_diff_ex1}
&\frac{1}{|\mathcal{C}_k|}
\sum_{i \in \mathcal{C}_k}
\left\{
\mathbb{E}\left[
\left(\widehat{Y}^{(1)}_{i,R+1} - Y_{i,R+1}\right)^2
\right]
-
\mathbb{E}\left[
\left(\widehat{Y}^{(2)}_{i,R+1} - Y_{i,R+1}\right)^2
\right]
\right\} \nonumber \\
&\hspace{1cm}
=
\begin{cases}
\displaystyle
\frac{1}{|\mathcal{C}_1|}
\sum_{i \in \mathcal{C}_1}
\left\{
\Delta_i
-
\left[\alpha_i-X_{i,R}\mathbb{B}(\tilde\beta_i)\right]^2
\right\},
& k = 1, \\[2ex]
\displaystyle
\frac{1}{|\mathcal{C}_2|}
\sum_{i \in \mathcal{C}_2}
\left[
\Delta_i
-
X_{i,R}^2\mathbb{B}(\tilde\beta_i)^2
\right],
& k = 2,
\end{cases}
\end{align}
where $\Delta_i = \mathbb{V}(\hat{\alpha}_i) + [ \mathbb{V}(\hat{\beta}_i) - \mathbb{V}(\tilde{\beta}_i)] X_{i,R}^2 + 2 X_{i,R} \operatorname{Cov}(\hat{\alpha}_i, \hat{\beta}_i) + [\mathbb{B}(\hat\alpha_i)+X_{i,R}\mathbb{B}(\hat\beta_i)]^2$, with $\mathbb{B}(\cdot)$ denoting the bias of an estimator. The derivations are conditional on the regressors in the estimation window and at the forecast origin; under strict stationarity, averaging over the regressors gives expected loss differentials that do not depend on the forecast origin, which is the setting of the rolling-window framework of Section~\ref{sec:frame}. In the fixed design, the least squares estimators of correctly specified models are unbiased, so that $\mathbb{B}(\hat\alpha_i)=\mathbb{B}(\hat\beta_i)=0$, and $\mathbb{B}(\tilde\beta_i)=0$ for $i\in\mathcal C_2$.

We recall the setup of the example. We assume that $X_{i,R}$ is known and fixed at the time of forecasting. The true data-generating process is given by
\[
Y_{i,R+1} = 
\begin{cases}
\alpha_i + \beta_i X_{i,R} + U_{i,R+1}, & i \in \mathcal{C}_1, \\
\beta_i X_{i,R} + U_{i,R+1}, & i \in \mathcal{C}_2.
\end{cases}
\]
We assume here for simplicity that $U_{i,R+1} \sim iid (0, \sigma^2)$ and is independent of all other variables. The predictors $\hat{\alpha}_i$, $\hat{\beta}_i$, and $\tilde{\beta}_i$ are estimated from a fixed window of past observations and are thus random, while $X_{i,R}$ is treated as fixed.
We analyze the two clusters separately.

\paragraph*{Case 1: $i \in \mathcal{C}_1$ (True DGP with intercept).} 
In this case, Forecaster 1 correctly includes both an intercept and a slope, whereas Forecaster 2 omits the intercept and thus suffers from misspecification bias.
The one-step-ahead forecast errors can be written as
\begin{align*}
\widehat{Y}^{(1)}_{i,R+1} - Y_{i,R+1} &= (\hat{\alpha}_i - \alpha_i) + (\hat{\beta}_i - \beta_i) X_{i,R} - U_{i,R+1}, \\
\widehat{Y}^{(2)}_{i,R+1} - Y_{i,R+1} &= -\alpha_i + (\tilde{\beta}_i - \beta_i) X_{i,R} - U_{i,R+1}.
\end{align*}
The expectation of squared forecast error of Forecaster 1 is, by a bias–variance decomposition:
\begin{align*}
\mathbb{E}
\left[
\left(
\widehat{Y}^{(1)}_{i,R+1}
-
Y_{i,R+1}
\right)^2
\right]
&=
\mathbb{E}
\left[
(\hat{\alpha}_i-\alpha_i)^2
\right]
+
X_{i,R}^2
\mathbb{E}
\left[
(\hat{\beta}_i-\beta_i)^2
\right]
\\
&\quad
+
2X_{i,R}
\mathbb{E}
\left[
(\hat{\alpha}_i-\alpha_i)
(\hat{\beta}_i-\beta_i)
\right]
+
\mathbb{E}
\left[
U_{i,R+1}^2
\right]
\\
&=
\mathbb{V}(\hat{\alpha}_i)
+
\mathbb{B}(\hat{\alpha}_i)^2
+
X_{i,R}^2
\left[
\mathbb{V}(\hat{\beta}_i)
+
\mathbb{B}(\hat{\beta}_i)^2
\right]
\\
&\quad
+
2X_{i,R}
\operatorname{Cov}(\hat{\alpha}_i,\hat{\beta}_i)
{+2X_{i,R}\mathbb{B}(\hat\alpha_i)\mathbb{B}(\hat\beta_i)}
+
\sigma^2 .
\end{align*}

Now, let us turn to Forecaster 2, which omits the intercept. This model is misspecified for units in the cluster $\mathcal{C}_1$. Since $\tilde{\beta}_i$ is the OLS estimator from a regression without intercept, it absorbs some of the variation of the omitted constant. The resulting forecast error has a fixed bias term $-\alpha_i$, in addition to the slope estimation error and innovation.
Taking the expectation of its square, we have
\vspace{-1mm}\[
\mathbb{E}[(\widehat{Y}^{(2)}_{i,R+1} - Y_{i,R+1})^2] 
= \left[\alpha_i-X_{i,R}\mathbb{B}(\tilde\beta_i)\right]^2 + X_{i,R}^2 \mathbb{V}(\tilde{\beta}_i) + \sigma^2.
\vspace{-1mm}\]
Subtracting these two expressions yields the expected forecast loss differential:
\[
\mathbb{E}[(\widehat{Y}^{(1)}_{i,R+1} - Y_{i,R+1})^2] - \mathbb{E}[(\widehat{Y}^{(2)}_{i,R+1} - Y_{i,R+1})^2]
=\Delta_i-\left[\alpha_i-X_{i,R}\mathbb{B}(\tilde\beta_i)\right]^2,
\]
with $\Delta_i$ as defined above. Averaging over $i \in \mathcal{C}_1$ establishes the first line of \eqref{eq:loss_diff_ex1}.

\paragraph*{Case 2: $i \in \mathcal{C}_2$ (True DGP without intercept).} 
Here, Forecaster 2 correctly specifies the model by excluding the intercept. Forecaster 1, on the contrary, includes an unnecessary intercept term, which leads to overparameterization.
The forecast errors are:
\begin{align*}
\widehat{Y}^{(1)}_{i,R+1} - Y_{i,R+1} &= \hat{\alpha}_i + (\hat{\beta}_i - \beta_i) X_{i,R} - U_{i,R+1}, \\
\widehat{Y}^{(2)}_{i,R+1} - Y_{i,R+1} &= (\tilde{\beta}_i - \beta_i) X_{i,R} - U_{i,R+1}.
\end{align*}
Again, we compute the expected squared forecast errors under each model. For Forecaster 1, who estimates both an intercept and slope, we have
\begin{align*}
\mathbb{E}
\left[
\left(
\widehat{Y}^{(1)}_{i,R+1}
-
Y_{i,R+1}
\right)^2
\right]
&=
\mathbb{V}(\hat{\alpha}_i)
+
\mathbb{B}(\hat{\alpha}_i)^2
\\
&\quad
+
X_{i,R}^2
\left[
\mathbb{V}(\hat{\beta}_i)
+
\mathbb{B}(\hat{\beta}_i)^2
\right]
\\
&\quad
+
2X_{i,R}
\operatorname{Cov}(\hat{\alpha}_i,\hat{\beta}_i)
{+2X_{i,R}\mathbb{B}(\hat\alpha_i)\mathbb{B}(\hat\beta_i)}
+
\sigma^2 .
\end{align*}

Now, we turn to Forecaster 2, which correctly omits the intercept. The expected forecast loss is:
\[
\mathbb{E}[(\widehat{Y}^{(2)}_{i,R+1} - Y_{i,R+1})^2] 
= \mathbb{V}(\tilde{\beta}_i) X_{i,R}^2 + \mathbb{B}(\tilde{\beta}_i)^2 X_{i,R}^2 + \sigma^2.
\]
Subtracting the two, we obtain the loss differential:
\[
\mathbb{E}[(\widehat{Y}^{(1)}_{i,R+1} - Y_{i,R+1})^2]
- \mathbb{E}[(\widehat{Y}^{(2)}_{i,R+1} - Y_{i,R+1})^2]
= \Delta_i - X_{i,R}^2\mathbb{B}(\tilde{\beta}_i)^2,
\]
where $\Delta_i$ is the same as previously defined.
Averaging over $i \in \mathcal{C}_2$ yields the second line of \eqref{eq:loss_diff_ex1}, completing the derivation.

\subsection{Details on Example 2 of the main paper}

In Example~2, the DGP includes individual fixed effects, $Y_{it}=\alpha_i+\beta_{k_i}'X_{i,t-1}+U_{it}$, Forecaster~1 uses the pooled fixed-effects (within) estimator $\hat\beta$ with intercepts $\hat\alpha_i=\bar Y_i-\hat\beta'\bar X_i$, and Forecaster~2 uses unit-specific least squares with intercept, $\tilde\alpha_i=\bar Y_i-\hat\beta_i'\bar X_i$, where bars denote averages over the estimation window. Since $\bar Y_i=\alpha_i+\beta_{k_i}'\bar X_i+\bar U_i$, the fixed effects cancel from the forecast errors, which become $(\hat\beta-\beta_{k_i})'\widetilde X_{i,R}+\bar U_i-U_{i,R+1}$ and $(\hat\beta_i-\beta_{k_i})'\widetilde X_{i,R}+\bar U_i-U_{i,R+1}$, with $\widetilde X_{i,R}=X_{i,R}-\bar X_i$. The random parts of $\hat\beta$ and $\hat\beta_i$ are linear in the window errors with weights given by the within-unit demeaned regressors, which sum to zero within each unit, so that they are uncorrelated with $\bar U_i$, and the terms in $\bar U_i$ cancel in the loss differential. The derivation below therefore applies with $X_{i,R}$ replaced by $\widetilde X_{i,R}$, $\hat\beta$ the pooled within estimator, and $\Sigma_X$ the average of $\widetilde X_{i,R}\widetilde X_{i,R}'$ over the cluster. The case without fixed effects is obtained with $\alpha_i=0$ and no intercepts.

We will now show that, under the setup described in Example 2 and standard regularity conditions, the cluster-level expected loss differential is
\begin{align}
\frac{1}{|\mathcal{C}_k|}&
\sum_{i \in \mathcal{C}_k}
\left\{
\mathbb{E}
\left[
(\widehat{Y}^{\mathrm{pooled}}_{i,R+1}-Y_{i,R+1})^2
\right]
-
\mathbb{E}
\left[
(\widehat{Y}^{\mathrm{het}}_{i,R+1}-Y_{i,R+1})^2
\right]
\right\}
\notag\\
&=
\left[
\mathbb{E}(\hat{\beta})-\beta_k
\right]'
\Sigma_X
\left[
\mathbb{E}(\hat{\beta})-\beta_k
\right] +
\mathrm{tr}
\left\{
\left[
\mathbb{V}(\hat{\beta})
-
\overline{\mathbb{V}(\hat{\beta})}
\right]
\Sigma_X
\right\}.
\label{eq:loss_diff_ex2}
\end{align}
where $\Sigma_X = |\mathcal{C}_k|^{-1} \sum_{i \in \mathcal{C}_k} X_{i,R} X_{i,R}'$ is the empirical second moment matrix of the regressors in cluster $\mathcal{C}_k$, and $\overline{\mathbb{V}(\hat{\beta})} = |\mathcal{C}_k|^{-1} \sum_{i \in \mathcal{C}_k} \mathbb{V}(\hat{\beta}_i)$ is the average variance of the unit-specific estimators.

The forecast error under pooled estimation is
\[
\widehat{Y}^{\text{pooled}}_{i,R+1} - Y_{i,R+1} = (\hat{\beta} - \beta_{k_i})' X_{i,R} - U_{i,R+1}.
\]
Squaring and taking expectation:
\begin{align*}
\mathbb{E}[(\widehat{Y}^{\text{pooled}}_{i,R+1} - Y_{i,R+1})^2] 
&= \mathbb{E}\{[(\hat{\beta} - \beta_{k_i})' X_{i,R}]^2 \} + \mathbb{E}(U_{i,R+1}^2) \\
&= \mathbb{E}[(\hat{\beta} - \beta_{k_i})' X_{i,R} X_{i,R}' (\hat{\beta} - \beta_{k_i})] + \sigma^2.
\end{align*}
Using the bias–variance decomposition:
\[
\mathbb{E}[(\hat{\beta} - \beta_{k_i})(\hat{\beta} - \beta_{k_i})'] = \mathbb{V}(\hat{\beta}) + [\mathbb{E}(\hat{\beta}) - \beta_{k_i}][\mathbb{E}(\hat{\beta}) - \beta_{k_i}]',
\]
we obtain
\[
\mathbb{E}[(\widehat{Y}^{\text{pooled}}_{i,R+1} - Y_{i,R+1})^2] 
= [\mathbb{E}(\hat{\beta}) - \beta_{k_i}]' X_{i,R} X_{i,R}' [\mathbb{E}(\hat{\beta}) - \beta_{k_i}] 
+ \mathrm{tr}[\mathbb{V}(\hat{\beta}) X_{i,R} X_{i,R}'] + \sigma^2.
\]
For Forecaster 2, the forecast error is
\[
\widehat{Y}^{\text{het}}_{i,R+1} - Y_{i,R+1} = (\hat{\beta}_i - \beta_{k_i})' X_{i,R} - U_{i,R+1}.
\]
Assuming $\mathbb{E}(\hat{\beta}_i) = \beta_{k_i}$, the expected squared forecast error is
\[
\mathbb{E}[(\widehat{Y}^{\text{het}}_{i,R+1} - Y_{i,R+1})^2] 
= \mathrm{tr}[\mathbb{V}(\hat{\beta}_i) X_{i,R} X_{i,R}'] + \sigma^2.
\]
Taking the difference yields
\begin{align*}
\mathbb{E}[(\widehat{Y}^{\text{pooled}}_{i,R+1} - Y_{i,R+1})^2] 
- \mathbb{E}[(\widehat{Y}^{\text{het}}_{i,R+1} - Y_{i,R+1})^2] 
= \; & [\mathbb{E}(\hat{\beta}) - \beta_{k_i}]' X_{i,R} X_{i,R}' [\mathbb{E}(\hat{\beta}) - \beta_{k_i}] \\
&+ \mathrm{tr} \{ [\mathbb{V}(\hat{\beta}) - \mathbb{V}(\hat{\beta}_i)] X_{i,R} X_{i,R}'\}.
\end{align*}
Letting
$\Sigma_X = \frac{1}{|\mathcal{C}_k|} \sum_{i \in \mathcal{C}_k} X_{i,R} X_{i,R}'$ and 
$\overline{\mathbb{V}(\hat{\beta}_i)} = \frac{1}{|\mathcal{C}_k|} \sum_{i \in \mathcal{C}_k} \mathbb{V}(\hat{\beta}_i)$,
we have
\begin{equation*}
\begin{split}
\frac{1}{|\mathcal{C}_k|} &\sum_{i \in \mathcal{C}_k} 
\{ \mathbb{E}[(\widehat{Y}^{\text{pooled}}_{i,R+1} - Y_{i,R+1})^2] 
- \mathbb{E}[(\widehat{Y}^{\text{het}}_{i,R+1} - Y_{i,R+1})^2] \} \\
&= [\mathbb{E}(\hat{\beta}) - \beta_k]' \Sigma_X [\mathbb{E}(\hat{\beta}) - \beta_k] 
+ \mathrm{tr}\{ [\mathbb{V}(\hat{\beta}) - \overline{\mathbb{V}(\hat{\beta}_i)}] \Sigma_X \},
\end{split}
\end{equation*}
noting that $\beta_{k_i} = \beta_k$ for all $i \in \mathcal{C}_k$, which establishes the claimed result.

\subsection{Remarks on the framework}\label{sec:supp_framework}

\begin{remark}\normalfont
The two conditional schemes of Section~\ref{sec:frame} of the paper need not be mutually exclusive.
In practice, forecast errors can be estimated from alternative panel data models that include both external predictors and observed or estimated common factors. Furthermore, model errors may also exhibit spatial or network dependence. Such models feature external predictors, strong CSD through common factors and weak CSD through spatial interactions \citep[see][]{chudik2011weak}, which the loss differentials inherit. While our current theoretical framework accommodates general forms of CSD, formal treatment of parametric structures of cross-dependence in loss differentials could lead to more powerful inference (see APUY for further insights on the distinction of weak vs. strong CSD in EPA testing).
\end{remark}

\begin{remark}\normalfont
In some cases, one may want to compare a large number of forecasts made for a given unit with a base forecast. For example, in comparing the inflation forecasts of the survey of professional forecasters with that of the IMF for the Euro area, the framework remains similar but the meaning of the indices changes: $Y_{t}$ denotes the Euro area inflation, $\widehat{Y}_{1,it}$ is the forecast of the $i$-th forecaster for period $t$, and $\widehat{Y}_{2,t}$ is the IMF forecast for period $t$.
The loss differentials are then $\Delta L_{it} = L(\widehat{Y}_{1,it},Y_{t}) - L(\widehat{Y}_{2,t},Y_{t})$ which still depend on the two indexes $i$ and $t$. Our framework applies whenever its assumptions on the loss differentials are satisfied.
\end{remark}

The following remarks concern the $p$-merging function of Section~\ref{sec:selective} of the paper.

\begin{remark}\normalfont
Unlike Fisher's method \citep{fisher32}, which assumes independence, or Bonferroni's $p$-merging function, which is conservative, the merging function $F_{homo,r}$ maintains optimal Type I control under general dependence structures.
\end{remark}

\begin{remark}\normalfont
A similar $p$-merging function was recently used by \cite{spreng23} in a multiple forecast comparison setting. The difference lies in the constant $b_{r,n_p}$: while \cite{spreng23} sets $b_{r,n_p} = r/(r+1)$, we use the constant of Proposition~5 of \citet{vovk20}, $b_{r,n_p} = [r/(r+1)] n_p^{1+1/r}$, which we found to yield smaller size distortions in our framework with a small number of $p$-values combined.
\end{remark}

\section{Proofs of the Results in the Main Text}\label{sec:proofs}

\paragraph*{Proof of Lemma \ref{lemma:clt}}

To prove Part \ref{lemma:clt:lln}, it suffices to show that each component satisfies $\hat{\theta}_{k}(\mathcal{C}) - \theta^0_{k}(\mathcal{C}) = o_p(1)$. Since $Z_{it} = \mu_i^0 + V_{it}$, $\mathbb{E}(V_{it}) = 0$, $\hat{\theta}_{k}(\mathcal{C}) = (|\mathcal{C}_k|T)^{-1} \sum_{i \in \mathcal{C}_k} \sum_{t=1}^T Z_{it}$, and $\theta^{0}_{k}(\mathcal{C}) = |\mathcal{C}_k|^{-1} \sum_{i \in \mathcal{C}_k} \mu_i^0$,
\[
\hat{\theta}_{k}(\mathcal{C}) - \theta^0_{k}(\mathcal{C})
=
\frac{1}{|\mathcal{C}_k|T}
\sum_{i \in \mathcal{C}_k}
\sum_{t=1}^T V_{it}.
\]
Hence $\mathbb{E}[\hat{\theta}_{k}(\mathcal{C}) - \theta^0_{k}(\mathcal{C})] = 0$.
Let
\[
a_k
=
\frac{\mathbf 1\{i\in\mathcal C_k\}_{i=1}^N}
{|\mathcal C_k|}
\in\mathbb R^N.
\]
Then
\[
\hat\theta_k(\mathcal C)-\theta_k^0(\mathcal C)
=
(a_k\otimes I_P)'\bar V
\]
and hence, by Jensen's inequality and Assumption~\ref{ass:alternative}\ref{ass:weakdepts},
\[
\mathbb E\left\lVert\hat\theta_k(\mathcal C)-\theta_k^0(\mathcal C)\right\rVert^2
\leq
\frac{1}{|\mathcal C_k|}\sum_{i\in\mathcal C_k}\mathbb E\lVert\bar V_i\rVert^2
\leq
\frac{1}{|\mathcal C_k|}\sum_{i\in\mathcal C_k}\frac{1}{T^2}\sum_{t,s=1}^T\mathbb E\lVert V_{it}V_{is}'\rVert
=
O(T^{-1}),
\]
and Chebyshev's inequality proves Part~\ref{lemma:clt:lln}.

For Part~\ref{lemma:clt:clt}, observe that
$
\sqrt{T}
\bigl[
\hat{\theta}_k(\mathcal{C})
-
\hat{\theta}_g(\mathcal{C})
-
\theta^0_k(\mathcal{C})
+
\theta^0_g(\mathcal{C})
\bigr]
=
(\delta^N_{k,g} \otimes I_P)' \sqrt{T}\bar{V},
$
where $\delta^N_{k,g} \in \mathbb{R}^N$ is the pairwise difference weight vector with $(\delta^N_{k,g})_i = \mathbf{1}\{i \in \mathcal{C}_k\}/|\mathcal{C}_k| - \mathbf{1}\{i \in \mathcal{C}_g\}/|\mathcal{C}_g|$, and $\bar{V} = (\bar{V}_1',\dots,\bar{V}_N')'$ is the $NP$-vector of time-averaged demeaned observations.

Let $L_{k,g} = \delta^N_{k,g} \otimes I_P$. Assumption~\ref{ass:clt} with $r=1$ gives $\sqrt T\bar V\overset{d}{\longrightarrow}G\sim\mathcal N(\mathbf 0_{NP},\Xi)$, and the continuous mapping theorem gives $L_{k,g}'\sqrt T\bar V\overset{d}{\longrightarrow}L_{k,g}'G$.
Here
$
L_{k,g}'G
\sim
\mathcal{N}\bigl(0,\Sigma_{k,g}(\mathcal{C})\bigr),
$
and
$\Sigma_{k,g}(\mathcal{C})
=
L_{k,g}'\Xi L_{k,g}.
$
Since $\Xi$ is positive definite by Assumption~\ref{ass:clt} and $L_{k,g}$ has full column rank, $\Sigma_{k,g}(\mathcal C)$ is positive definite. The relative consistency condition in Assumption~\ref{ass:clt} and Slutsky's theorem then yield
\[
\widehat{\Sigma}_{k,g}(\mathcal{C})^{-1/2}
\sqrt{T}
\bigl[
\hat{\theta}_k(\mathcal{C})
-
\hat{\theta}_g(\mathcal{C})
-
\theta^0_k(\mathcal{C})
+
\theta^0_g(\mathcal{C})
\bigr]
\overset{d}{\longrightarrow}
\mathcal{N}(0,I_P),
\]
which proves Part~\ref{lemma:clt:clt}.

For Part~\ref{lemma:clt:joint}, write $\sqrt T[\hat\theta(\mathcal C)-\theta^0(\mathcal C)]=\mathsf A(\mathcal C)'\sqrt T\bar V$, where the $(i,k)$ block of the $NP\times KP$ matrix $\mathsf A(\mathcal C)$ is $\mathbf 1\{i\in\mathcal C_k\}|\mathcal C_k|^{-1}I_P$, 
so that $\Omega(\mathcal C)\to\mathsf A(\mathcal C)'\Xi\mathsf A(\mathcal C)$. Let $L=\mathsf A(\mathcal C)\mathsf R'$. The columns of $\mathsf A(\mathcal C)$ have disjoint supports, so $L$ has full column rank $q$, and $L'\Xi L$, the limit of $\mathsf R\Omega(\mathcal C)\mathsf R'$, is positive definite because $\Xi$ is. As for Part~\ref{lemma:clt:clt}, Assumption~\ref{ass:clt} and the continuous mapping theorem give the result, since $(L'\Xi L)^{-1/2}L'G\sim\mathcal N(0,I_q)$.

\paragraph*{Proof of Lemma \ref{lemma:perturbation}}

Let $\mathcal C$ be a partition of the panel units with $K\geq2$, and let $\nu_{k,g}$ be the associated $NT\times 1$ vector. Recall that
$\nu_{k,g}=(\nu'_{k,g,1},\dots,\nu'_{k,g,N})'$, $\nu_{k,g,i}=\iota_T\delta_{k,g,i}$, where $\iota_T$ is a $T$-vector of ones and
$
\delta_{k,g,i}
=
\mathbf 1\{k_i=k\} / |\mathcal C_k|
-
\mathbf 1\{k_i=g\} / |\mathcal C_g|.
$
As in the main text, $\Pi_{k,g}=I-\nu_{k,g}\nu_{k,g}'/\lVert\nu_{k,g}\rVert^2$. The following lemmas are used below.

\begin{lemma}\normalfont\label{lemma:average_equivalence}
Let
$
\bar Z_i
=
\frac{1}{T}\sum_{t=1}^T Z_{it}
$.
For any center $\theta\in\mathbb R^P$,
$
\sum_{t=1}^T
\lVert Z_{it}-\theta\rVert^2
=
\sum_{t=1}^T
\lVert Z_{it}-\bar Z_i\rVert^2
+
T\lVert\bar Z_i-\theta\rVert^2.
$
Consequently, conditional on the same initialization and the same tie-breaking rule, every assignment and center update produced by Algorithm~\ref{algo} is identical to the corresponding update of Lloyd's Kmeans algorithm applied to $\bar Z_1,\dots,\bar Z_N$. In particular, the complete algorithmic path and the final partition are measurable functions of
$\bar Z=(\bar Z_1',\dots,\bar Z_N')'$.
\end{lemma}

\begin{proof}
For any $\theta\in\mathbb R^P$,
$
Z_{it}-\theta
=
Z_{it}-\bar Z_i+\bar Z_i-\theta
$.
Expanding the squared norm and summing over $t$ gives
$
\sum_{t=1}^T
\lVert Z_{it}-\theta\rVert^2
=
\sum_{t=1}^T
\lVert Z_{it}-\bar Z_i\rVert^2
+
T\lVert\bar Z_i-\theta\rVert^2
$
because
$
\sum_{t=1}^T(Z_{it}-\bar Z_i)=0.
$
The first term does not depend on the candidate cluster. Hence the assignment minimizing the Panel Kmeans criterion is exactly the assignment minimizing $\lVert\bar Z_i-\theta_k\rVert^2$. Moreover, for any partition $\mathcal C$,
$
(|\mathcal C_k|T)^{-1}
\sum_{i\in\mathcal C_k}
\sum_{t=1}^T Z_{it}
=
(|\mathcal C_k|)^{-1}
\sum_{i\in\mathcal C_k}\bar Z_i.
$
Thus the center updates are also identical. The result follows inductively over the iterations of the two algorithms.

Finally, the perturbation \eqref{eq:perturbed} shifts all observations of a unit by the same amount, so it leaves the within-unit deviations unchanged and modifies only the time averages. Thus the selection outcome along the perturbation is the one obtained from Lloyd's algorithm applied to the perturbed time averages, with the information criterion computed from the same within-unit deviations.

\end{proof}

Lemma~\ref{lemma:average_equivalence} is used for the computation of the truncation set and in the proofs. The test and its perturbation remain formulated for Panel Kmeans in the original panel space, while the selection-relevant probabilistic argument can be conducted using the unit averages.

\begin{lemma}\normalfont\label{lemma:independence}
Let $X\sim\mathcal N(m,\Xi)$ with $\Xi$ positive definite, $L=L_{k,g}$ for a fixed partition $\mathcal C$, $\Sigma=L'\Xi L$, and suppose $L'm=0$. Let $U=X-\Xi L\Sigma^{-1}L'X$, $D=\lVert\Sigma^{-1/2}L'X\rVert$ and $J=\mathrm{dir}(\Sigma^{-1/2}L'X)$. Then $U$, $D$ and $J$ are mutually independent, and $D\sim\chi_P$.
\end{lemma}

\begin{proof}
$L'X$ and $U$ are affine in $X$, hence jointly Gaussian, and $\mathrm{Cov}(L'X,U)=L'\Xi-L'\Xi L\Sigma^{-1}L'\Xi=0$, so $\Sigma^{-1/2}L'X$ is independent of $U$. Under $L'm=0$, $\Sigma^{-1/2}L'X\sim\mathcal N(0,I_P)$, whose norm and direction are independent, with norm distributed as $\chi_P$. No restriction on $\Xi$ is needed; with the Euclidean residual $(I-L(L'L)^{-1}L')X$ used by \cite{chen23}, the covariance $L'\Xi[I-L(L'L)^{-1}L']$ vanishes only if the column space of $L$ is invariant under $\Xi$.
\end{proof}

The proof of Lemma \ref{lemma:perturbation} follows lines similar to the proof of Theorem 1 of \cite{gao24} and Proposition 1 of \cite{chen23}. First, writing $\eta=\sqrt TL_{k,g}'\bar Z$ and dropping the dependence on $\mathcal C$,
\begin{equation}\label{eq:zrewritten}
\widehat U_{k,g}+D_{k,g}\,\widehat\Xi L_{k,g}\widehat\Sigma^{-1/2}_{k,g}\mathrm{dir}(\widehat\Sigma^{-1/2}_{k,g}\eta)
=\widehat U_{k,g}+\widehat\Xi L_{k,g}\widehat\Sigma^{-1/2}_{k,g}\widehat\Sigma^{-1/2}_{k,g}\eta
=\sqrt T\bar Z,
\end{equation}
which is \eqref{eq:zrewritten_inthetext}, valid at every realization with $\eta\neq0$.

Next, consider the Gaussian limit experiment, in which $\sqrt T\bar Z$ is replaced by $X=\sqrt T\mu^0+G$ with $G\sim\mathcal N(0,\Xi)$, and $\widehat\Xi$ by $\Xi$. Let $U=X-\Xi L\Sigma^{-1}L'X$, $D=\lVert\Sigma^{-1/2}L'X\rVert$, $J=\mathrm{dir}(\Sigma^{-1/2}L'X)$ and $v(j)=\Xi L\Sigma^{-1/2}j$, so that $X=U+D\,v(J)$. Under $\mathcal H^{k,g}_0$, $L'\mu^0=0$, and by Lemma~\ref{lemma:independence}, $D\sim\chi_P$ independently of $(U,J)$; the distributional statement is exact in this experiment. For $(u,j)$, let
\[
\mathcal T(u,j)=\left\{\phi\geq0:\hat s\big(u+\phi\,v(j)\big)=s\right\},
\]
where $s$ is the observed selection outcome and the selection outcome is computed from the time averages $(u+\phi v(j))/\sqrt T$ and the within-unit deviations, which the perturbation leaves unchanged (Lemma~\ref{lemma:average_equivalence}). Then $\{\hat s(X)=s\}=\{D\in\mathcal T(U,J)\}$, and, conditional on $U=u$, $J=j$ and $\hat s(X)=s$, $D$ has the $\chi_P$ distribution truncated to $\mathcal T(u,j)$. Therefore
\[
p_\infty[d_{k,g}(\widehat{\mathcal C})]=\Pr\left[D\geq d_{k,g}(\widehat{\mathcal C})\mid\mathcal A^G\right]=1-F_{\chi_P}[d_{k,g}(\widehat{\mathcal C});\mathcal T(u,j)],
\]
and, evaluated at the observed $(\hat u_{k,g},\hat\jmath)$ with $\widehat\Xi$ in place of $\Xi$, $\mathcal T(\hat u_{k,g},\hat\jmath)$ is the set \eqref{eq:set_s} with the perturbation \eqref{eq:perturbed}. This completes the proof.

\paragraph*{Proof of Proposition \ref{proposition:pairwise_p}}

\begin{lemma}\normalfont\label{lemma:wald_diverges}
Let $\mathcal C$ be a fixed partition, suppose that Assumptions~\ref{ass:alternative} and \ref{ass:clt} hold, and that $\lVert\theta^0_k(\mathcal C)-\theta^0_g(\mathcal C)\rVert\geq c_0$ for some $c_0>0$ and all $T$. Then there is $c>0$ such that $\Pr\{D_{k,g}(\mathcal C)\geq c\sqrt T\}\to1$ as $B\to\infty$ and $T\to\infty$ with $B/T\to0$.
\end{lemma}

\begin{proof}
By definition, $T^{-1/2}D_{k,g}(\mathcal C)=\lVert\widehat\Sigma_{k,g}(\mathcal C)^{-1/2}[\hat\theta_k(\mathcal C)-\hat\theta_g(\mathcal C)]\rVert\geq\lambda_{\max}[\widehat\Sigma_{k,g}(\mathcal C)]^{-1/2}\lVert\hat\theta_k(\mathcal C)-\hat\theta_g(\mathcal C)\rVert$. By Lemma~\ref{lemma:clt}\ref{lemma:clt:lln}, $\lVert\hat\theta_k(\mathcal C)-\hat\theta_g(\mathcal C)\rVert\geq c_0/2$ with probability approaching one. Since $L_{k,g}'\sqrt T\bar V=\sum_i(\delta^N_{k,g})_i\sqrt T\bar V_i$ with $\sum_i|(\delta^N_{k,g})_i|=2$, Minkowski's inequality and Assumption~\ref{ass:alternative}\ref{ass:weakdepts} give $\lambda_{\max}[\Sigma_{k,g}(\mathcal C)]\leq\mathrm{tr}[\Sigma_{k,g}(\mathcal C)]\leq4\max_i\mathbb E\lVert\sqrt T\bar V_i\rVert^2\leq C$, and the relative consistency in Assumption~\ref{ass:clt} gives $\lambda_{\max}[\widehat\Sigma_{k,g}(\mathcal C)]\leq2C$ with probability approaching one. The result holds with $c=c_0/(2\sqrt{2C})$.
\end{proof}

We first prove Part (a), in two steps. Step~1. In the experiment of the proof of Lemma~\ref{lemma:perturbation}, let $\nu(u,j)=\int_{\mathcal T(u,j)}dF_{\chi_P}$ and $p^G=1-F_{\chi_P}[D;\mathcal T(U,J)]$. Conditional on $(U,J)=(u,j)$ with $\nu(u,j)>0$ and on $D\in\mathcal T(u,j)$, $p^G$ is the survival function of a continuous distribution evaluated at a draw from that distribution, hence uniform on $(0,1)$. Thus $\Pr\{p^G\leq\alpha,\hat s(X)=s\mid U,J\}=\alpha\,\nu(U,J)$, which also holds when $\nu(U,J)=0$, and, taking expectations, $\Pr\{p^G\leq\alpha\mid\hat s(X)=s\}=\alpha$ whenever $\Pr\{\hat s(X)=s\}>0$. This holds for any selection map, in particular for the selected partition with $\widehat K_{IC}$, given the initializations, and for the finer selection outcome that also includes the paths of all initializations. It also gives $\Pr\{\nu(U,J)=0,\hat s(X)=s\}=\mathbb E[\mathbf 1\{\nu(U,J)=0\}\nu(U,J)]=0$, so that the truncated distribution is well defined almost surely on the selection event.

Step~2. Under Assumption~\ref{ass:stability}, write $\sqrt T\bar Z=\sqrt T\mu^*+Y_T$ with $Y_T=h+\sqrt T(\bar Z-\mu^0)\overset{d}{\longrightarrow}Y\sim\mathcal N(h,\Xi)$ by Assumption~\ref{ass:clt} with $N$ fixed. Since $L'\mu^0=0$ for every $T$, $L'\mu^*=0$ and $L'h=0$. Hence $D_{k,g}(\widehat{\mathcal C})$ and the direction depend on $Y_T$ only, and $\widehat U_{k,g}=\sqrt T\mu^*+\widetilde U_T$ with $\widetilde U_T=Y_T-\widehat\Xi L\widehat\Sigma^{-1}L'Y_T$, where the term $\sqrt T\mu^*$ appears in every point of the perturbation path. The pair $(\mathbf 1\{\hat s=s\},p[d_{k,g}(\widehat{\mathcal C})])$ is therefore a function $\Psi_T(Y_T,\widehat\Xi,\widehat\Gamma_T)$, with $\widehat\Gamma_T$ the within-unit second moment matrix used by the information criterion. Let $\Psi_\infty(y)=(\mathbf 1\{\hat s_\infty(y)=s\},p_\infty(y))$, where $p_\infty$ is the $p$-value of Step~1 for the selection map $\hat s_\infty$. Let $\mathbb N_0$ be the null set on which the convergence in Assumption~\ref{ass:stability} fails. The map $y\mapsto(\tilde u(y),\Sigma^{-1/2}L'y)$ is a linear bijection, and polar coordinates preserve null sets, so by Fubini's theorem, for almost every $y$, the set of $\phi\geq0$ for which $\tilde u(y)+\phi v(j(y))\in\mathbb N_0$ is Lebesgue null. For such $y$ and all $(y_T,\Xi_T,\Gamma_T)\to(y,\Xi,\Gamma)$, Assumption~\ref{ass:stability} gives $\mathbf 1\{\hat s=s\}\to\mathbf 1\{\hat s_\infty(y)=s\}$ and $\mathbf 1\{\phi\in\mathcal T_T\}\to\mathbf 1\{\phi\in\mathcal T_\infty(y)\}$ for almost every $\phi$, so that, by dominated convergence and $D_T\to D(y)$, $\int_{\mathcal T_T\cap[D_T,\infty)}dF_{\chi_P}\to\int_{\mathcal T_\infty(y)\cap[D(y),\infty)}dF_{\chi_P}$ and $\int_{\mathcal T_T}dF_{\chi_P}\to\int_{\mathcal T_\infty(y)}dF_{\chi_P}$, the latter being positive almost surely on $\{\hat s_\infty(Y)=s\}$ by Step~1, since the Gaussian law of $Y$ and Lebesgue measure have the same null sets. Since $(Y_T,\widehat\Xi,\widehat\Gamma_T)\overset{d}{\longrightarrow}(Y,\Xi,\Gamma)$, the extended continuous mapping theorem \citep[Theorem~18.11]{vaart98} gives $\Psi_T(Y_T,\widehat\Xi,\widehat\Gamma_T)\overset{d}{\longrightarrow}\Psi_\infty(Y)$. By Step~1, $p_\infty(Y)$ is uniform conditional on $\{\hat s_\infty(Y)=s\}$, so $\Pr\{p_\infty(Y)=\alpha,\hat s_\infty(Y)=s\}=0$, and therefore $\Pr\{p\leq\alpha,\hat s=s\}\to\alpha\Pr\{\hat s_\infty(Y)=s\}$ while $\Pr\{\hat s=s\}\to\Pr\{\hat s_\infty(Y)=s\}>0$. This proves Part (a). No conditioning on events of probability zero is used outside the Gaussian limit experiment.

For Part (b), with $h=0$ and $\mu^*$ a regular point, $\bar Z\overset{p}{\longrightarrow}\mu^*$ gives $\Pr(\hat s=s)\to1$. Under $\mathcal H^{k,g}_1$ for $\mathcal C(s)$, Lemma~\ref{lemma:wald_diverges} applied to $\mathcal C(s)$, which equals $\widehat{\mathcal C}$ with probability approaching one, gives $D_{k,g}(\widehat{\mathcal C})\geq c\sqrt T$ with probability approaching one for some $c>0$. For $|\phi-D_{k,g}|\leq\varepsilon\sqrt T$, the perturbed time averages lie within $\varepsilon\lVert\widehat\Xi L\widehat\Sigma^{-1/2}\rVert+\lVert\bar Z-\mu^*\rVert$ of $\mu^*$, so, by regularity and for $\varepsilon$ small enough, $[D_{k,g}-\varepsilon\sqrt T,D_{k,g}]\subseteq\mathcal T$ with probability approaching one. With $\bar F$ the $\chi_P$ survival function and $\varepsilon<c$,
\[
p[d_{k,g}(\widehat{\mathcal C})]\leq\frac{\bar F(D_{k,g})}{\bar F(D_{k,g}-\varepsilon\sqrt T)-\bar F(D_{k,g})}\overset{p}{\longrightarrow}0,
\]
because $\bar F(d)/\bar F(d-\varepsilon\sqrt T)\leq Cd^P\exp(-\varepsilon cT/2)\to0$ for $d\geq c\sqrt T$. Divergence of the statistic alone is not sufficient, since the truncation set could concentrate around the observed value. This proves Part (b).

\paragraph*{Proof of Theorem \ref{theorem:homo}}

\begin{lemma}\normalfont\label{lemma:portmanteau}
    Let $G_{T} = (G_{1,T},\dots,G_{n,T})'$ be a random vector in $[0,1]^n$ such that $G_{T} \overset{d}{\longrightarrow} G$ as $T \to \infty$.
    Define
    $
    \mathcal{R}_{\alpha} = \left\lbrace (x_1,\dots,x_n) \in [0,1]^{n} : F(x_1,\dots,x_n) \leq \alpha \right\rbrace
    $
    for all $\alpha \in (0,1)$, where
    $
    F(x_1,\dots,x_n) = f(x_1,\dots,x_n) \wedge 1
    $
    for some continuous function $f : [0,1]^n \to \mathbb{R}$.
    Then
    $
    \limsup_{T \to \infty} \Pr \left( G_{T} \in \mathcal{R}_{\alpha} \right) \leq \Pr \left( G \in \mathcal{R}_{\alpha} \right).
    $
\end{lemma}

\begin{proof}
    Since $f$ is continuous, $F = f \wedge 1$ is continuous, so that $\mathcal{R}_\alpha$ is a closed subset of $[0,1]^n$. The result follows from the Portmanteau lemma \citep[Lemma~2.2]{vaart98}.
\end{proof}

Define $p^*[D_{k,g}(\widehat{\mathcal{C}})]$ as the limit of the random variable $p[d_{k,g}(\widehat{\mathcal{C}})]$ which satisfies $p[d_{k,g}(\widehat{\mathcal{C}})] \overset{d}{\longrightarrow} p^*[D_{k,g}(\widehat{\mathcal{C}})] \sim \mathbb{U}[0,1]$ as $T\to \infty$ for all $k,g \in \{1,\dots,K\}$, $k \neq g$, which holds by Proposition \ref{proposition:pairwise_p}\ref{proposition:pairwise_p:partb}, conditional on the selection outcome $\hat s=s$. Since every pairwise $p$-value is a function of $(Y_T,\widehat\Xi,\widehat\Gamma_T)$, the argument of Step~2 in the proof of Proposition~\ref{proposition:pairwise_p} applies to the vector of pairwise $p$-values jointly, which converges conditionally on $\hat s=s$ to a vector of $p$-values that are each uniform conditionally on $\hat s_\infty(Y)=s$ under $\mathcal H^{homo}_0$.
By Proposition~5 of \citet{vovk20}, which establishes the validity of the merging function under arbitrary dependence (see also Theorem~1 of \citet{spreng23}), we have
\begin{equation*}
    \Pr \left[ \frac{r}{r+1} n_p^{1+1/r} \left\lbrace \frac{1}{n_p} \sum_{(k,g)\in\mathcal P_K} \{ p^*[D_{k,g}(\widehat{\mathcal{C}})] \}^r \right\rbrace^{1/r} \leq \alpha \right] \leq \alpha.
\end{equation*}
Part~\ref{theorem:homo:parta} then follows from Lemma~\ref{lemma:portmanteau}, applied conditionally on $\hat s=s$.

For Part~\ref{theorem:homo:partb}, Proposition~\ref{proposition:pairwise_p}\ref{proposition:pairwise_p:partc} gives $p[D_{k,g}(\widehat{\mathcal C})] \overset{p}{\longrightarrow} 0$ for at least one pair $k \neq g$. Since $r<0$, $n_p^{-1}\sum_{(k,g)\in\mathcal P_K}p_{k,g}^r\geq n_p^{-1}(\min_{k,g}p_{k,g})^r$, so that $F_{homo,r}\leq[r/(r+1)]\,n_p\min_{k,g}p_{k,g}\overset{p}{\longrightarrow}0$.

\paragraph*{Proof of Proposition \ref{proposition:overall}}

Part~\ref{proposition:overall:parta} follows from Theorem~3.1 of \cite{sun13} under our Assumptions \ref{ass:alternative} and \ref{ass:clt} by setting $\mathcal{C} = (1,\dots,1)'$, since Assumption~\ref{ass:clt} implies the functional central limit theorem for the partial sums of the cross-sectional average of $V_{it}$ required there.
Part \ref{proposition:overall:partb} follows from Section 4.1 of \cite{sun11} under the same assumptions. 
When $B\to\infty$ with $B/T\to0$, $\widehat\Omega_o$ is consistent \citep{sun13}, so $PB(B-P+1)^{-1}W_{oepa}\overset{d}{\longrightarrow}\chi^2_P$, and $\mathbb F_{P,B-P+1}$ converges to $\chi^2_P/P$.

\paragraph*{Proof of Theorem \ref{theorem:main}}

For Part (a), let $E$ be the event that $\mathcal H'_0$ holds for $\widehat{\mathcal C}$. On $E$, the O-EPA hypothesis holds, since its parameter does not depend on the partition, and $\mathcal H^{k,g}_0$ holds for every pair. Let $\tilde p_j=p_j$ on $E$ and $\tilde p_j=1$ otherwise, for each of the $n_p+1$ component $p$-values. For the pairwise components, $\Pr(\tilde p_j\leq x)=\sum_s\Pr(p_j\leq x\mid\hat s=s)\Pr(\hat s=s)$, the sum being over the selection outcomes for which the pairwise hypothesis holds, so $\limsup\Pr(\tilde p_j\leq x)\leq x$ by Proposition~\ref{proposition:pairwise_p}\ref{proposition:pairwise_p:partb}; for the O-EPA component, the same bound follows from Proposition~\ref{proposition:overall} when O-EPA holds, and trivially otherwise. Since the merging function is nondecreasing in each argument and valid under arbitrary dependence \citep[Proposition~5]{vovk20}, and $F_{SI,r}(p)=F_{SI,r}(\tilde p)$ on $E$, $\Pr(F_{SI,r}\leq\alpha,E)\leq\Pr\{F_{SI,r}(\tilde p)\leq\alpha\}$, whose limit superior is at most $\alpha$ by Lemma~\ref{lemma:portmanteau}. The O-EPA $p$-value is not uniform conditional on $\hat s$, which is why the statement is not conditional on the selection outcome.

Similarly, Part (b) follows because, under the alternative, either $p[D_{k,g}(\widehat{\mathcal C})] \overset{p}{\longrightarrow} 0$ for at least one pair $k\neq g$ by Proposition~\ref{proposition:pairwise_p}\ref{proposition:pairwise_p:partc}, or the O-EPA test statistic diverges.

\paragraph*{Proof of Proposition \ref{proposition:bic}}

Consider the mapping $Z \mapsto \widehat{\mathcal{C}}$ where $Z$ is the input of Algorithm \ref{algo} and $\widehat{\mathcal{C}}$ is the partition of the panel units output by the algorithm.
Notice that $Z \mapsto \widehat{\mathcal{C}}$ is the composition of two deterministic procedures:
\begin{enumerate*}
\item selection of the number of clusters $\widehat{K}_{IC}$ via the minimization of $IC(K)$ in \eqref{eq:kmeans_khat}, and
\item estimation of the clustering assignment $\widehat{\mathcal{C}}$ by solving the Panel Kmeans problem \eqref{eq:kmeans} with $K = \widehat{K}_{IC}$.
\end{enumerate*}
Since both steps are deterministic functions of the data, the composite map $Z \mapsto \widehat{\mathcal{C}}$ is itself deterministic.

Now fix a particular realization $\mathcal{C}^\ast$ of the clustering.
The number of clusters in $\mathcal{C}^\ast$ is fixed.
Denote this number by $K^\ast$.
Then,
$
\{ \widehat{\mathcal{C}} = \mathcal{C}^\ast \} \subseteq \{ \widehat{K}_{IC} = K^\ast \},
$
since the selected partition is computed with $\widehat K_{IC}$ clusters. Hence, conditioning on the event $\{ \widehat{\mathcal{C}} = \mathcal{C}^\ast \}$ implicitly restricts us to the subset of the sample space where $\widehat{K}_{IC} = K^\ast$. This yields
\[
\Pr\left[D_{k,g}(\widehat{\mathcal{C}}) \in \mathcal{T} \;\middle|\; \widehat{\mathcal{C}} = \mathcal{C}^\ast \right] = \Pr\left[D_{k,g}(\widehat{\mathcal{C}}) \in \mathcal{T} \;\middle|\; \widehat{K}_{IC} = K^\ast,\; \widehat{\mathcal{C}} = \mathcal{C}^\ast \right],
\]
as claimed.

The event $\{\widehat{\mathcal C}=\mathcal C^\ast\}$ is smaller than the event that Algorithm~\ref{algo} with $K^\ast$ clusters returns $\mathcal C^\ast$, so its truncation set involves the selection of $\widehat K_{IC}$. Given the partitions retained for each $K\in\mathcal K$, the matrix inside the determinant of $IC(K)$ equals $\widehat\Gamma_w+N^{-1}\sum_i(\bar z_i(\phi)-\bar z_{\hat k_i}(\phi))(\bar z_i(\phi)-\bar z_{\hat k_i}(\phi))'$, where $\widehat\Gamma_w=(NT)^{-1}\sum_{i,t}(z_{it}-\bar z_i)(z_{it}-\bar z_i)'$ is not affected by the perturbation and $\bar z_{k}(\phi)$ is the average of $\bar z_i(\phi)$ over cluster $k$. Since $\bar z(\phi)$ is affine in $\phi$, this matrix is a quadratic matrix polynomial in $\phi$, and its determinant is a polynomial of degree at most $2P$, which is positive because $\widehat\Gamma_w$ is positive definite. Since the logarithm is increasing, $IC(\widehat K_{IC})\leq IC(K)$, with strict inequality for $K<\widehat K_{IC}$ by the tie-breaking rule, is equivalent to a polynomial inequality of degree at most $2P$. The assignments of the algorithm and the comparisons between initializations are defined by quadratic inequalities in $\phi$ (Section~\ref{sec:calc_of_p}), and each initialization changes its path finitely many times along $\phi\geq0$. The truncation set is therefore a finite union of intervals, and the proofs of Lemma~\ref{lemma:perturbation} and Proposition~\ref{proposition:pairwise_p} apply to the selection outcome that includes $\widehat K_{IC}$, since they hold for any selection map satisfying Assumption~\ref{ass:stability}.

\paragraph*{Consistency of the information criterion}

\begin{lemma}\normalfont\label{lemma:ic_consistent}
Let $N$ be fixed, suppose Assumptions~\ref{ass:alternative}--\ref{ass:separation} hold with $K^0\in\mathcal K$, and that Panel Kmeans attains the global minimum of \eqref{eq:kmeans} for each $K\in\mathcal K$ with probability approaching one. Then $\Pr(\widehat K_{IC}=K^0)\to1$.
\end{lemma}

\begin{proof}
Write $M_K=N^{-1}\sum_i(\bar Z_i-\hat\theta_{\hat k_i}(\widehat{\mathcal C}_K))(\bar Z_i-\hat\theta_{\hat k_i}(\widehat{\mathcal C}_K))'$, so that $\mathrm{tr}(M_K)=N^{-1}\mathrm{WSS}(\bar Z;\widehat{\mathcal C}_K)$, and $\Delta_K=\log\det(\widehat\Gamma_w+M_K)-\log\det\widehat\Gamma_w$. For $\Gamma$ positive definite and $M$ positive semidefinite, $\log\{1+\mathrm{tr}(\Gamma^{-1}M)/P\}\leq\log\det(\Gamma+M)-\log\det\Gamma\leq\mathrm{tr}(\Gamma^{-1}M)$. For $K<K^0$, every partition into $K$ clusters contains a cluster with units from two true clusters, so $\mathrm{WSS}(\bar Z;\widehat{\mathcal C}_K)\geq c_1>0$ with probability approaching one by Assumption~\ref{ass:separation}, and $\Delta_K\geq c_2>0$ since $\widehat\Gamma_w$ converges to a positive definite matrix. For $K\geq K^0$, refining $\mathcal C^0$ gives a partition into $K$ clusters with $\mathrm{WSS}(\bar Z;\cdot)\leq\sum_i\lVert\bar V_i\rVert^2=O_p(T^{-1})$, so $0\leq\Delta_K=O_p(T^{-1})$. The penalty increment between $K$ and $K^0$ is $(K-K^0)P\varsigma\log(NT)/(NT)$, which tends to zero and, with $N$ fixed, is of larger order than $T^{-1}$. Hence $IC(K)>IC(K^0)$ for all $K\neq K^0$ with probability approaching one.
\end{proof}

\paragraph*{Proof of Corollary \ref{cor:holm}}

Part~\ref{cor:holm:parta} follows from the conditional uniformity of $p[d_{k}(\widehat{\mathcal{C}})]$ for $k\in\mathcal I_0(s)$ and the Bonferroni inequality over $\mathcal I_0(s)$, as in \cite{holm1979simple}, since the number of hypotheses is fixed given the selection outcome. Part~\ref{cor:holm:partb} follows from $p[d_{k}(\widehat{\mathcal{C}})]\overset{p}{\longrightarrow}0$.

\section{Split Sample Test Statistic}\label{sec:split}

The main text accounts for the estimation of the clusters through selective conditional inference. A simpler alternative is to split the sample in the time dimension. This section develops a Split Sample test similar to the homogeneity tests of \cite{patton23}.

Let $\mathcal{S}_1 = \{1,2,\dots, \lfloor \gamma T \rfloor\}$ and $\mathcal{S}_2 = \{\lfloor \gamma T \rfloor+l+1,\lfloor \gamma T \rfloor+l+2,\dots,T\}$ be the training and testing samples, where $\gamma \in (0,1)$ is the proportion of observations in the training sample and the gap $l \geq 1$ makes the two samples asymptotically independent (Assumption~\ref{ass:ss_ind} below). We typically choose $\gamma < 0.5$, because the Panel Kmeans estimator of the cluster memberships is super-consistent \citep{bonhomme15a}, whereas the power of the test depends on the number of observations in the testing sample.

Let $\widehat{\mathcal{C}}_{\mathcal{S}_1}$ be the partition of the panel units obtained from the Panel Kmeans estimator \eqref{eq:kmeans} using the sample of $N$ cross-sectional units and the training set $\mathcal{S}_1$.
We define $\hat{\theta}_{\mathcal{S}_2}(\widehat{\mathcal{C}}_{\mathcal{S}_1})= [\hat{\theta}'_{1,\mathcal{S}_2}(\widehat{\mathcal{C}}_{\mathcal{S}_1}), \ldots ,\hat{\theta}'_{K,\mathcal{S}_2}(\widehat{\mathcal{C}}_{\mathcal{S}_1})]'$, with $\hat{\theta}_{k,\mathcal{S}_2}(\widehat{\mathcal{C}}_{\mathcal{S}_1}) = |\mathcal{S}_2|^{-1} \sum_{t \in \mathcal{S}_2} \bar{Z}_{k,t}(\widehat{\mathcal{C}}_{\mathcal{S}_1})$, $\bar{Z}_{k,t}(\widehat{\mathcal{C}}_{\mathcal{S}_1}) = |\widehat{\mathcal{C}}_{k,\mathcal{S}_1}|^{-1} \sum_{i \in \widehat{\mathcal{C}}_{k,\mathcal{S}_1}} Z_{it}$, and let $\bar Z_t(\widehat{\mathcal C}_{\mathcal S_1})$ stack $\bar{Z}_{k,t}(\widehat{\mathcal{C}}_{\mathcal{S}_1})$, $k=1,\dots,K$.
A Split Sample test statistic for $\mathcal{H}_0$ is
\begin{equation}\label{eq:wald_unknown_con}
W_{SS}(\widehat{\mathcal{C}}_{\mathcal{S}_1})
=
\frac{B-KP+1}{KPB}
|\mathcal{S}_2|
\hat{\theta}'_{\mathcal{S}_2}
(\widehat{\mathcal{C}}_{\mathcal{S}_1})
\widehat{\Omega}^{-1}_{\mathcal{S}_2}
(\widehat{\mathcal{C}}_{\mathcal{S}_1})
\hat{\theta}_{\mathcal{S}_2}
(\widehat{\mathcal{C}}_{\mathcal{S}_1}),
\end{equation}
where
\begin{align*}
\widehat{\Omega}_{\mathcal{S}_2}
(\widehat{\mathcal{C}}_{\mathcal{S}_1})
&=
\frac{1}{B}
\sum_{j=1}^B
\widehat{\Lambda}_{j}
(\widehat{\mathcal{C}}_{\mathcal{S}_1})
\widehat{\Lambda}_{j}'
(\widehat{\mathcal{C}}_{\mathcal{S}_1}),
\\
\widehat{\Lambda}_{j}
(\widehat{\mathcal{C}}_{\mathcal{S}_1})
&=
\sqrt{\frac{2}{|\mathcal{S}_2|}}
\sum_{t \in \mathcal{S}_2}
\left[
\bar{Z}_{t}
(\widehat{\mathcal{C}}_{\mathcal{S}_1})
-
\hat{\theta}_{\mathcal{S}_2}
(\widehat{\mathcal{C}}_{\mathcal{S}_1})
\right]
\cos
\left[
\pi j
\left(
\frac{t'-1/2}{|\mathcal{S}_2|}
\right)
\right],
\end{align*}
and $t'=t-\lfloor\gamma T\rfloor-l$ is the position of $t$ in $\mathcal S_2$.
The asymptotic properties of the Split Sample test rely on the following assumption.
\begin{assumptionSS}\normalfont\label{ass:ss_ind}
$\{V_t\}$ is $\beta$-mixing with $\beta(m)\leq C\rho^m$ for some $C>0$ and $\rho\in(0,1)$, and the gap $l=l_T$ satisfies $l_T\to\infty$ and $l_T/T\to0$.
\end{assumptionSS}

Under Assumption~\ref{ass:ss_ind}, the training and testing samples are asymptotically independent: by Berbee's coupling lemma, $\{V_t\}_{t\in\mathcal S_2}$ can be replaced by a sequence with the same distribution that is independent of $\{V_t\}_{t\in\mathcal S_1}$, at the cost of an event of probability at most $\beta(l_T)\to0$ \citep{lunde19}. Geometric $\beta$-mixing covers stationary autoregressive processes with continuously distributed innovations, such as those of the simulations, which are not independent of their past beyond any fixed lag.

The following result is similar to Theorem~6 of \cite{patton23}; the differences are discussed in the remarks below.
\begin{theorem}\normalfont\label{theorem:ss}
Suppose that Assumptions~\ref{ass:alternative}--\ref{ass:clt} and~\ref{ass:ss_ind} hold.
Then, for $B$ fixed, $|\mathcal{S}_1|,|\mathcal{S}_2| \to \infty$ as $T \to \infty$, the following results hold.
\begin{enumerate}[label=(\alph*)]
    \item\label{theorem:ss:parta} Under $\mathcal{H}_0$, $W_{SS}(\widehat{\mathcal{C}}_{\mathcal{S}_1}) \overset{d}{\longrightarrow} \mathbb{F}_{KP,B-KP+1}$.
    \item\label{theorem:ss:partb}
Suppose now that $\mathcal H_0$ fails for $\widehat{\mathcal C}_{\mathcal S_1}$, in the sense that $\lVert\theta^0(\widehat{\mathcal C}_{\mathcal S_1})\rVert\geq c_0$ with probability approaching one for some $c_0>0$. Then, for any $C>0$, $\Pr[W_{SS}(\widehat{\mathcal{C}}_{\mathcal{S}_1})>C] \to 1$.
\end{enumerate}
\end{theorem}

\begin{proof}
The proof follows that of Lemma~\ref{lemma:clt}, except that the CLT is established conditionally on $\mathcal{G}_{\mathcal{S}_1} = \sigma(\{ Z_{it} \}_{i=1}^N, t \in \mathcal{S}_1)$.
We first show that each $P \times 1$ sub-vector of $\hat{\theta}_{\mathcal{S}_2}(\widehat{\mathcal{C}}_{\mathcal{S}_1})$ satisfies $\hat{\theta}_{k,\mathcal{S}_2}(\widehat{\mathcal{C}}_{\mathcal{S}_1}) = \theta^0_{k}(\widehat{\mathcal{C}}_{\mathcal{S}_1}) + o_p(1)$.
By Assumption~\ref{ass:ss_ind} and the coupling argument above, it suffices to establish the results for a testing sample that is independent of the training sample, for which we have 
\begin{equation*}
\begin{split}
\mathbb{E} (\hat{\theta}_{k,\mathcal{S}_2}(\widehat{\mathcal{C}}_{\mathcal{S}_1}) - \theta^0_{k}(\widehat{\mathcal{C}}_{\mathcal{S}_1}) \;|\; \mathcal{G}_{\mathcal{S}_1}) 
&= \mathbb{E} \left( \frac{1}{|\widehat{\mathcal{C}}_{k,\mathcal{S}_1}| |\mathcal{S}_2|} \sum_{i=1}^N \sum_{t \in \mathcal{S}_2} V_{it} \mathbf{1}\{ \hat{k}_{i,\mathcal{S}_1} = k \} \;\middle|\; \mathcal{G}_{\mathcal{S}_1} \right) \\
&= \frac{1}{|\widehat{\mathcal{C}}_{k,\mathcal{S}_1}| |\mathcal{S}_2|} \sum_{i=1}^N \sum_{t \in \mathcal{S}_2} \mathbb{E} (V_{it} \;|\; \mathcal{G}_{\mathcal{S}_1}) \mathbf{1}\{ \hat{k}_{i,\mathcal{S}_1} = k \} = 0.
\end{split}
\end{equation*}

For the conditional variance,
\begin{equation*}
\begin{aligned}
&
\left\lVert
\mathbb{E}
\left[
\left(
\hat{\theta}_{k,\mathcal{S}_2}
(\widehat{\mathcal{C}}_{\mathcal{S}_1})
-
\theta^0_k
(\widehat{\mathcal{C}}_{\mathcal{S}_1})
\right)
\left(
\hat{\theta}_{k,\mathcal{S}_2}
(\widehat{\mathcal{C}}_{\mathcal{S}_1})
-
\theta^0_k
(\widehat{\mathcal{C}}_{\mathcal{S}_1})
\right)'
\,\Big|\,
\mathcal{G}_{\mathcal{S}_1}
\right]
\right\rVert
\\
&=
\left\lVert
\mathbb{E}
\left[
\frac{1}
{\left(
|\widehat{\mathcal{C}}_{k,\mathcal{S}_1}|
|\mathcal{S}_2|
\right)^2}
\sum_{i,j=1}^N
\sum_{t,s \in \mathcal{S}_2}
V_{it}V'_{js}
\mathbf{1}
\left\{
\hat{k}_{i,\mathcal{S}_1}=k
\right\}
\mathbf{1}
\left\{
\hat{k}_{j,\mathcal{S}_1}=k
\right\}
\,\Big|\,
\mathcal{G}_{\mathcal{S}_1}
\right]
\right\rVert
\\
&\leq
\frac{1}
{|\widehat{\mathcal{C}}_{k,\mathcal{S}_1}|^2|\mathcal{S}_2|}
\sum_{i,j=1}^N
\left\lVert
\frac{1}{|\mathcal{S}_2|}
\sum_{t,s \in \mathcal{S}_2}
\mathbb{E}
\left[
V_{it}V'_{js}
\,\Big|\,
\mathcal{G}_{\mathcal{S}_1}
\right]
\right\rVert
\\
&\quad
\times
\mathbf{1}
\left\{
\hat{k}_{i,\mathcal{S}_1}=k
\right\}
\mathbf{1}
\left\{
\hat{k}_{j,\mathcal{S}_1}=k
\right\}
\\
&\leq
\frac{1}
{|\widehat{\mathcal{C}}_{k,\mathcal{S}_1}|^2|\mathcal{S}_2|}
\sum_{i,j=1}^N
\left\lVert
\frac{1}{|\mathcal{S}_2|}
\sum_{t,s \in \mathcal{S}_2}
\mathbb{E}
\left[
V_{it}V'_{js}
\,\Big|\,
\mathcal{G}_{\mathcal{S}_1}
\right]
\right\rVert
\\
&=
O_p\left(|\mathcal{S}_2|^{-1}\right)
\end{aligned}
\end{equation*}
by Assumptions~\ref{ass:alternative} and~\ref{ass:clusternumbers}, since $|\widehat{\mathcal C}_{k,\mathcal S_1}|\geq1$ and $N$ is fixed, from which it follows that $\hat{\theta}_{k,\mathcal{S}_2}(\widehat{\mathcal{C}}_{\mathcal{S}_1}) = \theta^0_{k}(\widehat{\mathcal{C}}_{\mathcal{S}_1}) + o_p(1)$.
Let $\bar V_t(\widehat{\mathcal C}_{\mathcal S_1})$ stack the cluster averages of $V_{it}$, $k=1,\dots,K$, and let $\Omega_{\mathcal S_2}(\widehat{\mathcal C}_{\mathcal S_1})$ be the long-run covariance matrix of $\bar V_t(\widehat{\mathcal C}_{\mathcal S_1})$ given $\mathcal G_{\mathcal S_1}$, which is positive definite by Assumption~\ref{ass:clt}. Given $\mathcal{G}_{\mathcal{S}_1}$, the partition $\widehat{\mathcal C}_{\mathcal S_1}$ is fixed, and Assumption~\ref{ass:clt} applied to the testing sample gives, under $\mathcal{H}_0$,
\begin{equation*}
|\mathcal S_2|^{1/2}\,\Omega_{\mathcal{S}_2}(\widehat{\mathcal{C}}_{\mathcal{S}_1})^{-1/2} [\hat{\theta}_{\mathcal{S}_2}(\widehat{\mathcal{C}}_{\mathcal{S}_1})- \theta^0(\widehat{\mathcal{C}}_{\mathcal{S}_1})]
= \Omega_{\mathcal{S}_2}(\widehat{\mathcal{C}}_{\mathcal{S}_1})^{-1/2} |\mathcal{S}_2|^{-1/2} \sum_{t \in \mathcal{S}_2} \bar{V}_{t}(\widehat{\mathcal{C}}_{\mathcal{S}_1}) \overset{d}{\longrightarrow} \mathcal{N}(0,I_{KP}).
\end{equation*}
Since the cosine weights sum to zero over $\mathcal S_2$, $\widehat\Lambda_j(\widehat{\mathcal C}_{\mathcal S_1})=\sqrt{2/|\mathcal S_2|}\sum_{t\in\mathcal S_2}\bar V_t(\widehat{\mathcal C}_{\mathcal S_1})\cos[\pi j(t'-1/2)/|\mathcal S_2|]+o_p(1)$, and the same argument applied to these weighted sums gives, jointly with the display above, $\Omega_{\mathcal S_2}(\widehat{\mathcal C}_{\mathcal S_1})^{-1/2}\widehat\Lambda_j(\widehat{\mathcal C}_{\mathcal S_1})\overset{d}{\longrightarrow}\xi_j$, $j=1,\dots,B$, where $\xi_0,\xi_1,\dots,\xi_B$ are independent $\mathcal N(0,I_{KP})$ vectors and $\xi_0$ is the limit in the display above, because the cosines are orthogonal to each other and to the constant. With $B$ fixed, $\widehat\Omega_{\mathcal S_2}(\widehat{\mathcal C}_{\mathcal S_1})$ is therefore not consistent: $\Omega_{\mathcal S_2}^{-1/2}\widehat\Omega_{\mathcal S_2}\Omega_{\mathcal S_2}^{-1/2}$ converges to the scaled Wishart matrix $B^{-1}\sum_{j=1}^B\xi_j\xi_j'$, which is independent of $\xi_0$. Since $W_{SS}$ is invariant to the scale $\Omega_{\mathcal S_2}$, the continuous mapping theorem gives $W_{SS}(\widehat{\mathcal C}_{\mathcal S_1})\overset{d}{\longrightarrow}[(B-KP+1)/(KPB)]\,\xi_0'(B^{-1}\sum_{j=1}^B\xi_j\xi_j')^{-1}\xi_0$ conditionally on $\mathcal G_{\mathcal S_1}$, and this limit has the $\mathbb{F}_{KP,B-KP+1}$ distribution by Theorem~3.1 of \cite{sun13}. Since the limit does not depend on $\mathcal G_{\mathcal S_1}$, the convergence also holds unconditionally, which proves Part~(a).

For Part~(b), the first step of the proof, which does not use $\mathcal H_0$, gives $\hat\theta_{\mathcal S_2}(\widehat{\mathcal C}_{\mathcal S_1})=\theta^0(\widehat{\mathcal C}_{\mathcal S_1})+o_p(1)$, so that $\lVert\hat\theta_{\mathcal S_2}(\widehat{\mathcal C}_{\mathcal S_1})\rVert\geq c_0/2$ with probability approaching one. Since $\widehat\Lambda_j(\widehat{\mathcal C}_{\mathcal S_1})$ is computed from demeaned data, it does not depend on $\theta^0(\widehat{\mathcal C}_{\mathcal S_1})$, and the argument of Part~(a) gives $\lambda_{\max}[\widehat\Omega_{\mathcal S_2}(\widehat{\mathcal C}_{\mathcal S_1})]=O_p(1)$. Hence
\[
W_{SS}(\widehat{\mathcal C}_{\mathcal S_1})\geq\frac{B-KP+1}{KPB}\,\frac{|\mathcal S_2|\,\lVert\hat\theta_{\mathcal S_2}(\widehat{\mathcal C}_{\mathcal S_1})\rVert^2}{\lambda_{\max}[\widehat\Omega_{\mathcal S_2}(\widehat{\mathcal C}_{\mathcal S_1})]},
\]
which diverges in probability because $|\mathcal S_2|\to\infty$. This completes the proof.
\end{proof}

We conclude with some remarks. The Split Sample test depends on the choice of the training and testing samples, which is arbitrary in practice, and it uses only part of the sample for inference, which reduces its power. The selective approach uses the whole sample, but it conditions on more than the selection outcome to remove the nuisance parameters, which also costs power, so that the relative power of the two approaches is an empirical question, studied in Section~\ref{sec:supp_mc}. Unlike the asymptotic tests of \cite{patton23}, our test includes a small sample correction. Our framework also allows strong CSD, which is ruled out by \cite{patton23}, and we test whether each cluster has a zero mean, whereas their procedure tests the homogeneity of the panel. Finally, when the relative performance of the forecasts changes over time, the Split Sample test targets the cluster means over the testing sample, which can differ from their averages over the whole evaluation period. It is then not a valid test of the time-averaged hypotheses of Section~\ref{sec:frame} of the paper, whereas the selective tests use the whole sample (Section~\ref{sec:supp_breaks}).

\section{Details on the Empirical Application}\label{sec:supp_app}

This section describes the forecasting methods and the treatment of extreme values in the application of Section~\ref{sec:app} of the paper, and reports summary statistics, the results for the post-crisis and post-COVID testing functions, robustness checks of the selective tests and unit-level tests of J-EPA.

\subsection{Forecasting methods}

All forecasts are one month ahead and use a rolling window of $R=60$ months. Within each window the returns of a currency are standardized by their window mean and standard deviation, and every forecast is a forecast of the standardized return; the AR(1) benchmark is therefore estimated without intercept on the standardized window, the window mean playing the role of the intercept. The macroeconomic predictors enter with a lag of one month and are standardized with the moments of the training observations of the window; predictors that are constant within a window are dropped.

\paragraph*{Univariate methods}
\begin{itemize}\itemsep0em
\item AR($p$): autoregression with intercept and $p\le6$ chosen by BIC.
\item Elastic Net: penalized regression on six lags \citep{ZouHastie2005}, with the mixing parameter on the grid $\{0,0.1,\dots,1\}$ and the penalty chosen by 5-fold cross-validation (\texttt{glmnet}, \citealp{Friedman2010}).
\item No change: zero return, that is the random walk forecast of the exchange rate, equal to $-\bar y/s$ in the standardized units of the window, with $\bar y$ and $s$ the window mean and standard deviation.
\item Shrunk AR(1): the AR(1) coefficient $\hat\phi$ multiplied by $\max(0,1-1/t^2)$, with $t$ its $t$-statistic, a James-Stein type shrinkage towards the no-predictability forecast.
\end{itemize}

\paragraph*{Panel AR(1) methods} Using the standardized windows of all currencies at each date \citep{pesaran26}:
\begin{itemize}\itemsep0em
\item Pooled AR(1): one coefficient for all currencies.
\item Empirical Bayes AR(1): each coefficient is shrunk towards the pooled one with weight $\tau^2/(\tau^2+v_i)$, where $v_i$ is its sampling variance and $\tau^2$ the estimated cross-sectional variance of the coefficients.
\item Grouped AR(1): coefficients pooled within three groups obtained by Kmeans on the unit coefficients.
\end{itemize}

\paragraph*{Methods with macroeconomic or cross-sectional predictors}
\begin{itemize}\itemsep0em
\item XGBoost: boosted regression trees on the own lag and the macroeconomic predictors \citep{ChenGuestrin2016}, with 100 rounds, learning rate 0.05, depth 2 and at least 5 observations per leaf.
\item SVM: $\varepsilon$-insensitive support vector regression with radial kernel on the own lag and the macroeconomic predictors, with the default hyperparameters of \texttt{e1071} \citep{Meyer2001}.
\item Random Forest: 500 regression trees on the own lag and the macroeconomic predictors, with the default hyperparameters of \texttt{randomForest} \citep{Breiman2001,LiawWiener2002}.
\item Factor ridge: ridge regression on the own lag, the lagged dollar factor, that is the cross-sectional average return \citep{lustig2011common}, and the first three principal components of the macroeconomic predictors \citep{stock2002macroeconomic}, with the penalty chosen by generalized cross-validation.
\item Global boosting: one boosting model trained on all currencies of the window together \citep{montero2021principles}, with six own lags, the lagged dollar factor and principal components, and the window standard deviation and first-order autocorrelation of each currency as features; 200 rounds, learning rate 0.03, depth 3, at least 50 observations per leaf and subsampling of 80\% of the observations and features.
\end{itemize}

\paragraph*{Combination} Equal-weight average of the AR($p$), Elastic Net, XGBoost, SVM and random forest forecasts.

\medskip
The default hyperparameters reflect an emphasis on simplicity and replicability; small trees and low learning rates keep the boosting models stable with 59 observations per window. We also computed an Elastic Net with macroeconomic predictors and a Markov-switching AR(1), which we do not report because of frequent extreme forecasts and failures of the estimation, respectively.

\subsection{Extreme values}

Three features of the data are treated before the tests. First, some estimated models, mostly AR($p$), produce explosive forecasts of more than five window standard deviations; following the insanity filter of \cite{swanson1997model}, every forecast is truncated at $\pm5$. Second, devaluations and changes of exchange rate regime produce realized returns of many window standard deviations, and loss differentials hundreds of times larger than the others; pegged currencies whose window standard deviation is almost zero produce the same effect after a small realignment. Third, currencies in a monetary union or with a hard peg to a common anchor have almost identical returns against the dollar, which makes the long-run covariance matrix $\Xi$ singular and produces ties in Panel Kmeans. Groups are defined by an absolute correlation of the standardized returns above 0.98, and one currency of each group is kept: the euro area, the United Kingdom, Australia, South Africa, India and Singapore when they are in the group, otherwise the first in alphabetical order. This drops 36 of the 122 currencies (the CFA and CFP franc zones and the other euro pegs, the sterling dependencies, the rand area, the Australian dollar area, Bhutan and Brunei). In the main results, the standardized realized returns are winsorized at $\pm5$ and the duplicated series are dropped, which leaves 86 currencies. For the post-crisis and post-COVID dummies, the dummy is zero in the training sample of the Split Sample test, so the clusters are estimated from the loss differentials alone, and we report the test with $\gamma=0.1$ only: with $\gamma\geq0.2$ the testing sample starts around the onset of the financial crisis, and the post-crisis dummy is (almost) constant in it.

\subsection{Descriptive analysis of loss differentials}\label{sec:supp_desc}

Table~\ref{tab:summary_stats} reports summary statistics of the loss differentials, in the convention of the paper: the loss of the method minus that of AR(1), so that negative values are improvements.
The medians are close to zero for every method, and the means are small relative to the standard deviations, which range from 0.31 for the shrunk AR(1) to 3.50 for SVM. The three panel or shrinkage estimators of the AR(1) coefficient (empirical Bayes, shrunk and grouped AR(1)) have negative means; SVM, the no-change forecast and XGBoost have the largest means, and the machine learning methods the largest interquartile ranges.

\begin{table}[ht]
\centering
\caption{Summary statistics of the loss differentials relative to AR(1)}\label{tab:summary_stats}
\begin{threeparttable}
\renewcommand{\arraystretch}{0.88}\small
\begin{tabular}{lccccc}
\toprule
Method & Mean & Std. Dev. & 1st Quartile & Median & 3rd Quartile \\
\midrule
AR($p$) & 0.076 & 1.347 & -0.010 & 0.000 & 0.013 \\
Elastic Net & 0.036 & 1.163 & -0.064 & 0.001 & 0.072 \\
No change & 0.298 & 3.106 & -0.115 & -0.003 & 0.128 \\
Shrunk AR(1) & -0.011 & 0.310 & -0.023 & 0.000 & 0.018 \\
Pooled AR(1) & 0.004 & 1.403 & -0.039 & 0.000 & 0.049 \\
Empirical Bayes AR(1) & -0.024 & 0.723 & -0.018 & 0.000 & 0.016 \\
Grouped AR(1) & -0.010 & 0.793 & -0.019 & 0.000 & 0.024 \\
XGBoost & 0.146 & 1.622 & -0.133 & 0.030 & 0.381 \\
SVM & 0.769 & 3.502 & -0.109 & 0.058 & 0.863 \\
Random Forest & 0.034 & 1.434 & -0.114 & 0.007 & 0.195 \\
Factor ridge & 0.017 & 1.077 & -0.056 & 0.001 & 0.069 \\
Global boosting & 0.020 & 1.565 & -0.115 & 0.008 & 0.190 \\
Combination & 0.012 & 1.057 & -0.094 & 0.001 & 0.127 \\
\bottomrule
\end{tabular}
\begin{tablenotes}\footnotesize
\item Note: Loss differentials of the main sample ($N=86$, $T=238$), defined as the loss of the method minus the loss of AR(1), in units of the standard deviation of the estimation window; a negative value is an improvement over AR(1).
\end{tablenotes}
\end{threeparttable}
\end{table}

Figure~\ref{fig:scatter_losses} plots the losses of six methods against those of the AR(1) benchmark. The average loss differentials are small relative to their dispersion, and their medians are close to zero for every method. Only the shrinkage estimators of the AR(1) coefficient improve on AR(1) on average: empirical Bayes ($-0.024$ window variances), the shrunk AR(1) ($-0.011$) and the grouped AR(1) ($-0.010$). SVM ($0.769$), the no-change forecast ($0.298$) and XGBoost ($0.146$) are the least accurate. The scatter plots show that the losses of the univariate and panel methods are close to those of AR(1) in most months, whereas the machine learning methods have large losses in many months in which AR(1) is accurate. The no-change forecast has a separate band of large losses at small AR(1) losses, which comes from a few currencies whose returns have a steady drift with little variation, so that ignoring the drift costs several window standard deviations every month.

\begin{figure}[htbp]
    \centering
    \begin{subfigure}[t]{0.31\textwidth}\centering
        \includegraphics[width=\linewidth]{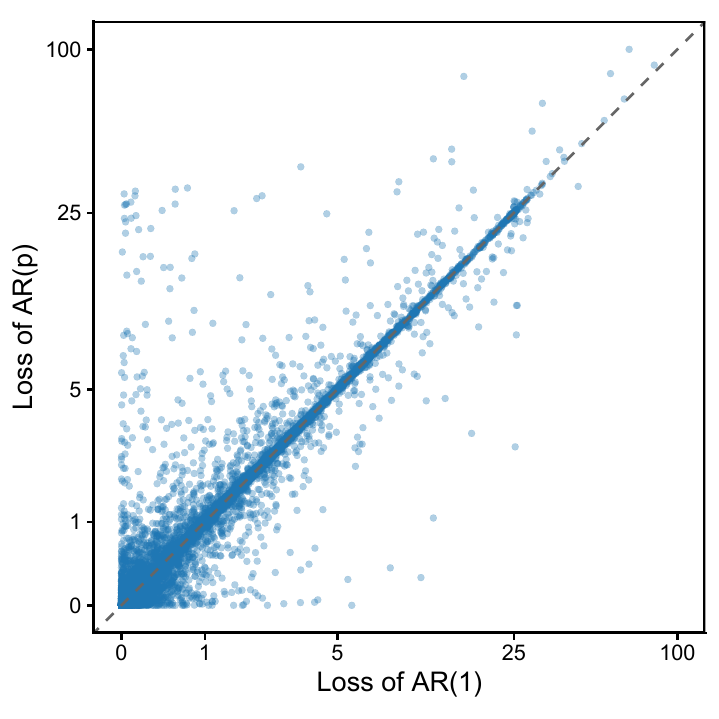}
        \caption{AR($p$)}
    \end{subfigure}\hfill
    \begin{subfigure}[t]{0.31\textwidth}\centering
        \includegraphics[width=\linewidth]{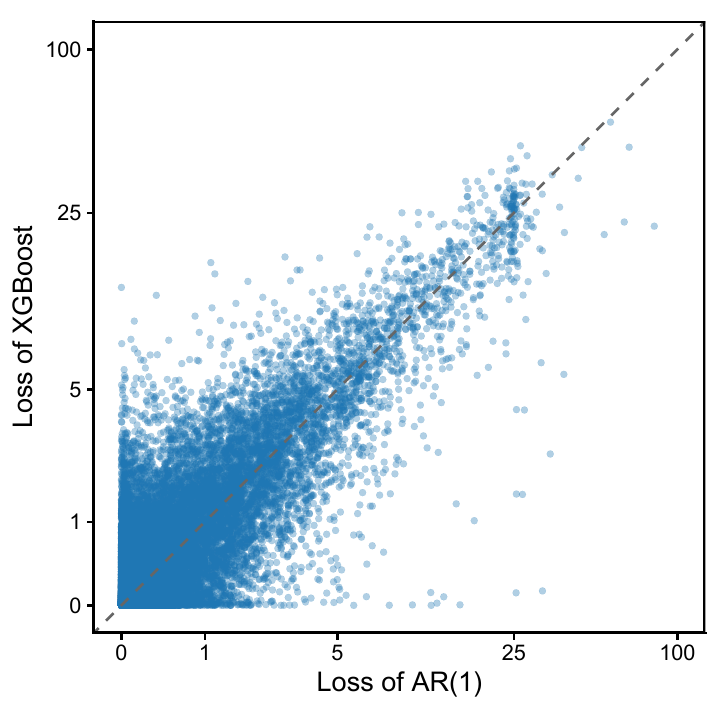}
        \caption{XGBoost}
    \end{subfigure}\hfill
    \begin{subfigure}[t]{0.31\textwidth}\centering
        \includegraphics[width=\linewidth]{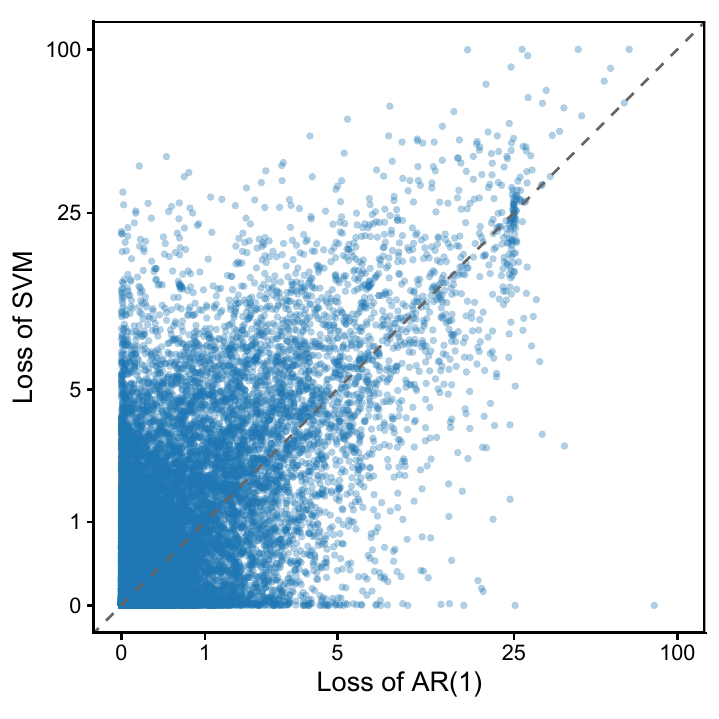}
        \caption{SVM}
    \end{subfigure}

    \medskip
    \begin{subfigure}[t]{0.31\textwidth}\centering
        \includegraphics[width=\linewidth]{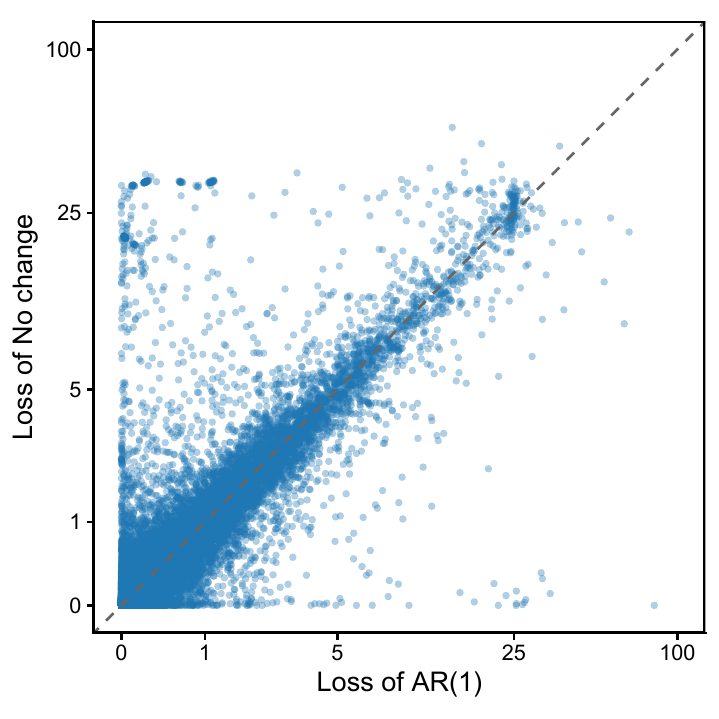}
        \caption{No change}
    \end{subfigure}\hfill
    \begin{subfigure}[t]{0.31\textwidth}\centering
        \includegraphics[width=\linewidth]{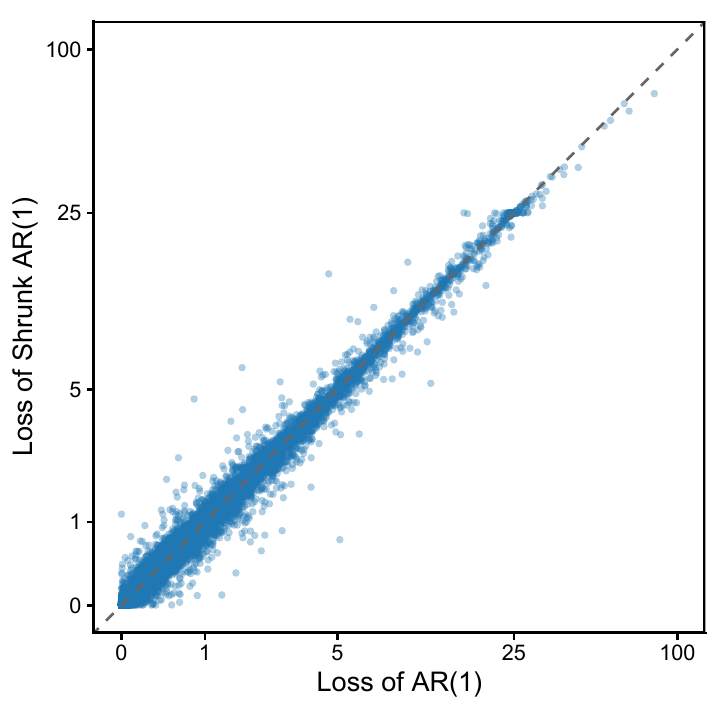}
        \caption{Shrunk AR(1)}
    \end{subfigure}\hfill
    \begin{subfigure}[t]{0.31\textwidth}\centering
        \includegraphics[width=\linewidth]{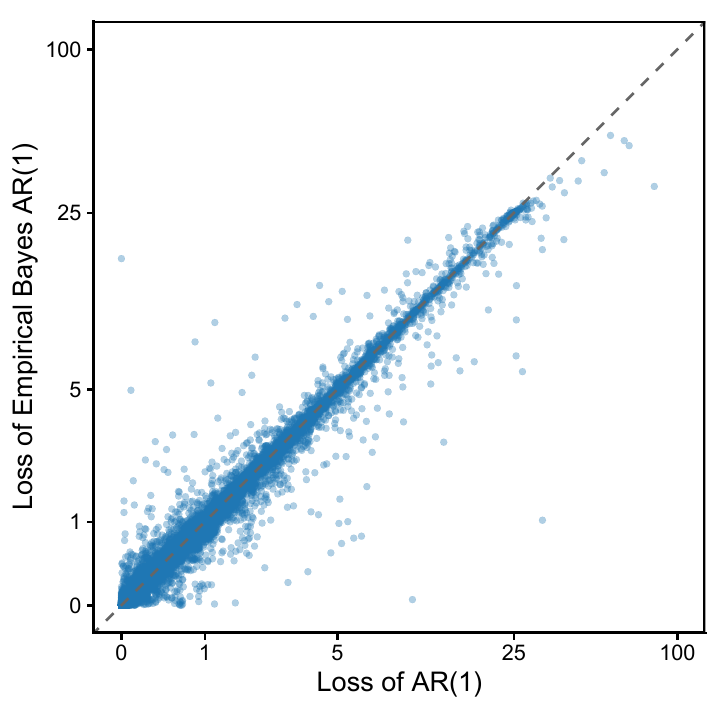}
        \caption{Empirical Bayes AR(1)}
    \end{subfigure}
    \caption{Quadratic losses of six competing methods against those of the AR(1) benchmark, for each currency and month (both axes on the scale of $\log(1+\text{loss})$). The dashed line is the 45-degree line.}
    \label{fig:scatter_losses}
\end{figure}

\subsection{Clusters that violate EPA}

Table~\ref{tab:app_clusters} reports, for the methods and testing functions for which the C-EPA hypothesis is rejected at 5\%, the clusters selected with $\widehat K_{IC}$, their sizes, the estimated cluster means in the convention of the paper and the selective $p$-values of the cluster-specific tests of Section~\ref{sec:clusterspecific} of the main paper, with the Holm adjustment across the clusters of each selected partition.
Only the 10 cases in which at least one cluster is identified are reported; in the other 27 rejections of C-EPA no cluster mean is significant after the adjustment. The identified clusters are of two kinds: large clusters with a small common loss or gain (AR($p$), XGBoost, SVM, empirical Bayes and the combination), and, for the no-change forecast with the dummies, Nicaragua, whose loss differential is an order of magnitude larger than those of all other currencies.

\begin{table}[ht]
\centering
\caption{Cluster-specific selective tests for the rejections of C-EPA}\label{tab:app_clusters}
\begin{threeparttable}
\renewcommand{\arraystretch}{0.88}\footnotesize\setlength{\tabcolsep}{4pt}
\begin{tabular}{llcccccc>{\raggedright\arraybackslash}p{3.2cm}}
\toprule
Method & Testing function & $\widehat K$ & Cluster & $|\widehat{\mathcal C}_k|$ & $\hat\theta_k(\widehat{\mathcal C})$ & $p[d_k(\widehat{\mathcal C})]$ & Holm & Currencies \\
\midrule
AR($p$) & Unconditional & 2 & 1 & 75 & 0.038 & 0.008 & 0.016$^{*}$ &  \\
 &  &  & 2 & 11 & 0.331 & 0.790 & 0.790 &  \\
\addlinespace[2pt]
AR($p$) & Post-crisis & 2 & 1 & 75 & 0.038 & 0.022 & 0.045$^{*}$ &  \\
 &  &  & 2 & 11 & 0.331 & 0.623 & 0.623 &  \\
\addlinespace[2pt]
No change & Post-crisis & 5 & 1 & 2 & 2.204 & 0.256 & 0.640 & CRI, RWA \\
 &  &  & 2 & 11 & 0.156 & 0.325 & 0.640 &  \\
 &  &  & 3 & 1 & 17.866 & 0.004 & 0.020$^{*}$ & NIC \\
 &  &  & 4 & 6 & 0.706 & 0.084 & 0.337 & ARG, DOM, GHA, HND, LAO, SLE \\
 &  &  & 5 & 66 & -0.040 & 0.213 & 0.640 &  \\
\addlinespace[2pt]
No change & Post-COVID & 5 & 1 & 4 & 0.852 & 0.976 & 1.000 & GHA, HND, LAO, SLE \\
 &  &  & 2 & 2 & 2.204 & 0.868 & 1.000 & CRI, RWA \\
 &  &  & 3 & 67 & -0.038 & 0.112 & 0.449 &  \\
 &  &  & 4 & 1 & 17.866 & 0.000 & 0.000$^{*}$ & NIC \\
 &  &  & 5 & 12 & 0.206 & 0.399 & 1.000 &  \\
\addlinespace[2pt]
Empirical Bayes AR(1) & Unconditional & 2 & 1 & 77 & -0.013 & 0.021 & 0.042$^{*}$ &  \\
 &  &  & 2 & 9 & -0.119 & 0.381 & 0.381 & COD, KAZ, LKA, MOZ, NIC, PNG, SLB, TJK, YEM \\
\addlinespace[2pt]
Empirical Bayes AR(1) & Post-crisis & 2 & 1 & 10 & -0.113 & 0.069 & 0.069 & COD, KAZ, LKA, MOZ, NIC, PAK, PNG, SLB, TJK, YEM \\
 &  &  & 2 & 76 & -0.013 & 0.017 & 0.034$^{*}$ &  \\
\addlinespace[2pt]
Empirical Bayes AR(1) & Post-COVID & 2 & 1 & 9 & -0.119 & 0.297 & 0.297 & COD, KAZ, LKA, MOZ, NIC, PNG, SLB, TJK, YEM \\
 &  &  & 2 & 77 & -0.013 & 0.023 & 0.047$^{*}$ &  \\
\addlinespace[2pt]
XGBoost & Lagged return & 2 & 1 & 76 & 0.162 & 0.000 & 0.000$^{*}$ &  \\
 &  &  & 2 & 10 & 0.026 & 0.021 & 0.021$^{*}$ & ARG, COD, HND, JAM, KAZ, KOR, LKA, PNG, RUS, TJK \\
\addlinespace[2pt]
SVM & Post-COVID & 2 & 1 & 53 & 0.966 & 0.003 & 0.006$^{*}$ &  \\
 &  &  & 2 & 33 & 0.454 & 0.538 & 0.538 &  \\
\addlinespace[2pt]
Combination & Unconditional & 3 & 1 & 7 & -0.161 & 0.055 & 0.111 & COD, HND, NIC, PNG, RWA, SLE, TJK \\
 &  &  & 2 & 29 & -0.019 & 0.992 & 0.992 &  \\
 &  &  & 3 & 50 & 0.054 & 0.002 & 0.006$^{*}$ &  \\
\bottomrule
\end{tabular}
\begin{tablenotes}\footnotesize
\item Note: Main sample ($N=86$, $T=238$). $\hat\theta_k(\widehat{\mathcal C})$ is the mean loss differential of the cluster (loss of the method minus loss of AR(1), so that a negative value is an improvement); with the lagged return as testing function, the first element is reported. Holm: $p$-value adjusted by the step-down procedure of \cite{holm1979simple} across the $\widehat K$ clusters; $^{*}$: below 5\%. Currencies: ISO codes for clusters of at most 10 currencies. Only the cases with at least one significant cluster are reported.
\end{tablenotes}
\end{threeparttable}
\end{table}

\subsection{Additional results}

Table~\ref{tab:app_dummies} reports the tests with the post-crisis and post-COVID dummies as testing functions, Table~\ref{tab:app_fixedK} the joint selective C-EPA and homogeneity tests for each fixed $K$, Table~\ref{tab:app_split} the Split Sample tests for the training proportions $\gamma\in\{0.1,0.2,0.3\}$, Table~\ref{tab:app_robust} the robustness of the selective tests for the cases discussed in Section~\ref{sec:app} of the main paper, Table~\ref{tab:app_khat} the number of clusters chosen by the information criterion and by cross-validation, and Table~\ref{tab:app_unit} the unit-level tests of J-EPA. For the dummies, the Split Sample test is reported with $\gamma=0.1$ only. With the dummies, the joint selective C-EPA test rejects after the Holm adjustment for the no-change forecast (both dummies) and the pooled AR(1) (post-COVID dummy), and the homogeneity tests for the no-change forecast. For the no-change forecast, the joint homogeneity test rejects after the Holm adjustment for $K=2$ and $K=3$ in Panel A of Table~\ref{tab:app_fixedK}; the rejections for the pooled AR(1) are specific to $K=3$. The merged $p$-values for $r\in\{-2,-5,-10,-\infty\}$ are in the replication files.

The conclusions above hold with the dummies, for which the joint and homogeneity tests reject after the Holm adjustment only for the no-change and pooled AR(1) forecasts. Conditioning on the paths of all initializations removes almost all rejections of the selective tests, as in the simulations, while the results for the no-change forecast are unaffected by the decomposition, the number of initializations and local demeaning, but weaken with the $\mathbb F$ reference law and with some values of $c_B$. Some of the weaker rejections, for the combination in Panel A and empirical Bayes in Panel B, disappear with 10 initializations instead of 100, so we do not interpret them. In 18 of the 52 cases the estimated shrinkage intensity is one, and the shrinkage results coincide with those of the Euclidean decomposition. The information criterion and cross-validation choose the same number of clusters in 26 of the 52 cases, cross-validation mostly choosing two. Among the unit-level tests of J-EPA in Panel A, the maximum test rejects for XGBoost, SVM and the no-change forecast, and the $\Delta$ tests also for AR($p$), the shrunk, pooled and empirical Bayes AR(1) and global boosting, without identifying the currencies responsible.

\begin{table}[ht]
\centering
\caption{Tests with the post-crisis and post-COVID dummies as testing functions}\label{tab:app_dummies}
\begin{threeparttable}
\renewcommand{\arraystretch}{0.85}\footnotesize
\begin{tabular}{lccccccccc}
\toprule
&&&&& \multicolumn{3}{c}{Metric of $\Xi$} & \multicolumn{2}{c}{Shrinkage} \\ \cmidrule(lr){6-8}\cmidrule(lr){9-10}
Method & $\widehat K$ & O-EPA & Naive & Split & $F_{SI,r}$ & Joint & Homog. & $F_{SI,r}$ & Homog. \\
\midrule
\multicolumn{10}{l}{A. Post-crisis dummy ($H_{i,t-1}=(1,\mathbf 1\{t\geq\text{2008:09}\})'$)} \\
\midrule
\addlinespace[2pt]\multicolumn{10}{l}{\footnotesize Univariate methods} \\
\quad AR($p$) & 2 & 0.00$^{\dagger}$ & 0.00$^{\dagger}$ & 0.00$^{\dagger}$ & 0.00$^{\dagger}$ & 0.02 & 0.67 & 0.00$^{\dagger}$ & 0.69 \\
\quad Elastic Net & 2 & 0.01 & 0.02$^{\dagger}$ & 0.08 & 0.02 & 0.59 & 0.60 & 0.02 & 0.59 \\
\quad No change & 5 & 0.00$^{\dagger}$ & 0.00$^{\dagger}$ & 0.00$^{\dagger}$ & 0.00$^{\dagger}$ & 0.00$^{\dagger}$ & 0.03 & 0.00$^{\dagger}$ & 0.00$^{\dagger}$ \\
\quad Shrunk AR(1) & 2 & 0.00$^{\dagger}$ & 0.00$^{\dagger}$ & 0.00$^{\dagger}$ & 0.00$^{\dagger}$ & 0.23 & 0.82 & 0.00$^{\dagger}$ & 0.87 \\
\addlinespace[2pt]\multicolumn{10}{l}{\footnotesize Panel AR(1) methods} \\
\quad Pooled AR(1) & 4 & 0.21 & 0.00$^{\dagger}$ & 0.09 & 0.46 & 0.01 & 0.40 & 1.00 & 1.00 \\
\quad Empirical Bayes AR(1) & 2 & 0.00$^{\dagger}$ & 0.00$^{\dagger}$ & 0.00$^{\dagger}$ & 0.00$^{\dagger}$ & 0.01 & 0.06 & 0.00$^{\dagger}$ & 0.03 \\
\quad Grouped AR(1) & 2 & 0.01 & 0.00$^{\dagger}$ & 0.07 & 0.03 & 0.15 & 0.31 & 0.03 & 0.03 \\
\addlinespace[2pt]\multicolumn{10}{l}{\footnotesize Methods with macroeconomic or cross-sectional predictors} \\
\quad XGBoost & 2 & 0.00$^{\dagger}$ & 0.00$^{\dagger}$ & 0.00$^{\dagger}$ & 0.00$^{\dagger}$ & 0.48 & 0.79 & 0.00$^{\dagger}$ & 0.88 \\
\quad SVM & 2 & 0.00$^{\dagger}$ & 0.00$^{\dagger}$ & 0.00$^{\dagger}$ & 0.00$^{\dagger}$ & 0.33 & 0.69 & 0.00$^{\dagger}$ & 0.59 \\
\quad Random Forest & 2 & 0.01 & 0.00$^{\dagger}$ & 0.13 & 0.02 & 0.12 & 0.24 & 0.02 & 0.46 \\
\quad Factor ridge & 2 & 0.34 & 0.05 & 0.84 & 0.49 & 0.10 & 0.23 & 0.71 & 0.43 \\
\quad Global boosting & 2 & 0.05 & 0.01$^{\dagger}$ & 0.04 & 0.11 & 0.99 & 0.99 & 0.11 & 0.99 \\
\addlinespace[2pt]
\quad Combination & 2 & 0.55 & 0.00$^{\dagger}$ & 0.16 & 0.62 & 0.17 & 0.29 & 0.79 & 0.38 \\
\midrule
\multicolumn{10}{l}{B. Post-COVID dummy ($H_{i,t-1}=(1,\mathbf 1\{t\geq\text{2020:01}\})'$)} \\
\midrule
\addlinespace[2pt]\multicolumn{10}{l}{\footnotesize Univariate methods} \\
\quad AR($p$) & 3 & 0.00$^{\dagger}$ & 0.02 & 0.01 & 0.01 & 0.38 & 0.33 & 0.01 & 0.06 \\
\quad Elastic Net & 3 & 0.01 & 0.01$^{\dagger}$ & 0.13 & 0.03 & 0.28 & 0.45 & 0.03 & 1.00 \\
\quad No change & 5 & 0.00$^{\dagger}$ & 0.00$^{\dagger}$ & 0.00$^{\dagger}$ & 0.00$^{\dagger}$ & 0.00$^{\dagger}$ & 0.00$^{\dagger}$ & 0.00$^{\dagger}$ & 0.00$^{\dagger}$ \\
\quad Shrunk AR(1) & 2 & 0.00$^{\dagger}$ & 0.00$^{\dagger}$ & 0.06 & 0.00$^{\dagger}$ & 0.53 & 0.62 & 0.00$^{\dagger}$ & 0.74 \\
\addlinespace[2pt]\multicolumn{10}{l}{\footnotesize Panel AR(1) methods} \\
\quad Pooled AR(1) & 3 & 0.14 & 0.00$^{\dagger}$ & 0.04 & 0.14 & 0.00$^{\dagger}$ & 0.11 & 0.12 & 0.09 \\
\quad Empirical Bayes AR(1) & 2 & 0.00$^{\dagger}$ & 0.00$^{\dagger}$ & 0.00$^{\dagger}$ & 0.00$^{\dagger}$ & 0.09 & 0.40 & 0.00$^{\dagger}$ & 0.52 \\
\quad Grouped AR(1) & 3 & 0.11 & 0.01$^{\dagger}$ & 0.06 & 0.47 & 0.74 & 0.36 & 0.48 & 0.50 \\
\addlinespace[2pt]\multicolumn{10}{l}{\footnotesize Methods with macroeconomic or cross-sectional predictors} \\
\quad XGBoost & 2 & 0.00$^{\dagger}$ & 0.00$^{\dagger}$ & 0.00$^{\dagger}$ & 0.00$^{\dagger}$ & 0.25 & 0.55 & 0.00$^{\dagger}$ & 0.57 \\
\quad SVM & 2 & 0.00$^{\dagger}$ & 0.00$^{\dagger}$ & 0.00$^{\dagger}$ & 0.00$^{\dagger}$ & 0.02 & 0.05 & 0.00$^{\dagger}$ & 0.05 \\
\quad Random Forest & 3 & 0.06 & 0.01$^{\dagger}$ & 0.45 & 0.27 & 0.87 & 1.00 & 0.27 & 1.00 \\
\quad Factor ridge & 2 & 0.55 & 0.70 & 0.32 & 1.00 & 0.63 & 0.80 & 0.57 & 0.27 \\
\quad Global boosting & 3 & 0.17 & 0.02 & 0.25 & 0.71 & 0.32 & 0.89 & 0.38 & 0.28 \\
\addlinespace[2pt]
\quad Combination & 2 & 0.18 & 0.00$^{\dagger}$ & 0.08 & 0.38 & 0.34 & 0.31 & 0.38 & 0.72 \\
\bottomrule
\end{tabular}
\begin{tablenotes}\footnotesize
\item Note: As in Table~\ref{tab:app_results} of the main paper, with the dummies as second testing function; Split: $\gamma=0.1$. $^{\dagger}$: significant at 5\% after the Holm adjustment across the 13 methods.
\end{tablenotes}
\end{threeparttable}
\end{table}
\begin{table}[ht]
\centering
\caption{Joint selective tests for each fixed number of clusters}\label{tab:app_fixedK}
\begin{threeparttable}
\renewcommand{\arraystretch}{0.85}\footnotesize
\begin{tabular}{lcccccccc}
\toprule
& \multicolumn{4}{c}{Joint C-EPA} & \multicolumn{4}{c}{Joint homogeneity} \\ \cmidrule(lr){2-5}\cmidrule(lr){6-9}
Method & $K=2$ & $K=3$ & $K=4$ & $K=5$ & $K=2$ & $K=3$ & $K=4$ & $K=5$ \\
\midrule
\multicolumn{9}{l}{A. Unconditional} \\
\midrule
\quad AR($p$) & 0.00$^{\dagger}$ & 0.09 & 0.60 & 0.55 & 0.95 & 0.85 & 0.73 & 0.47 \\
\quad Elastic Net & 0.38 & 0.98 & 0.45 & 0.52 & 0.42 & 0.98 & 0.45 & 0.54 \\
\quad No change & 0.00$^{\dagger}$ & 0.01 & 0.01 & 0.08 & 0.00$^{\dagger}$ & 0.00$^{\dagger}$ & 0.01 & 0.08 \\
\quad Shrunk AR(1) & 0.51 & 0.59 & 0.47 & 0.73 & 0.91 & 0.47 & 0.06 & 0.64 \\
\quad Pooled AR(1) & 0.57 & 0.03 & 0.87 & 0.32 & 0.54 & 0.11 & 0.89 & 0.31 \\
\quad Empirical Bayes AR(1) & 0.05 & 0.66 & 0.58 & 0.14 & 0.49 & 0.79 & 0.66 & 0.42 \\
\quad Grouped AR(1) & 0.49 & 0.99 & 0.96 & 0.65 & 0.47 & 0.99 & 0.92 & 0.82 \\
\quad XGBoost & 0.14 & 0.03 & 0.70 & 0.54 & 0.27 & 0.06 & 0.75 & 0.72 \\
\quad SVM & 0.14 & 0.03 & 0.02 & 0.04 & 0.69 & 0.23 & 0.34 & 0.10 \\
\quad Random Forest & 0.47 & 0.55 & 0.24 & 0.77 & 0.73 & 0.43 & 0.22 & 0.40 \\
\quad Factor ridge & 0.71 & 0.95 & 0.73 & 0.77 & 0.69 & 0.19 & 0.69 & 0.85 \\
\quad Global boosting & 0.78 & 0.30 & 0.37 & 0.99 & 0.82 & 0.17 & 0.63 & 0.90 \\
\quad Combination & 0.38 & 0.08 & 0.35 & 0.37 & 0.44 & 0.39 & 0.28 & 0.30 \\
\midrule
\multicolumn{9}{l}{B. Conditional on the lagged return} \\
\midrule
\quad AR($p$) & 0.01 & 0.60 & 0.54 & 0.13 & 0.67 & 0.55 & 0.77 & 0.34 \\
\quad Elastic Net & 0.06 & 0.29 & 0.57 & 0.98 & 0.08 & 0.27 & 0.67 & 0.96 \\
\quad No change & 0.00$^{\dagger}$ & 0.00$^{\dagger}$ & 0.03 & 0.27 & 0.00$^{\dagger}$ & 0.02 & 0.03 & 0.38 \\
\quad Shrunk AR(1) & 0.06 & 0.05 & 0.08 & 0.01 & 0.62 & 0.55 & 0.17 & 0.64 \\
\quad Pooled AR(1) & 0.36 & 0.60 & 0.71 & 0.10 & 0.35 & 0.81 & 0.76 & 0.10 \\
\quad Empirical Bayes AR(1) & 0.00$^{\dagger}$ & 0.28 & 0.66 & 0.59 & 0.08 & 0.43 & 0.25 & 0.57 \\
\quad Grouped AR(1) & 0.31 & 0.94 & 0.36 & 0.82 & 0.26 & 0.90 & 0.52 & 0.89 \\
\quad XGBoost & 0.00$^{\dagger}$ & 0.01 & 0.13 & 0.00$^{\dagger}$ & 0.03 & 0.31 & 0.60 & 0.02 \\
\quad SVM & 0.04 & 0.60 & 0.11 & 0.01 & 0.45 & 0.93 & 0.41 & 0.09 \\
\quad Random Forest & 0.11 & 0.00$^{\dagger}$ & 0.00$^{\dagger}$ & 0.08 & 0.12 & 0.00$^{\dagger}$ & 0.02 & 0.10 \\
\quad Factor ridge & 0.01 & 0.00$^{\dagger}$ & 0.46 & 0.01 & 0.25 & 0.60 & 0.89 & 0.03 \\
\quad Global boosting & 0.01 & 0.07 & 0.17 & 0.04 & 0.01 & 0.09 & 0.01 & 0.48 \\
\quad Combination & 0.09 & 0.18 & 0.23 & 0.35 & 0.66 & 0.18 & 0.40 & 0.37 \\
\bottomrule
\end{tabular}
\begin{tablenotes}\footnotesize
\item Note: $p$-values of the joint selective C-EPA and homogeneity tests, with Panel Kmeans run for each fixed $K$ (100 initializations). $^{\dagger}$: significant at 5\% after the Holm adjustment across the 13 methods.
\end{tablenotes}
\end{threeparttable}
\end{table}
\begin{table}[ht]
\centering
\caption{Split Sample tests for several training proportions}\label{tab:app_split}
\begin{threeparttable}
\renewcommand{\arraystretch}{0.85}\footnotesize
\begin{tabular}{lcccccccc}
\toprule
& \multicolumn{3}{c}{Unconditional} & \multicolumn{3}{c}{Lagged return} & Post-crisis & Post-COVID \\ \cmidrule(lr){2-4}\cmidrule(lr){5-7}\cmidrule(lr){8-8}\cmidrule(lr){9-9}
Method & $\gamma=0.1$ & $0.2$ & $0.3$ & $\gamma=0.1$ & $0.2$ & $0.3$ & $\gamma=0.1$ & $\gamma=0.1$ \\
\midrule
AR($p$) & 0.00$^{\dagger}$ & 0.00$^{\dagger}$ & 0.00$^{\dagger}$ & 0.04 & 0.04 & 0.05 & 0.00$^{\dagger}$ & 0.01 \\
Elastic Net & 0.05 & 0.04 & 0.07 & 0.15 & 0.11 & 0.09 & 0.08 & 0.13 \\
No change & 0.00$^{\dagger}$ & 0.00$^{\dagger}$ & 0.00$^{\dagger}$ & 0.00$^{\dagger}$ & 0.00$^{\dagger}$ & 0.00$^{\dagger}$ & 0.00$^{\dagger}$ & 0.00$^{\dagger}$ \\
Shrunk AR(1) & 0.00$^{\dagger}$ & 0.00$^{\dagger}$ & 0.00$^{\dagger}$ & 0.01 & 0.01 & 0.01 & 0.00$^{\dagger}$ & 0.06 \\
Pooled AR(1) & 0.01 & 0.10 & 0.24 & 0.07 & 0.03 & 0.14 & 0.09 & 0.04 \\
Empirical Bayes AR(1) & 0.00$^{\dagger}$ & 0.00$^{\dagger}$ & 0.01$^{\dagger}$ & 0.00$^{\dagger}$ & 0.00$^{\dagger}$ & 0.01 & 0.00$^{\dagger}$ & 0.00$^{\dagger}$ \\
Grouped AR(1) & 0.02 & 0.01 & 0.12 & 0.05 & 0.41 & 0.45 & 0.07 & 0.06 \\
XGBoost & 0.00$^{\dagger}$ & 0.00$^{\dagger}$ & 0.00$^{\dagger}$ & 0.00$^{\dagger}$ & 0.00$^{\dagger}$ & 0.00$^{\dagger}$ & 0.00$^{\dagger}$ & 0.00$^{\dagger}$ \\
SVM & 0.00$^{\dagger}$ & 0.00$^{\dagger}$ & 0.00$^{\dagger}$ & 0.00$^{\dagger}$ & 0.00$^{\dagger}$ & 0.00$^{\dagger}$ & 0.00$^{\dagger}$ & 0.00$^{\dagger}$ \\
Random Forest & 0.11 & 0.51 & 0.00$^{\dagger}$ & 0.11 & 0.13 & 0.13 & 0.13 & 0.45 \\
Factor ridge & 0.42 & 0.78 & 0.77 & 0.50 & 0.87 & 0.11 & 0.84 & 0.32 \\
Global boosting & 0.04 & 0.20 & 0.05 & 0.06 & 0.30 & 0.29 & 0.04 & 0.25 \\
Combination & 0.07 & 0.82 & 0.13 & 0.32 & 0.35 & 0.49 & 0.16 & 0.08 \\
\bottomrule
\end{tabular}
\begin{tablenotes}\footnotesize
\item Note: $p$-values of the Split Sample test; the clusters are estimated on the first share $\gamma$ of the months, with $K$ chosen by the information criterion. For the dummies, only $\gamma=0.1$ is reported, since with larger training samples the post-crisis dummy is almost constant in the testing sample. $^{\dagger}$: significant at 5\% after the Holm adjustment across the 13 methods.
\end{tablenotes}
\end{threeparttable}
\end{table}
\begin{table}[ht]
\centering
\caption{Robustness of the selective tests}\label{tab:app_robust}
\begin{threeparttable}
\renewcommand{\arraystretch}{0.85}\footnotesize
\begin{tabular}{lccccccc}
\toprule
& \multicolumn{2}{c}{No change} & \multicolumn{2}{c}{Pooled AR(1)} & EB AR(1) & XGBoost & Combination \\ \cmidrule(lr){2-3}\cmidrule(lr){4-5}\cmidrule(lr){6-6}\cmidrule(lr){7-7}\cmidrule(lr){8-8}
Specification & Uncond. & COVID & Uncond. & COVID & Lagged & Lagged & Uncond. \\
\midrule
\multicolumn{8}{l}{A. Joint selective C-EPA test} \\
\midrule
\quad Baseline & 0.08 & 0.00 & 0.02 & 0.00 & 0.00 & 0.00 & 0.29 \\
\quad Conditioning on all paths & 0.72 & 1.00 & 0.73 & 0.40 & 0.16 & 0.93 & 0.97 \\
\quad Euclidean decomposition & 0.00 & 0.00 & 0.01 & 0.02 & 0.24 & 0.00 & 0.97 \\
\quad Shrinkage estimator & 0.00 & 0.00 & 0.01 & 0.02 & 0.05 & 0.00 & 0.94 \\
\quad $\mathbb F$ reference law & 0.86 & 0.16 & 0.14 & 0.06 & 0.02 & 0.24 & 0.42 \\
\quad $\varsigma=3$ & 0.08 & 0.00 & 0.12 & 0.08 & 0.00 & 0.00 & 0.38 \\
\quad 10 initializations & 0.08 & 0.00 & 0.10 & 0.00 & 0.50 & 0.01 & 0.13 \\
\quad Local demeaning & 0.03 & 0.00 & 0.02 & 0.00 & 0.00 & 0.00 & 0.85 \\
\quad $c_B=0.5$ & 0.08 & 0.37 & 0.01 & 0.00 & 0.01 & 0.00 & 0.01 \\
\quad $c_B=1$ & 0.00 & 0.77 & 0.04 & 0.00 & 0.00 & 0.04 & 0.02 \\
\quad $c_B=2$ & 0.06 & 0.00 & 0.08 & 0.00 & 0.05 & 0.05 & 0.01 \\
\midrule
\multicolumn{8}{l}{B. Joint selective homogeneity test} \\
\midrule
\quad Baseline & 0.08 & 0.00 & 0.05 & 0.12 & 0.06 & 0.03 & 0.93 \\
\quad Conditioning on all paths & 0.72 & 1.00 & 0.26 & 0.47 & 0.32 & 0.06 & 0.97 \\
\quad Euclidean decomposition & 0.00 & 0.00 & 0.03 & 0.09 & 0.51 & 0.22 & 0.97 \\
\quad Shrinkage estimator & 0.00 & 0.00 & 0.03 & 0.09 & 0.41 & 0.09 & 0.96 \\
\quad $\mathbb F$ reference law & 0.85 & 0.17 & 0.18 & 0.34 & 0.11 & 0.15 & 0.95 \\
\quad $\varsigma=3$ & 0.08 & 0.00 & 0.26 & 0.24 & 0.08 & 0.03 & 0.44 \\
\quad 10 initializations & 0.08 & 0.00 & 0.02 & 0.31 & 0.50 & 0.09 & 0.21 \\
\quad Local demeaning & 0.04 & 0.00 & 0.07 & 0.06 & 0.05 & 0.03 & 0.95 \\
\quad $c_B=0.5$ & 0.08 & 0.36 & 0.09 & 0.20 & 0.07 & 0.00 & 0.08 \\
\quad $c_B=1$ & 0.01 & 0.90 & 0.08 & 0.03 & 0.32 & 0.11 & 0.03 \\
\quad $c_B=2$ & 0.07 & 0.00 & 0.09 & 0.05 & 0.43 & 0.10 & 0.01 \\
\bottomrule
\end{tabular}
\begin{tablenotes}\footnotesize
\item Note: $p$-values of the joint selective tests with $K$ chosen by the information criterion, for the cases discussed in Section~\ref{sec:app} of the main paper. Baseline: metric of $\Xi$, conditioning on the selection outcome, $\chi$ reference law, $\varsigma=1.5$, 100 initializations, global demeaning and $c_B$ of the rule of Section~\ref{sec:tests} of the main paper. Uncond.: unconditional; COVID: post-COVID dummy; Lagged: lagged return.
\end{tablenotes}
\end{threeparttable}
\end{table}
\begin{table}[ht]
\centering
\caption{Number of clusters chosen by the information criterion and by cross-validation}\label{tab:app_khat}
\begin{threeparttable}
\renewcommand{\arraystretch}{0.85}\footnotesize
\begin{tabular}{lcccccccc}
\toprule
& \multicolumn{2}{c}{Unconditional} & \multicolumn{2}{c}{Lagged return} & \multicolumn{2}{c}{Post-crisis} & \multicolumn{2}{c}{Post-COVID} \\ \cmidrule(lr){2-3}\cmidrule(lr){4-5}\cmidrule(lr){6-7}\cmidrule(lr){8-9}
Method & IC & CV & IC & CV & IC & CV & IC & CV \\
\midrule
AR($p$) & 2 & 2 & 4 & 2 & 2 & 2 & 3 & 3 \\
Elastic Net & 2 & 2 & 3 & 2 & 2 & 2 & 3 & 2 \\
No change & 5 & 2 & 5 & 4 & 5 & 2 & 5 & 2 \\
Shrunk AR(1) & 3 & 3 & 2 & 2 & 2 & 3 & 2 & 3 \\
Pooled AR(1) & 3 & 2 & 5 & 3 & 4 & 3 & 3 & 2 \\
Empirical Bayes AR(1) & 2 & 2 & 2 & 5 & 2 & 2 & 2 & 2 \\
Grouped AR(1) & 3 & 2 & 2 & 2 & 2 & 3 & 3 & 2 \\
XGBoost & 2 & 2 & 2 & 2 & 2 & 2 & 2 & 2 \\
SVM & 3 & 2 & 2 & 2 & 2 & 2 & 2 & 2 \\
Random Forest & 2 & 2 & 2 & 3 & 2 & 2 & 3 & 2 \\
Factor ridge & 3 & 2 & 2 & 2 & 2 & 2 & 2 & 2 \\
Global boosting & 3 & 3 & 5 & 2 & 2 & 3 & 3 & 2 \\
Combination & 3 & 2 & 3 & 2 & 2 & 2 & 2 & 2 \\
\bottomrule
\end{tabular}
\begin{tablenotes}\footnotesize
\item Note: IC: information criterion with $\varsigma=1.5$; CV: cross-validation over time; $K_{\max}=5$.
\end{tablenotes}
\end{threeparttable}
\end{table}
\begin{table}[ht]
\centering
\caption{Unit-level tests of J-EPA}\label{tab:app_unit}
\begin{threeparttable}
\renewcommand{\arraystretch}{0.85}\footnotesize
\begin{tabular}{lccccccc}
\toprule
Method & Max-$t$ & $\Delta$ & $\Delta_{cc}$ & Choi & CCE $\Delta$ & CCE max & Joint PE \\
\midrule
\multicolumn{8}{l}{A. Unconditional} \\
\midrule
\quad AR($p$) & 1.00 & 0.00 & 0.01 & 0.00 & 0.00 & 0.00 & 0.02 \\
\quad Elastic Net & 1.00 & 0.48 & 0.48 & 0.14 & 0.01 & 0.01 & 0.40 \\
\quad No change & 0.00 & 0.00 & 0.00 & 0.00 & 0.00 & 0.00 & 0.00 \\
\quad Shrunk AR(1) & 0.48 & 0.00 & 0.05 & 0.00 & 0.00 & 0.00 & 0.73 \\
\quad Pooled AR(1) & 0.58 & 0.00 & 0.00 & 0.00 & 0.00 & 0.76 & 0.03 \\
\quad Empirical Bayes AR(1) & 0.91 & 0.00 & 0.00 & 0.00 & 0.00 & 0.00 & 0.06 \\
\quad Grouped AR(1) & 1.00 & 0.72 & 0.69 & 0.19 & 0.08 & 0.08 & 1.00 \\
\quad XGBoost & 0.00 & 0.00 & 0.00 & 0.00 & 0.00 & 0.00 & 0.02 \\
\quad SVM & 0.00 & 0.00 & 0.00 & 0.00 & 0.00 & 0.00 & 0.00 \\
\quad Random Forest & 0.67 & 0.11 & 0.25 & 0.01 & 0.02 & 0.02 & 0.69 \\
\quad Factor ridge & 1.00 & 0.74 & 0.63 & 0.21 & 0.59 & 0.59 & 1.00 \\
\quad Global boosting & 0.66 & 0.03 & 0.24 & 0.01 & 0.51 & 0.78 & 0.37 \\
\quad Combination & 0.84 & 0.06 & 0.19 & 0.00 & 0.69 & 0.69 & 0.37 \\
\midrule
\multicolumn{8}{l}{B. Conditional on the lagged return} \\
\midrule
\quad AR($p$) & 1.00 & 0.14 & 0.20 & 0.00 & 0.00 & 0.00 & 0.67 \\
\quad Elastic Net & 0.00 & 0.01 & 0.06 & 0.01 & 0.10 & 0.00 & 0.00 \\
\quad No change & 0.00 & 0.00 & 0.00 & 0.00 & 0.00 & 0.00 & 0.00 \\
\quad Shrunk AR(1) & 0.00 & 0.00 & 0.03 & 0.00 & 0.00 & 0.00 & 0.02 \\
\quad Pooled AR(1) & 0.00 & 0.00 & 0.01 & 0.00 & 0.00 & 0.03 & 0.00 \\
\quad Empirical Bayes AR(1) & 0.06 & 0.00 & 0.00 & 0.00 & 0.00 & 0.00 & 0.00 \\
\quad Grouped AR(1) & 0.14 & 0.34 & 0.39 & 0.06 & 0.04 & 0.04 & 0.38 \\
\quad XGBoost & 0.01 & 0.00 & 0.00 & 0.00 & 0.00 & 0.00 & 0.00 \\
\quad SVM & 0.00 & 0.00 & 0.00 & 0.00 & 0.00 & 0.00 & 0.00 \\
\quad Random Forest & 0.73 & 0.00 & 0.03 & 0.00 & 0.02 & 0.02 & 0.14 \\
\quad Factor ridge & 1.00 & 0.56 & 0.53 & 0.09 & 0.96 & 1.00 & 0.01 \\
\quad Global boosting & 0.01 & 0.00 & 0.13 & 0.00 & 0.23 & 0.05 & 0.07 \\
\quad Combination & 0.95 & 0.06 & 0.20 & 0.01 & 0.88 & 0.88 & 0.44 \\
\bottomrule
\end{tabular}
\begin{tablenotes}\footnotesize
\item Note: Max-$t$: maximum of the unit-level Wald statistics with the Bonferroni correction; $\Delta$: standardized average of the unit-level Wald statistics \citep{pesaran2008testing,pesaran2024testing}; $\Delta_{cc}$: the same with the variance corrected for cross-sectional correlation; Choi: inverse normal combination of the unit-level $p$-values; CCE $\Delta$ and CCE max: tests on the intercepts of the units after regressing on the cross-sectional averages, combined with the O-EPA test; Joint PE: joint selective C-EPA test enhanced by the maximum test \citep{fan2015power}.
\end{tablenotes}
\end{threeparttable}
\end{table}

\section{Panel Kmeans Algorithm}\label{sec:supp_kmeans}

Problem \eqref{eq:kmeans} is typically solved by an iterative algorithm such as those of \cite{lloyd82} and \cite{hartigan75}; Algorithm~\ref{algo} generalizes Lloyd's algorithm.

\begin{algorithm}
\caption{Panel Kmeans}
\label{algo}
\KwIn{Data $\{Z_{it}: i=1,\dots,N,\ t=1,\dots,T\}$, number of clusters $K$}
\KwOut{Cluster assignments $k_i$, cluster centers $\theta_k$}

Draw $K$ distinct units $j_1,\dots,j_K$ at random, independently of the data, and set $\theta_k^{(0)}=T^{-1}\sum_{t=1}^TZ_{j_kt}$ for $k = 1, \dots, K$; set $m \gets 0$\;

\Repeat{$k_i^{(m)} = k_i^{(m-1)}$ for all $i = 1, \dots, N$}{
    \For{$i \gets 1$ \KwTo $N$}{
        $k_i^{(m+1)} \gets \arg\min\limits_{k \in \{1, \dots, K\}} \sum_{t=1}^T \lVert Z_{it} - \theta_k^{(m)} \rVert^2$\;
    }
    \For{$k \gets 1$ \KwTo $K$}{
        Update cluster $\mathcal{C}_k^{(m+1)} \gets \{i : k_i^{(m+1)} = k\}$\;
        $\theta_k^{(m+1)} \gets \frac{1}{|\mathcal{C}_k^{(m+1)}| T} \sum\limits_{i \in \mathcal{C}_k^{(m+1)}} \sum\limits_{t=1}^T Z_{it}$\;
    }
    $m \gets m + 1$\;
}
\end{algorithm}

Although the criterion in \eqref{eq:kmeans} is written in terms of the whole panel, Lemma~\ref{lemma:average_equivalence} shows that, with time-invariant cluster centers, the assignment and center-update steps of Algorithm~\ref{algo} are exactly those of Lloyd's algorithm applied to the unit-specific time averages $\bar Z_1,\dots,\bar Z_N$. The panel structure matters for the asymptotic distribution and the long-run variance, but the clustering path depends on the data only through $\bar Z$. Like classical Kmeans, the algorithm solves a non-convex problem whose objective is piecewise quadratic in the centers and discontinuous in the assignments, so that it can converge to local minima, which motivates multiple random initializations. Extending the selective inference framework of \cite{chen23} from classical Kmeans to Panel Kmeans requires accounting for the temporal and cross-sectional dependence of the observations.

As in \cite{chen23}, Algorithm~\ref{algo} is initialized with $K$ units chosen at random, whose time averages serve as the first cluster centers. The indices of these units are drawn independently of the data, and the first assignment, like the following ones, is defined by comparisons of distances to cluster centers, so the truncation set calculations of Section~\ref{sec:calc_of_p} treat all assignments in the same way. Assigning each unit at random to a cluster and computing the corresponding centers is not a suitable alternative: the first centers are then close to the overall mean, and when the clusters are well separated some of them are empty after the first assignment for most initializations.

\section{Calculation of the Truncation Set $\mathcal{T}$}\label{sec:calc_of_p}

For convenience, we restate the expression for the truncation set $\mathcal{T}$:
\begin{equation*}
\mathcal{T} = \left\lbrace \phi \in \mathbb{R}_{\geq 0} : \hat s[z(\phi)] = \hat s(z) \right\rbrace.
\end{equation*}

By Lemma~\ref{lemma:average_equivalence} and \eqref{eq:perturbed}, the selection outcome along the perturbation path is determined by the perturbed time averages $\bar z_i(\phi)$, which satisfy $\sqrt T\bar z_i(\phi)=\hat u_i+\phi\hat v_i$, with $\hat v=\hat\xi L_{k,g}\widehat\sigma^{-1/2}_{k,g}\mathrm{dir}(\widehat\sigma^{-1/2}_{k,g}\sqrt TL_{k,g}'\bar z)$, and $\hat u_i$ and $\hat v_i$ the $i$th blocks of $\hat u_{k,g}$ and $\hat v$. Since all comparisons are invariant to multiplication of the time averages by $\sqrt T$, we work with $x_i(\phi)=\hat u_i+\phi\hat v_i$. The selection outcome at $\phi$, $\hat s[z(\phi)]$, is obtained by running the whole procedure on $x(\phi)$, with the same initializations, that is, the same indices $j_1,\dots,j_K$ of the units used as first centers for every initialization and every $K\in\mathcal K$. The truncation set is the set of $\phi\geq0$ at which this procedure selects $\widehat K_{IC}(z)$ clusters and the observed partition, up to the labels of the clusters. It is computed exactly from three types of comparisons, all of which are polynomial in $\phi$.

For a set of units $\mathcal G$, let $\bar u_{\mathcal G}=|\mathcal G|^{-1}\sum_{j\in\mathcal G}\hat u_j$, $\bar v_{\mathcal G}=|\mathcal G|^{-1}\sum_{j\in\mathcal G}\hat v_j$, $a_{i,\mathcal G}=\hat u_i-\bar u_{\mathcal G}$ and $b_{i,\mathcal G}=\hat v_i-\bar v_{\mathcal G}$, so that
\begin{equation}\label{eq:s_quadratic}
\lVert x_i(\phi)-\bar x_{\mathcal G}(\phi)\rVert^2=\lVert b_{i,\mathcal G}\rVert^2\phi^2+2a_{i,\mathcal G}'b_{i,\mathcal G}\,\phi+\lVert a_{i,\mathcal G}\rVert^2.
\end{equation}

\paragraph*{Assignments.} At iteration $m$ of an initialization, let $\mathcal G_1,\dots,\mathcal G_K$ be the clusters of the previous iteration (for the first assignment, the singletons $\{j_1\},\dots,\{j_K\}$ of the units drawn at the initialization). Unit $i$ is assigned to cluster $k$ if $\lVert x_i(\phi)-\bar x_{\mathcal G_k}(\phi)\rVert^2-\lVert x_i(\phi)-\bar x_{\mathcal G_g}(\phi)\rVert^2<0$ for $g<k$ and $\leq0$ for $g>k$, reflecting the tie-breaking rule. Each comparison is a quadratic $A\phi^2+B\phi+C$ with $A=\lVert b_{i,\mathcal G_k}\rVert^2-\lVert b_{i,\mathcal G_g}\rVert^2$, $B=2(a_{i,\mathcal G_k}'b_{i,\mathcal G_k}-a_{i,\mathcal G_g}'b_{i,\mathcal G_g})$ and $C=\lVert a_{i,\mathcal G_k}\rVert^2-\lVert a_{i,\mathcal G_g}\rVert^2$. Starting from $\phi=0$, the path of an initialization is constant on the interval up to the smallest root, larger than the current $\phi$, of the quadratics defined by its current path. At that root the algorithm is run again just beyond it, starting from the iteration whose comparison changes sign, since the comparisons of the earlier iterations do not depend on the later ones. Proceeding in this way gives, for each initialization, a finite partition of the half-line into segments on which its path, and hence its final partition, or the fact that it is discarded because a cluster becomes empty, is constant.

\paragraph*{Retained initialization.} On each piece of the half-line on which the final partitions of all initializations with $K$ clusters are constant, the retained one is the first initialization with the smallest within-cluster sum of squares $\mathrm{WSS}(\phi;\mathcal C)=\sum_k\sum_{i\in\mathcal C_k}\lVert x_i(\phi)-\bar x_{\mathcal C_k}(\phi)\rVert^2$, a quadratic in $\phi$. The comparisons $\mathrm{WSS}(\phi;\mathcal C_{s})-\mathrm{WSS}(\phi;\mathcal C_{s'})$ between distinct final partitions split the piece into sub-intervals on which the retained partition is constant.

\paragraph*{Number of clusters.} On each piece on which the retained partitions for all $K\in\mathcal K$ are constant, $\mathbb M_K(\phi)=\widehat\Gamma_w+(NT)^{-1}\sum_i(x_i(\phi)-\bar x_{\hat k_i}(\phi))(x_i(\phi)-\bar x_{\hat k_i}(\phi))'$ is a quadratic matrix polynomial in $\phi$, where $\widehat\Gamma_w=(NT)^{-1}\sum_{i,t}(z_{it}-\bar z_i)(z_{it}-\bar z_i)'$ and the clusters are those of the retained partition with $K$ clusters. With $\mathrm{pen}(K)=(KP+N)\varsigma\log(NT)/(NT)$, $\widehat K_{IC}$ is selected where
\[
e^{\mathrm{pen}(\widehat K_{IC})}\det\mathbb M_{\widehat K_{IC}}(\phi)-e^{\mathrm{pen}(K)}\det\mathbb M_K(\phi)<0\ \text{ for } K<\widehat K_{IC},\qquad\leq0\ \text{ for } K>\widehat K_{IC},
\]
polynomial inequalities of degree at most $2P$. Their coefficients are obtained by evaluating the left side at $2P+1$ points and interpolating, and their real roots as the eigenvalues of the companion matrix; for $P=1$ they are quadratics.

\paragraph*{Truncation set.} $\mathcal T$ is the union of the pieces on which the selected number of clusters is $\widehat K_{IC}(z)$ and the retained partition with $\widehat K_{IC}(z)$ clusters is the observed one. Since $p[d_{k,g}(\widehat{\mathcal C})]$ is a ratio of masses of the truncated $\chi_P$ distribution, the half-line need not be explored entirely: the computation proceeds in blocks of $\phi$ downwards from the observed statistic, and above it up to a point beyond which the $\chi_P$ mass is negligible, and stops when the masses accumulated so far determine the $p$-value to the desired accuracy. The cost is linear in the number of initializations and in the number of segments visited.

\paragraph*{Conditioning on all paths.} If, as in \cite{chen23}, the selection outcome also includes the paths of all initializations for all $K\in\mathcal K$ and the indices of the retained initializations, the truncation set is the intersection of the sets on which every assignment comparison of the observed paths, every comparison between the retained initialization and the others, and every comparison of the information criterion keeps its sign. This set is contained in the one above, and it is much smaller when clusters overlap, because intermediate assignments change at small perturbations that leave the selected partition unchanged.

\paragraph*{Numerical aspects.} Roots are computed with tolerances relative to the scale of the coefficients, and comparisons whose coefficients are all within tolerance of zero are treated as ties. The $p$-value is computed on the log scale, using the $\chi^2_P$ survival function at $\phi^2$, to avoid underflow when $d_{k,g}$ is large. When $\hat\xi L_{k,g}$ is proportional to $L_{k,g}$, $\hat v_i$ is proportional to $\hat\delta_{k,g,i}$ times a common vector, and the formulas reduce to those of the Euclidean decomposition of \cite{chen23} adapted to panels.

The computation is summarized in Algorithm~\ref{algo:selective}.

\begin{algorithm}[ht]
\caption{Selective $p$-value for $\mathcal H^{k,g}_0$}
\label{algo:selective}
\KwIn{Data $\{z_{it}\}$, $\mathcal K$, number of initializations $S$, $\varsigma$, $B$, initializations}
Compute $\bar z$, the cosine transforms $\hat\lambda_j$, $j=1,\dots,B$, of the unit-demeaned series, and $\widehat\Gamma_w$\;
Run Algorithm~\ref{algo} from each initialization for each $K\in\mathcal K$; retain the best initialization for each $K$ and select $\widehat K_{IC}$; set $\widehat{\mathcal C}$\;
Compute $L_{k,g}$, $g_j=L_{k,g}'\hat\lambda_j$, $\widehat\sigma_{k,g}=B^{-1}\sum_jg_jg_j'$, $\hat\xi L_{k,g}=B^{-1}\sum_j\hat\lambda_jg_j'$\;
Compute $d_{k,g}$, the direction, $\hat u_{k,g}$ and $\hat v$\;
\For{each initialization and each $K\in\mathcal K$}{
  Follow its path along $\phi$ and record the segments on which its final partition is constant\;
}
On the common refinement of the segments, find the retained partition for each $K$ (quadratic comparisons) and the selected $\widehat K_{IC}$ (polynomials of degree at most $2P$)\;
Set $\mathcal T$ equal to the union of the pieces on which the selected outcome is $(\widehat K_{IC}(z),\widehat{\mathcal C})$\;
\KwOut{$p[d_{k,g}(\widehat{\mathcal C})]=1-F_{\chi_P}[d_{k,g}(\widehat{\mathcal C});\mathcal T]$}
\end{algorithm}

\section{Primitive Conditions for Assumption~\ref*{ass:clt}}\label{sec:hdclt_primitives}

Assumption~\ref{ass:clt} is a high-level condition. With $N$ fixed, it is a functional central limit theorem for the $NP$-dimensional process $V_t=(V_{1t}',\dots,V_{Nt}')'$, which holds under standard mixing and moment conditions: if $\{V_t\}$ is $\alpha$-mixing of size $-r/(r-2)$ and $\sup_t\mathbb E\lvert V_{p,it}\rvert^{r}<\infty$ for some $r>2$, all $i$ and $p$, and $\Xi=\lim_{T\to\infty}\mathrm{Var}(\sqrt T\bar V)$ exists and is positive definite, then the functional central limit theorem for mixing sequences \citep[Chapter~7]{white01} gives Assumption~\ref{ass:clt}. No restriction is placed on the dependence between the series of different units beyond the joint mixing condition, which accommodates common factors with heterogeneous loadings.

\subsection{Primitive conditions under rolling-window forecasts}\label{sec:primitive_forecasts}

Let $W_t$ collect all variables observed at time $t$ for all units, including targets, predictors and conditioning variables.

\begin{proposition}\normalfont\label{prop:primitive_forecasts}
Suppose that $\{W_t\}$ is strictly stationary and $\alpha$-mixing with $\alpha_W(m)\leq C\rho^m$, that the forecasts made at $t-\tau$ and the testing functions $H_{i,t-\tau}$ are measurable functions of $(W_{t-\tau-R+1},\dots,W_{t-\tau})$ for a fixed window length $R$, and that the loss is a measurable function of the forecast and the target. Then $\{Z_t\}$ is strictly stationary and $\alpha$-mixing with $\alpha_Z(m)\leq\alpha_W(m-R-\tau+1)$ for $m\geq R+\tau$.

\end{proposition}

\begin{proof}
$Z_t$ is a measurable function of a block of $R+\tau$ consecutive $W$'s, so strict stationarity is inherited, and the $\sigma$-fields generated by $\{Z_s\}_{s\leq t}$ and $\{Z_s\}_{s\geq t+m}$ are contained in those generated by $\{W_s\}_{s\leq t}$ and $\{W_s\}_{s\geq t+m-R-\tau+1}$ \citep[Theorem~3.49]{white01}.

\end{proof}

Proposition~\ref{prop:primitive_forecasts} and finite moments of order $4+\delta$ imply Assumption~\ref{ass:clt} through the functional central limit theorem for mixing sequences \citep[Chapter~7]{white01}, for any cross-sectional dependence, including common factors with heterogeneous loadings. For least squares forecasts, as in Examples~1 and~2, the moment conditions require moments of the inverse of the window design matrices, which exist for Gaussian designs when $R$ is large relative to the number of regressors and the order of the moment; the Lipschitz condition holds on events on which these design matrices are bounded away from singularity. The fixed effects of Example~2 are parameters when $N$ is fixed.

\section{Stability of the Selection Outcome}\label{sec:stabilization}

This section gives sufficient conditions for Assumption~\ref{ass:stability}. Write $\hat s_T(y)$ for the selection outcome computed from the time averages $\mu^*+y/\sqrt T$, with within-unit second moments $\Gamma_T\to\Gamma$.

\begin{lemma}\normalfont\label{lemma:stabilization}
Assumption~\ref{ass:stability} holds in each of the following cases.
\begin{enumerate}[label=(\roman*)]
\item $\mu^*=0$, with $\hat s_\infty(y)$ given by the partition selected by the procedure applied to $y$ with $\widehat K_{IC}=2$.
\item $\mu^*$ is a regular point: there are $\delta>0$ and $T_0$ such that the selection outcome computed from time averages $z$ with $\lVert z-\mu^*\rVert<\delta$ and within-unit second moments within $\delta$ of $\Gamma$ is constant for $T\geq T_0$. Then $\hat s_\infty$ is constant.
\item $\mu^*$ is general, including grouped expected loss differentials, provided that, for all $K,K'\in\mathcal K$, either the within-cluster scatter matrices of the limiting partitions at $\mu^*$ coincide or the corresponding values of $\log\det(\Gamma+N^{-1}\mathbb M)$ differ.
\end{enumerate}
\end{lemma}

\begin{proof}
Since the selection outcome is a function of the paths of all initializations, of the retained initializations and of $\widehat K_{IC}$, it suffices that these are eventually constant. (i) With $\mu^*=0$, the procedure is applied to $y_T$, since multiplying all time averages by $\sqrt T$ does not change the comparisons of the algorithm. The paths and the choices of initializations are determined by finitely many polynomial inequalities in $y$ that do not depend on $T$. A polynomial defining an assignment comparison is not identically zero, and a comparison between two initializations is identically zero only if they end in the same partition, in which case the tie-breaking rule decides for every $y$. Outside the union of the zero sets of the polynomials that are not identically zero, which is Lebesgue null, every inequality is strict and remains so along any sequence $y_T\to y$. For the information criterion, $T[IC(K)-\log\det\Gamma_T]-T\,\mathrm{pen}_T(K)$ converges to $N^{-1}\mathrm{tr}(\Gamma^{-1}M_K(y))$, where $M_K(y)$ is the within-cluster scatter of $y$, while $T[\mathrm{pen}_T(K)-\mathrm{pen}_T(K')]=(K-K')P\varsigma\log(NT)/N$ diverges, so $\widehat K_{IC}=2$ for $T$ large. (ii) If $y_T\to y$, then $\mu^*+y_T/\sqrt T\to\mu^*$, and the selection outcome is eventually the constant value. (iii) This is a proof sketch. Proceed by induction over iterations: if the assignments at iteration $m-1$ are eventually constant along $y_T\to y$ outside a null set, the centers at iteration $m$ are fixed linear combinations of the points, and each assignment comparison evaluated at $\sqrt T\mu^*+y_T$ has the form $Tq_2+\sqrt Tq_1(y_T)+q_0(y_T)$, with $q_2$ the comparison at $\mu^*$, $q_1$ linear and $q_0$ quadratic in $y$, and coefficients not depending on $T$. If $q_2\neq0$, its sign decides for $T$ large; if $q_2=0$ and $q_1\not\equiv0$, the sign of $q_1(y)$ decides outside a hyperplane; if $q_2=0$ and $q_1\equiv0$, the sign of $q_0(y)$ decides outside its zero set, unless $q_0\equiv0$, in which case the tie-breaking rule decides for every $y$. Comparisons between initializations are treated in the same way. For the information criterion, the order-one differences are nonzero by assumption when the scatter matrices differ; when they coincide they cancel exactly for every $\Gamma_T$, and the comparison is decided by a term of order $T^{-1/2}$, linear in $y$, outside a null set, or, if that term vanishes identically, by the penalty, of order $T^{-1}\log T$, which dominates the remaining $O(T^{-1})$ terms.
\end{proof}

\section{Additional Monte Carlo Results}\label{sec:supp_mc}

This section reports the experiments summarized at the end of Section~\ref{sec:monte} of the paper. All of them use $N=80$, 1000 replications, $K_{max}=5$, $\varsigma=1.5$, $S=10$ Forgy initializations and the bandwidth $B$ of Section~\ref{sec:maintests}, unless stated otherwise; they are produced by the script \texttt{Simulation/run\_simulations.R} of the replication package. Rejection rates are at the 5\% level. Under the null hypothesis and in Cases~1 and~2, the designs are those of Section~\ref{sec:monte}.

\subsection{Split Sample, homogeneity and O-EPA tests}

Table~\ref{tab:supp_main} completes Tables~\ref{tab:size_main} and~\ref{tab:power} of the paper with the Split Sample test with training proportions $\gamma=0.1$ and $0.3$, the selective tests of homogeneity, merged ($F_{homo,r}$) and joint, and the O-EPA test. All tests have size close to 5\%, the joint homogeneity test being slightly oversized for small $T$. The Split Sample test is more powerful with larger training proportions, because the estimated partition is then closer to the true one while the testing sample remains long. The homogeneity tests have similar power in the two cases, in which the differences between the cluster means are the same; the O-EPA test has power only in Case~1.

\begin{table}[htbp]
\centering
\begin{threeparttable}
\caption{Additional rejection rates in the main design}\label{tab:supp_main}
\scriptsize\setlength{\tabcolsep}{3pt}\renewcommand{\arraystretch}{0.9}
\begin{tabular}{lrr ccccc c ccccc}
\toprule
&&& \multicolumn{5}{c}{Unconditional} && \multicolumn{5}{c}{Conditional}\\ \cmidrule{4-8}\cmidrule{10-14}
Design & $\psi$ & $T$ & Split$_{0.1}$ & Split$_{0.3}$ & Homog. & Homog. J. & O-EPA && Split$_{0.1}$ & Split$_{0.3}$ & Homog. & Homog. J. & O-EPA\\
\midrule
Null & 0.000 & 20 & .044 & .059 & .037 & .069 & .074 && .040 & .062 & .058 & .069 & .077\\
Null & 0.000 & 50 & .051 & .050 & .049 & .063 & .063 && .056 & .056 & .055 & .074 & .057\\
Null & 0.000 & 100 & .055 & .068 & .062 & .075 & .052 && .064 & .065 & .073 & .073 & .057\\
Null & 0.000 & 200 & .066 & .058 & .040 & .049 & .056 && .070 & .067 & .054 & .056 & .063\\
Case 1 & 0.125 & 50 & .065 & .056 & .052 & .082 & .080 && .057 & .067 & .065 & .076 & .112\\
Case 1 & 0.125 & 200 & .230 & .252 & .079 & .096 & .225 && .153 & .189 & .098 & .103 & .205\\
Case 1 & 0.250 & 50 & .178 & .165 & .058 & .096 & .215 && .113 & .129 & .072 & .093 & .229\\
Case 1 & 0.250 & 200 & .853 & .940 & .176 & .243 & .696 && .614 & .682 & .207 & .292 & .639\\
Case 1 & 0.375 & 50 & .367 & .430 & .081 & .144 & .407 && .196 & .207 & .099 & .162 & .392\\
Case 1 & 0.375 & 200 & .998 & 1.00 & .290 & .358 & .946 && .938 & .969 & .375 & .483 & .921\\
Case 1 & 0.500 & 50 & .621 & .725 & .098 & .178 & .675 && .349 & .395 & .116 & .221 & .641\\
Case 1 & 0.500 & 200 & 1.00 & 1.00 & .463 & .494 & .996 && .996 & 1.00 & .535 & .580 & .993\\
Case 2 & 0.125 & 50 & .064 & .058 & .045 & .072 & .089 && .050 & .063 & .068 & .082 & .075\\
Case 2 & 0.125 & 200 & .149 & .187 & .069 & .097 & .057 && .099 & .101 & .077 & .087 & .066\\
Case 2 & 0.250 & 50 & .104 & .139 & .072 & .111 & .061 && .080 & .079 & .069 & .089 & .090\\
Case 2 & 0.250 & 200 & .716 & .937 & .197 & .266 & .056 && .377 & .615 & .245 & .321 & .078\\
Case 2 & 0.375 & 50 & .240 & .365 & .091 & .154 & .051 && .118 & .162 & .118 & .187 & .071\\
Case 2 & 0.375 & 200 & .994 & 1.00 & .355 & .392 & .057 && .850 & .965 & .444 & .508 & .084\\
Case 2 & 0.500 & 50 & .505 & .726 & .124 & .215 & .055 && .249 & .358 & .174 & .273 & .071\\
Case 2 & 0.500 & 200 & 1.00 & 1.00 & .492 & .502 & .057 && .991 & .999 & .618 & .638 & .098\\
\bottomrule
\end{tabular}
\begin{tablenotes}\footnotesize
\item Note: Split$_\gamma$: Split Sample test with training proportion $\gamma$; Homog.: merged selective test of homogeneity $F_{homo,r}$ with $r=-20$; Homog. J.: joint selective test of homogeneity (Remark~\ref{rem:joint}); O-EPA: overall EPA test. Case~1: O-EPA fails; Case~2: O-EPA holds.
\end{tablenotes}
\end{threeparttable}
\end{table}

\subsection{Merging parameter}

Table~\ref{tab:supp_r} reports the merged tests $F_{SI,r}$ and $F_{homo,r}$ for $r\in\{-2,-5,-10,-20,-\infty\}$. The rejection rates increase with $|r|$ and are almost the same for $r\le-10$; $r=-2$, closer to the arithmetic mean, is conservative and less powerful. Size is controlled for all values of $r$.

\begin{table}[htbp]
\centering
\begin{threeparttable}
\caption{Rejection rates of the merged selective tests for several values of $r$, unconditional tests}\label{tab:supp_r}
\scriptsize\setlength{\tabcolsep}{3pt}\renewcommand{\arraystretch}{0.9}
\begin{tabular}{lrr ccccccccccc}
\toprule
&&& \multicolumn{5}{c}{$F_{SI,r}$} & \multicolumn{5}{c}{$F_{homo,r}$}\\ \cmidrule(lr){4-8}\cmidrule(lr){9-13}
Design & $\psi$ & $T$ & $-2$ & $-5$ & $-10$ & $-20$ & $-\infty$ & $-2$ & $-5$ & $-10$ & $-20$ & $-\infty$\\
\midrule
Null & 0.000 & 20 & .029 & .038 & .041 & .041 & .043 & .025 & .031 & .035 & .037 & .037\\
Null & 0.000 & 50 & .025 & .048 & .055 & .057 & .060 & .033 & .044 & .047 & .049 & .050\\
Null & 0.000 & 100 & .028 & .047 & .053 & .059 & .064 & .043 & .058 & .060 & .062 & .066\\
Null & 0.000 & 200 & .023 & .039 & .042 & .047 & .048 & .030 & .036 & .040 & .040 & .041\\
Case 1 & 0.125 & 50 & .028 & .046 & .055 & .059 & .063 & .034 & .047 & .050 & .052 & .054\\
Case 1 & 0.125 & 200 & .089 & .126 & .142 & .145 & .150 & .051 & .067 & .075 & .079 & .080\\
Case 1 & 0.250 & 50 & .076 & .096 & .108 & .112 & .119 & .039 & .050 & .057 & .058 & .059\\
Case 1 & 0.250 & 200 & .400 & .481 & .503 & .514 & .519 & .129 & .160 & .172 & .176 & .177\\
Case 1 & 0.375 & 50 & .158 & .209 & .227 & .234 & .241 & .053 & .072 & .077 & .081 & .082\\
Case 1 & 0.375 & 200 & .781 & .827 & .843 & .849 & .850 & .230 & .270 & .283 & .290 & .299\\
Case 1 & 0.500 & 50 & .294 & .377 & .391 & .396 & .402 & .062 & .085 & .095 & .098 & .102\\
Case 1 & 0.500 & 200 & .955 & .973 & .975 & .977 & .979 & .385 & .439 & .456 & .463 & .472\\
Case 2 & 0.125 & 50 & .033 & .050 & .054 & .058 & .060 & .033 & .044 & .044 & .045 & .048\\
Case 2 & 0.125 & 200 & .041 & .056 & .060 & .061 & .064 & .054 & .063 & .065 & .069 & .072\\
Case 2 & 0.250 & 50 & .046 & .065 & .072 & .073 & .073 & .048 & .063 & .070 & .072 & .075\\
Case 2 & 0.250 & 200 & .130 & .167 & .177 & .182 & .186 & .145 & .183 & .192 & .197 & .201\\
Case 2 & 0.375 & 50 & .059 & .078 & .089 & .092 & .094 & .059 & .086 & .088 & .091 & .093\\
Case 2 & 0.375 & 200 & .257 & .316 & .326 & .334 & .341 & .275 & .328 & .349 & .355 & .363\\
Case 2 & 0.500 & 50 & .069 & .098 & .108 & .112 & .118 & .078 & .115 & .122 & .124 & .132\\
Case 2 & 0.500 & 200 & .405 & .456 & .468 & .479 & .483 & .425 & .473 & .485 & .492 & .499\\
\bottomrule
\end{tabular}
\begin{tablenotes}\footnotesize
\item Note: See Table~\ref{tab:supp_main}. The column headings give the value of $r$.
\end{tablenotes}
\end{threeparttable}
\end{table}

\subsection{Conditioning event, decomposition, initialization and tuning parameters}

Table~\ref{tab:supp_var} varies one element of the selective tests at a time, for the unconditional tests and $T\in\{50,200\}$. Base is the specification of Section~\ref{sec:monte}. Paths conditions in addition on the paths of all initializations, as in \cite{chen23}. Eucl. uses the Euclidean decomposition, that is $L(L'L)^{-1}\widehat\Sigma$ instead of $\widehat\Xi L$, and Shrink the shrinkage estimator of $\Xi L$ of Section~\ref{sec:selective}. $\mathbb F$ uses the Hotelling law of the statistic with $B$ cosines as reference instead of the $\chi$ law. R.part. initializes Algorithm~\ref{algo} by a random assignment of the units instead of $K$ randomly chosen units. $S$ is the number of initializations, $K=3$ fixes the number of clusters at its true value, $\varsigma=3$ is the penalty of the information criterion, and $c_B$ the constant of $B=\lceil c_BT^{2/3}\rceil$ instead of the data-dependent rule.

Four results stand out. First, conditioning on the paths of all initializations controls size but leaves the pairwise tests without power: the truncation sets are then so short that the selective $p$-values are close to uniform under the alternative, which confirms the choice of conditioning on the selection outcome only. Second, the shrinkage estimator of $\Xi L$ roughly doubles the power of all selective tests, with size close to the nominal level for $T=200$ and slightly above it for $T=50$; in this design the Euclidean decomposition performs similarly, which indicates that the dependence between its residual and the tested difference is weak here, although its validity is not guaranteed in general. Third, the $\mathbb F$ reference law makes the tests very conservative, because the truncation amplifies the heavier tail of the reference law, and we do not recommend it. Fourth, the results are insensitive to the number and type of initializations, while fixing $K$ at its true value, a larger penalty $\varsigma$, which selects fewer clusters, and larger values of $c_B$ increase power, with size close to the nominal level in all cases.

\begin{table}[htbp]
\centering
\begin{threeparttable}
\caption{Rejection rates of $F_{SI,r}$ under alternative specifications, unconditional tests}\label{tab:supp_var}
\scriptsize\setlength{\tabcolsep}{2.5pt}\renewcommand{\arraystretch}{0.9}
\begin{tabular}{lrrccccccccccccc}
\toprule
Design & $\psi$ & $T$ & Base & Paths & Eucl. & Shrink & $\mathbb F$ & R.part. & $S=1$ & $S=100$ & $K=3$ & $\varsigma=3$ & $c_B=0.5$ & $c_B=1$ & $c_B=2$\\
\midrule
Null & 0.00 & 50 & .043 & .042 & .066 & .065 & .032 & .051 & .050 & .046 & .040 & .061 & .055 & .057 & .052\\
Null & 0.00 & 200 & .045 & .038 & .050 & .059 & .032 & .042 & .040 & .043 & .047 & .041 & .039 & .041 & .048\\
Case 1 & 0.25 & 50 & .110 & .093 & .159 & .161 & .086 & .114 & .110 & .111 & .107 & .175 & .118 & .125 & .193\\
Case 1 & 0.25 & 200 & .525 & .460 & .618 & .621 & .464 & .512 & .518 & .532 & .554 & .687 & .541 & .588 & .667\\
Case 1 & 0.50 & 50 & .412 & .376 & .549 & .559 & .375 & .407 & .422 & .412 & .461 & .563 & .394 & .501 & .608\\
Case 1 & 0.50 & 200 & .979 & .966 & .994 & .994 & .966 & .978 & .975 & .978 & .992 & .996 & .977 & .993 & .999\\
Case 2 & 0.25 & 50 & .072 & .044 & .118 & .113 & .042 & .067 & .062 & .073 & .077 & .108 & .079 & .073 & .080\\
Case 2 & 0.25 & 200 & .181 & .034 & .381 & .382 & .039 & .171 & .165 & .187 & .213 & .324 & .176 & .229 & .253\\
Case 2 & 0.50 & 50 & .121 & .033 & .354 & .346 & .036 & .116 & .104 & .134 & .146 & .216 & .132 & .131 & .199\\
Case 2 & 0.50 & 200 & .479 & .033 & .926 & .918 & .033 & .439 & .440 & .494 & .645 & .739 & .493 & .609 & .697\\
\bottomrule
\end{tabular}
\begin{tablenotes}\footnotesize
\item Note: Base: specification of Section~\ref{sec:monte}; Paths: conditioning on the paths of all initializations; Eucl.: Euclidean decomposition; Shrink: shrinkage estimator of $\Xi L$; $\mathbb F$: Hotelling reference law; R.part.: random-partition initialization; $S$: number of initializations; $K=3$: number of clusters fixed; $\varsigma=3$: penalty of the information criterion; $c_B$: constant of $B=\lceil c_BT^{2/3}\rceil$.
\end{tablenotes}
\end{threeparttable}
\end{table}
\begin{table}[htbp]
\centering
\begin{threeparttable}
\caption{Rejection rates of the joint selective C-EPA test under alternative specifications, unconditional tests}\label{tab:supp_var_joint}
\scriptsize\setlength{\tabcolsep}{2.5pt}\renewcommand{\arraystretch}{0.9}
\begin{tabular}{lrrccccccccccccc}
\toprule
Design & $\psi$ & $T$ & Base & Paths & Eucl. & Shrink & $\mathbb F$ & R.part. & $S=1$ & $S=100$ & $K=3$ & $\varsigma=3$ & $c_B=0.5$ & $c_B=1$ & $c_B=2$\\
\midrule
Null & 0.00 & 50 & .068 & .050 & .092 & .094 & .023 & .070 & .069 & .084 & .065 & .077 & .076 & .074 & .066\\
Null & 0.00 & 200 & .066 & .045 & .053 & .057 & .026 & .061 & .070 & .069 & .061 & .053 & .066 & .054 & .058\\
Case 1 & 0.25 & 50 & .109 & .056 & .162 & .166 & .023 & .114 & .106 & .109 & .117 & .135 & .129 & .116 & .122\\
Case 1 & 0.25 & 200 & .274 & .056 & .388 & .391 & .029 & .241 & .223 & .279 & .278 & .391 & .264 & .304 & .329\\
Case 1 & 0.50 & 50 & .226 & .052 & .392 & .397 & .039 & .200 & .195 & .222 & .226 & .281 & .218 & .216 & .253\\
Case 1 & 0.50 & 200 & .496 & .053 & .713 & .704 & .046 & .490 & .484 & .511 & .605 & .712 & .486 & .544 & .573\\
Case 2 & 0.25 & 50 & .107 & .039 & .160 & .159 & .029 & .115 & .113 & .119 & .107 & .135 & .117 & .107 & .104\\
Case 2 & 0.25 & 200 & .307 & .055 & .403 & .416 & .034 & .290 & .260 & .293 & .312 & .410 & .288 & .298 & .334\\
Case 2 & 0.50 & 50 & .244 & .046 & .402 & .399 & .038 & .215 & .217 & .244 & .255 & .337 & .256 & .232 & .269\\
Case 2 & 0.50 & 200 & .541 & .061 & .765 & .760 & .037 & .496 & .507 & .544 & .662 & .764 & .544 & .584 & .633\\
\bottomrule
\end{tabular}
\begin{tablenotes}\footnotesize
\item Note: Base: specification of Section~\ref{sec:monte}; Paths: conditioning on the paths of all initializations; Eucl.: Euclidean decomposition; Shrink: shrinkage estimator of $\Xi L$; $\mathbb F$: Hotelling reference law; R.part.: random-partition initialization; $S$: number of initializations; $K=3$: number of clusters fixed; $\varsigma=3$: penalty of the information criterion; $c_B$: constant of $B=\lceil c_BT^{2/3}\rceil$.
\end{tablenotes}
\end{threeparttable}
\end{table}
\begin{table}[htbp]
\centering
\begin{threeparttable}
\caption{Rejection rates of the selective homogeneity test $F_{homo,r}$ under alternative specifications, unconditional tests}\label{tab:supp_var_hom}
\scriptsize\setlength{\tabcolsep}{2.5pt}\renewcommand{\arraystretch}{0.9}
\begin{tabular}{lrrccccccccccccc}
\toprule
Design & $\psi$ & $T$ & Base & Paths & Eucl. & Shrink & $\mathbb F$ & R.part. & $S=1$ & $S=100$ & $K=3$ & $\varsigma=3$ & $c_B=0.5$ & $c_B=1$ & $c_B=2$\\
\midrule
Null & 0.00 & 50 & .035 & .034 & .061 & .059 & .021 & .043 & .046 & .044 & .032 & .052 & .047 & .058 & .036\\
Null & 0.00 & 200 & .040 & .029 & .051 & .057 & .024 & .038 & .039 & .038 & .044 & .051 & .042 & .040 & .044\\
Case 1 & 0.25 & 50 & .049 & .034 & .114 & .118 & .020 & .055 & .052 & .053 & .057 & .114 & .062 & .055 & .071\\
Case 1 & 0.25 & 200 & .169 & .034 & .352 & .355 & .039 & .158 & .165 & .177 & .194 & .346 & .156 & .209 & .274\\
Case 1 & 0.50 & 50 & .107 & .025 & .316 & .331 & .024 & .095 & .109 & .115 & .143 & .219 & .108 & .119 & .157\\
Case 1 & 0.50 & 200 & .442 & .025 & .876 & .857 & .040 & .422 & .422 & .455 & .612 & .706 & .438 & .573 & .652\\
Case 2 & 0.25 & 50 & .059 & .028 & .117 & .106 & .025 & .060 & .057 & .064 & .068 & .110 & .069 & .070 & .057\\
Case 2 & 0.25 & 200 & .202 & .032 & .407 & .399 & .032 & .194 & .185 & .198 & .235 & .373 & .193 & .248 & .272\\
Case 2 & 0.50 & 50 & .128 & .027 & .372 & .366 & .025 & .115 & .106 & .136 & .160 & .259 & .134 & .140 & .205\\
Case 2 & 0.50 & 200 & .501 & .029 & .929 & .920 & .031 & .457 & .451 & .520 & .676 & .754 & .508 & .635 & .710\\
\bottomrule
\end{tabular}
\begin{tablenotes}\footnotesize
\item Note: Base: specification of Section~\ref{sec:monte}; Paths: conditioning on the paths of all initializations; Eucl.: Euclidean decomposition; Shrink: shrinkage estimator of $\Xi L$; $\mathbb F$: Hotelling reference law; R.part.: random-partition initialization; $S$: number of initializations; $K=3$: number of clusters fixed; $\varsigma=3$: penalty of the information criterion; $c_B$: constant of $B=\lceil c_BT^{2/3}\rceil$.
\end{tablenotes}
\end{threeparttable}
\end{table}

\subsection{A design with strongly heterogeneous long-run covariances}

To compare the decompositions in a design in which the eigenvector condition of the Euclidean decomposition is far from satisfied, we generate the loss differentials directly as $Z_{it}=\psi_{k^0_i}+\lambda_iF_t+\sigma_ie_{it}$, where $F_t$ is a stationary AR(1) factor with coefficient 0.5 and unit variance, the loadings $\lambda_i\sim U(0,2)$, half of the units have $\sigma_i=0.5$ and the others $\sigma_i=2$, and $e_{it}$ are unit-variance AR(1) processes with coefficients drawn from $U(0,0.6)$, all fixed across replications. The clusters are those of Section~\ref{sec:monte}, with $\psi=0$ under the null hypothesis and $(\psi_1,\psi_2,\psi_3)=\psi\cdot(-1.2,-0.8,1)$, $\psi\in\{0.25,0.5\}$, under the alternative. Table~\ref{tab:supp_euclfail} reports the rejection rates of the selective tests with the metric of $\Xi$, the Euclidean decomposition and the shrinkage estimator of $\Xi L$. With the metric of $\Xi$, $F_{SI,r}$ and the homogeneity test are slightly conservative, with rejection rates between 0.027 and 0.045 under the null hypothesis for $T\geq50$, and the joint selective test slightly over-rejects (0.083 to 0.086). Although Lemma~\ref{lemma:independence} does not cover the Euclidean decomposition in this design, its rejection rates and those of the shrinkage estimator lie between 0.046 and 0.081 for $T\geq50$, so the failure of the eigenvector condition does not produce noticeable size distortions. Both are much more powerful than the metric of $\Xi$: with $\psi=0.5$ and $T=200$, $F_{SI,r}$ rejects in 96.3\% and 82.6\% of the replications with the Euclidean decomposition and the shrinkage estimator, against 42.8\% with the metric of $\Xi$, which moves each unit along the perturbation path by its own estimated amount and is more affected by the noise in $\widehat\Xi L$ when the long-run variances are heterogeneous. With $T=20$, all tests over-reject, and so does the O-EPA test (0.168), because of the persistent common factor.

\begin{table}[htbp]
\centering
\begin{threeparttable}
\caption{Rejection rates with strongly heterogeneous long-run covariances}\label{tab:supp_euclfail}
\footnotesize\setlength{\tabcolsep}{4pt}\renewcommand{\arraystretch}{0.9}
\begin{tabular}{rr ccc ccc ccc}
\toprule
&& \multicolumn{3}{c}{Metric of $\Xi$} & \multicolumn{3}{c}{Euclidean} & \multicolumn{3}{c}{Shrinkage}\\ \cmidrule(lr){3-5}\cmidrule(lr){6-8}\cmidrule(lr){9-11}
$\psi$ & $T$ & $F_{SI,r}$ & Joint & Homog. & $F_{SI,r}$ & Joint & Homog. & $F_{SI,r}$ & Joint & Homog.\\
\midrule
0.00 & 20 & .081 & .124 & .046 & .145 & .147 & .117 & .139 & .149 & .113\\
0.00 & 50 & .034 & .086 & .040 & .057 & .081 & .061 & .055 & .079 & .058\\
0.00 & 100 & .030 & .083 & .027 & .056 & .074 & .058 & .051 & .073 & .049\\
0.00 & 200 & .040 & .084 & .045 & .046 & .066 & .047 & .050 & .067 & .052\\
0.25 & 50 & .053 & .223 & .050 & .206 & .369 & .218 & .173 & .371 & .182\\
0.25 & 200 & .130 & .436 & .145 & .581 & .675 & .605 & .447 & .678 & .469\\
0.50 & 50 & .162 & .457 & .168 & .632 & .714 & .656 & .522 & .714 & .543\\
0.50 & 200 & .428 & .698 & .432 & .963 & .914 & .964 & .826 & .912 & .832\\
\bottomrule
\end{tabular}
\begin{tablenotes}\footnotesize
\item Note: Unconditional tests, $N=80$, 1000 replications. $\psi=0$: null hypothesis; $\psi>0$: $(\psi_1,\psi_2,\psi_3)=\psi\cdot(-1.2,-0.8,1)$, O-EPA holds. Joint: joint selective C-EPA test; Homog.: merged homogeneity test $F_{homo,r}$.
\end{tablenotes}
\end{threeparttable}
\end{table}

\subsection{Number of clusters}

Table~\ref{tab:supp_ic} reports the average estimated number of clusters, the frequency with which the true number $K^0=3$ is selected and the adjusted Rand index between the estimated and true partitions, for the information criterion with $\varsigma\in\{1.5,3\}$ and for 10-fold cross-validation over time. The criterion with $\varsigma=1.5$ selects $K^0$ in 50\% to 84\% of the replications for small and moderate separations, but tends to select too many clusters when the separation is large ($\psi=0.5$ and $T=200$), where $\varsigma=3$ performs better; cross-validation selects too few clusters, mostly two. The adjusted Rand indices are moderate, as expected with overlapping clusters, and increase with $\psi$ and $T$. Under the null hypothesis, the criterion with $\varsigma=1.5$ selects two or three clusters, which the selective tests take into account.

\begin{table}[htbp]
\centering
\begin{threeparttable}
\caption{Estimation of the number of clusters, unconditional loss differentials}\label{tab:supp_ic}
\footnotesize\setlength{\tabcolsep}{3pt}\renewcommand{\arraystretch}{0.9}
\begin{tabular}{lrr ccc ccc ccc}
\toprule
&&& \multicolumn{3}{c}{IC, $\varsigma=1.5$} & \multicolumn{3}{c}{IC, $\varsigma=3$} & \multicolumn{3}{c}{Cross-validation}\\ \cmidrule(lr){4-6}\cmidrule(lr){7-9}\cmidrule(lr){10-12}
Design & $\psi$ & $T$ & $\overline{K}$ & $\widehat K=K^0$ & ARI & $\overline{K}$ & $\widehat K=K^0$ & ARI & $\overline{K}$ & $\widehat K=K^0$ & ARI\\
\midrule
Case 1 & 0.125 & 50 & 2.87 & .794 & .046 & 2.07 & .067 & .060 & 2.21 & .114 & .058\\
Case 1 & 0.125 & 200 & 2.79 & .737 & .162 & 2.02 & .024 & .210 & 2.31 & .153 & .196\\
Case 1 & 0.250 & 50 & 2.99 & .844 & .122 & 2.11 & .114 & .161 & 2.27 & .148 & .156\\
Case 1 & 0.250 & 200 & 3.20 & .761 & .334 & 2.15 & .148 & .439 & 2.45 & .170 & .415\\
Case 1 & 0.375 & 50 & 3.21 & .774 & .203 & 2.21 & .213 & .260 & 2.40 & .193 & .255\\
Case 1 & 0.375 & 200 & 3.61 & .392 & .422 & 2.41 & .402 & .546 & 2.40 & .158 & .561\\
Case 1 & 0.500 & 50 & 3.37 & .631 & .271 & 2.40 & .399 & .330 & 2.49 & .198 & .336\\
Case 1 & 0.500 & 200 & 3.90 & .127 & .466 & 2.59 & .555 & .590 & 2.31 & .122 & .649\\
Case 2 & 0.125 & 50 & 2.88 & .807 & .053 & 2.06 & .058 & .069 & 2.18 & .090 & .069\\
Case 2 & 0.125 & 200 & 2.82 & .752 & .173 & 2.04 & .035 & .231 & 2.33 & .148 & .215\\
Case 2 & 0.250 & 50 & 3.05 & .826 & .136 & 2.13 & .133 & .183 & 2.27 & .145 & .179\\
Case 2 & 0.250 & 200 & 3.21 & .744 & .364 & 2.17 & .169 & .472 & 2.43 & .191 & .448\\
Case 2 & 0.375 & 50 & 3.22 & .771 & .233 & 2.29 & .283 & .304 & 2.37 & .191 & .305\\
Case 2 & 0.375 & 200 & 3.71 & .295 & .448 & 2.42 & .410 & .587 & 2.41 & .127 & .596\\
Case 2 & 0.500 & 50 & 3.51 & .498 & .306 & 2.50 & .496 & .379 & 2.42 & .159 & .402\\
Case 2 & 0.500 & 200 & 3.96 & .076 & .488 & 2.63 & .561 & .631 & 2.26 & .109 & .681\\
\bottomrule
\end{tabular}
\begin{tablenotes}\footnotesize
\item Note: $\overline K$: average estimated number of clusters; $\widehat K=K^0$: frequency of selecting the true number of clusters $K^0=3$; ARI: average adjusted Rand index between the estimated and true partitions.
\end{tablenotes}
\end{threeparttable}
\end{table}

\subsection{Structural instability}\label{sec:supp_breaks}

To study structural instability, for $t \in \{ 1,\dots,T/2 \}$ we set $(\psi_1, \psi_2, \psi_3) = \psi/2 \cdot (1, 1, 1) + \psi \cdot (-1.2, -0.8, 1)$ as in Case 1, and for $t \in \{ T/2+1,\dots,T \}$ we set $(\psi_1, \psi_2, \psi_3) = -\psi/2 \cdot (1, 1, 1) - \psi \cdot (-1.2, -0.8, 1)$, with $\psi=0.25$, so that there is no global improvement in predictive ability. The conditional expectations then vary over time, and the null hypotheses hold in the time-averaged sense of Section~\ref{sec:frame} of the paper. The full-sample tests target these time averages, which are zero. The OS estimator that demeans by the full-sample average absorbs the shift in the mean and overestimates the long-run variance, which makes these tests conservative rather than exact. The Split Sample test instead targets the mean of the testing sample, which is not zero: with training proportion $\gamma\geq1/2$, the testing sample lies after the break, where the cluster means are $-(\psi_1,\psi_2,\psi_3)$, and with $\gamma<1/2$ it contains the break and its cluster means are $-\gamma/(1-\gamma)$ times those before the break. Figure~\ref{fig:cepa_size} and Table~\ref{tab:supp_breaks} report the rejection rates. The Predetermined and $F_{SI,r}$ tests reject in 2.0\% to 7.6\% of the replications, with no tendency to increase with $T$, and the unconditional $F_{SI,r}$ test is conservative for $T\geq50$; the joint selective test rejects in 6.5\% to 10.1\% of the replications. With $\gamma=0.5$, the rejection rate of the Split Sample test grows with $T$ and reaches 94.9\% for the unconditional test and 70.8\% for the conditional test with $T=200$, so that this test is not valid for the time-averaged hypotheses. With $\gamma=0.2$, it rejects in only 5.5\% to 7.8\% of the replications, because the shift inside the testing sample inflates the estimated long-run variance at the same rate as the squared mean of the testing sample, so that the statistic does not diverge with an essentially fixed number of cosines; this masks rather than removes the invalidity. Computing the long-run variance from locally demeaned data, as in \cite{harvey2024testing}, removes the effect of the shift but underestimates the long-run variance, because the kernel smoother also removes the low-frequency variation that the OS estimator uses, and the selective tests then reject in 9.0\% to 17.6\% of the replications; we therefore use global demeaning.

\begin{figure}[htbp]
  \centering
  \begin{subfigure}[t]{0.45\textwidth}
    \centering
    \includegraphics[width=\textwidth]{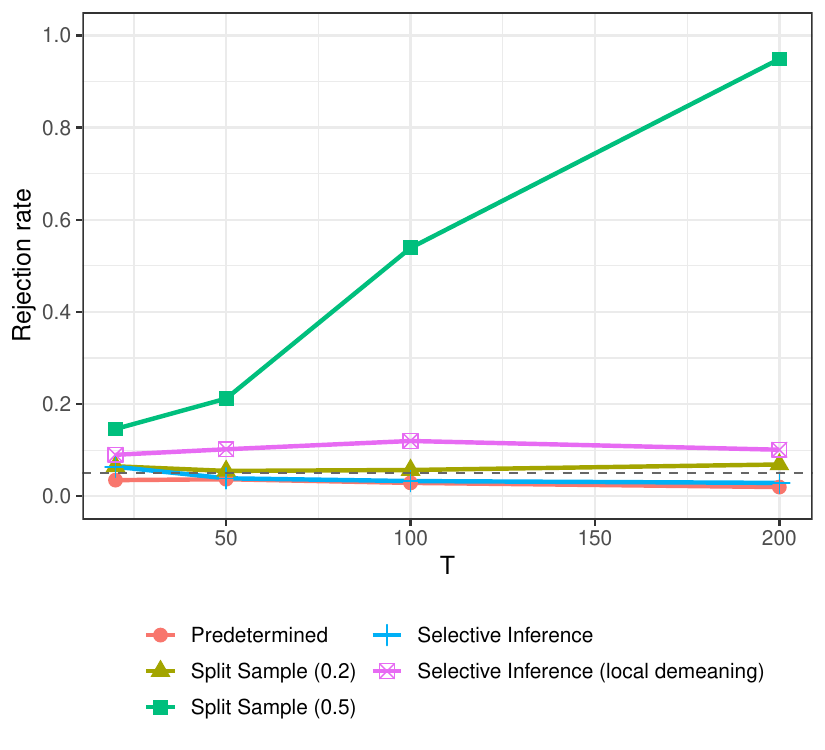}
    \caption{Unconditional tests ($H_{i,t-1}=1$)}
  \end{subfigure}
  \hfill
  \begin{subfigure}[t]{0.45\textwidth}
    \centering
    \includegraphics[width=\textwidth]{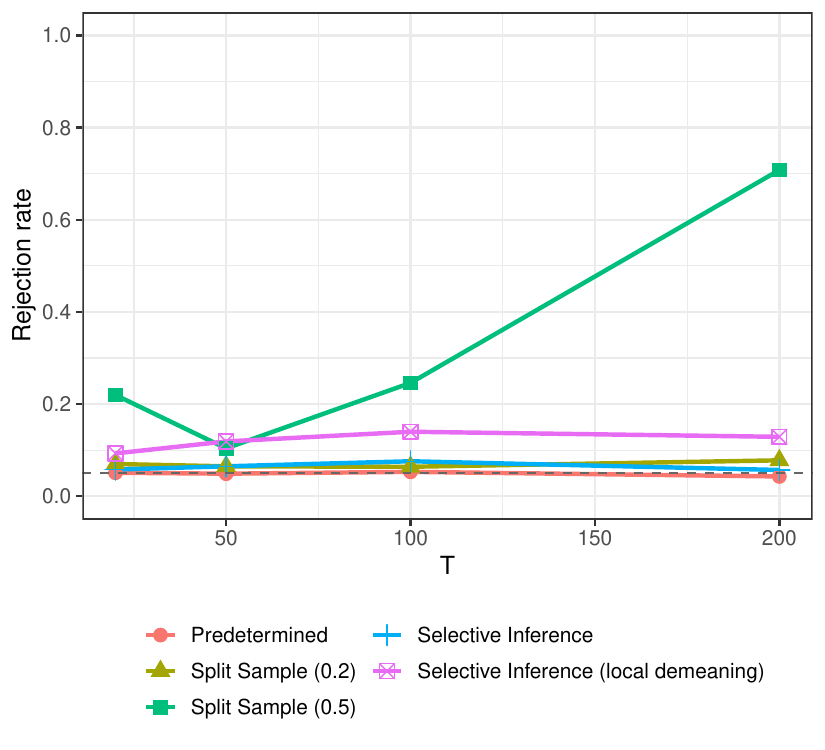}
    \caption{Conditional tests ($H_{i,t-1}=(1,Y_{i,t-1})'$)}
  \end{subfigure}
  \caption{Empirical size of C-EPA tests under structural instability ($\alpha = 0.05$, $N = 80$, $\psi = 0.25$).}
  \label{fig:cepa_size}
\end{figure}

Table~\ref{tab:supp_breaks} also reports the rejection rates of the O-EPA test and of the Predetermined and O-EPA tests with local demeaning. Local demeaning makes all full-sample tests oversized, with rejection rates between 0.09 and 0.18.

\begin{table}[htbp]
\centering
\begin{threeparttable}
\caption{Rejection rates under structural instability ($\psi=0.25$)}\label{tab:supp_breaks}
\footnotesize\setlength{\tabcolsep}{3pt}\renewcommand{\arraystretch}{0.9}
\begin{tabular}{rl cccccccccc}
\toprule
$T$ && Pred. & Pred.$_{loc}$ & Split$_{0.2}$ & Split$_{0.5}$ & $F_{SI,r}$ & Joint & $F_{SI,r,loc}$ & Joint$_{loc}$ & O-EPA & O-EPA$_{loc}$\\
\midrule
20 & Unc. & .035 & .138 & .065 & .146 & .064 & .077 & .090 & .120 & .057 & .129\\
 & Cond. & .051 & .132 & .070 & .220 & .058 & .101 & .093 & .130 & .062 & .115\\
50 & Unc. & .037 & .146 & .055 & .212 & .039 & .065 & .102 & .142 & .040 & .126\\
 & Cond. & .049 & .158 & .065 & .104 & .065 & .100 & .119 & .143 & .058 & .111\\
100 & Unc. & .029 & .152 & .057 & .539 & .033 & .067 & .120 & .176 & .033 & .129\\
 & Cond. & .053 & .184 & .064 & .246 & .076 & .100 & .140 & .168 & .050 & .145\\
200 & Unc. & .020 & .106 & .069 & .949 & .029 & .065 & .101 & .126 & .033 & .117\\
 & Cond. & .043 & .171 & .078 & .708 & .057 & .085 & .129 & .139 & .041 & .120\\
\bottomrule
\end{tabular}
\begin{tablenotes}\footnotesize
\item Note: Unc.: unconditional tests; Cond.: conditional tests. The subscript $loc$ denotes the long-run variance computed from locally demeaned data. Split$_{\gamma}$: Split Sample test with training proportion $\gamma$.
\end{tablenotes}
\end{threeparttable}
\end{table}

\subsection{A true cluster split by Panel Kmeans}

To study the case in which Panel Kmeans splits a true cluster, we generate two true clusters of equal size with expected loss differentials $\pm a$ and AR coefficients 0.1 and 0.3, and estimate the partition with $K=3$. When two of the three selected clusters are contained in the same true cluster, their equality hypothesis holds and we compute the selective and naive tests of it. With $a\in\{0.25,0.5\}$ and $T=50$, and with $a=0.25$ and $T=200$, Panel Kmeans almost never produces such a pair, because the estimated clusters mix units of both true clusters. With $a=0.5$ and $T=200$, it does in 27\% of the replications for the unconditional tests and 12\% for the conditional tests; among them, the selective test rejects in 4.7\% and 9.2\% of the cases and the naive test in 100\% and 99.2\%. The selective test thus has rejection rates close to the nominal level for the selected pair, given the small number of such replications, also when the null hypothesis holds for it and not for the other clusters, whereas the naive test almost always rejects.

\subsection{Heterogeneity within the true clusters}

Table~\ref{tab:supp_within} uses expected loss differentials $\psi_i=\theta_{k_i}+su_i$, where the $u_i$ are fixed across replications, sum to zero within each true cluster and have unit variance, and $s\in\{0,0.125,0.25,0.5\}$. With $\theta=0$, C-EPA holds for the true clusters, and the Predetermined test, which uses them, has size close to 5\%. The selective and Split Sample tests reject more often as $s$ and $T$ increase: they test C-EPA for the clusters estimated by Panel Kmeans, which converge to the population Kmeans partition of the $\psi_i$, whose cluster means are not zero when $s>0$. This is the population hypothesis discussed in Section~\ref{sec:selective} of the paper, not a size distortion. With separated true clusters ($\theta$ as in Case~1 with $\psi=0.5$), within-cluster heterogeneity reduces the adjusted Rand index but hardly the power of the tests.

\begin{table}[htbp]
\centering
\begin{threeparttable}
\caption{Rejection rates with heterogeneity within the true clusters, unconditional tests}\label{tab:supp_within}
\footnotesize\setlength{\tabcolsep}{3pt}\renewcommand{\arraystretch}{0.9}
\begin{tabular}{lrr ccccccc c}
\toprule
$\theta$ & $s$ & $T$ & Pred. & Naive & Split & $F_{SI,r}$ & Joint & Homog. & O-EPA & ARI\\
\midrule
$\theta=0$ & 0.000 & 50 & .060 & .996 & .064 & .059 & .086 & .048 & .058 & 0.02\\
$\theta=0$ & 0.000 & 200 & .039 & 1.00 & .063 & .057 & .070 & .054 & .058 & 0.02\\
$\theta=0$ & 0.125 & 50 & .044 & 1.00 & .065 & .046 & .075 & .045 & .061 & 0.02\\
$\theta=0$ & 0.125 & 200 & .044 & 1.00 & .147 & .071 & .120 & .069 & .043 & 0.02\\
$\theta=0$ & 0.250 & 50 & .056 & 1.00 & .116 & .059 & .097 & .055 & .056 & 0.02\\
$\theta=0$ & 0.250 & 200 & .053 & 1.00 & .820 & .115 & .240 & .118 & .052 & 0.02\\
$\theta=0$ & 0.500 & 50 & .053 & 1.00 & .584 & .103 & .242 & .107 & .060 & 0.02\\
$\theta=0$ & 0.500 & 200 & .053 & 1.00 & 1.00 & .392 & .486 & .403 & .045 & 0.03\\
Separated & 0.000 & 50 & 1.00 & 1.00 & .698 & .417 & .225 & .117 & .661 & 0.27\\
Separated & 0.000 & 200 & 1.00 & 1.00 & 1.00 & .976 & .512 & .460 & .997 & 0.47\\
Separated & 0.125 & 50 & 1.00 & 1.00 & .700 & .441 & .224 & .113 & .674 & 0.26\\
Separated & 0.125 & 200 & 1.00 & 1.00 & 1.00 & .980 & .547 & .464 & .995 & 0.41\\
Separated & 0.250 & 50 & 1.00 & 1.00 & .749 & .382 & .254 & .113 & .679 & 0.22\\
Separated & 0.250 & 200 & 1.00 & 1.00 & 1.00 & .982 & .555 & .474 & .999 & 0.32\\
Separated & 0.500 & 50 & 1.00 & 1.00 & .906 & .362 & .294 & .153 & .655 & 0.14\\
Separated & 0.500 & 200 & 1.00 & 1.00 & 1.00 & .979 & .597 & .601 & .998 & 0.17\\
\bottomrule
\end{tabular}
\begin{tablenotes}\footnotesize
\item Note: $\theta=0$: C-EPA holds for the true clusters; Separated: $\theta$ as in Case~1 with $\psi=0.5$. Split: Split Sample test with $\gamma=0.2$; ARI: average adjusted Rand index.
\end{tablenotes}
\end{threeparttable}
\end{table}

\subsection{Unit-level tests of joint EPA}

Table~\ref{tab:supp_jepa} studies the J-EPA hypothesis of Section~\ref{sec:unitlevel} in a design without clusters, in which each unit has its own expected loss differential $\psi_i=dz_i$, with $z_i$ fixed across replications: under the null $d=0$; offsetting effects with $\sum_iz_i=0$, so that O-EPA holds; shifted effects with mean 0.5, so that it fails; and sparse effects, with four units with $|z_i|=4$ and offsetting signs. We compare the selective tests, the joint selective test enhanced by screening on the maximum of the unit-level statistics with budget 0.01 \citep{fan2015power}, and unit-level tests: the maximum of the unit-level statistics with the Bonferroni correction, the standardized average of the unit-level Wald statistics \citep{pesaran2008testing,pesaran2024testing} and the inverse normal combination \citep{choi2001unit}. With $T=200$, the rejection rates under the null hypothesis lie between 0.030 and 0.084; with $T=50$, the Bonferroni test is conservative (0.022 and 0.026) and the rejection rates of the other tests lie between 0.057 and 0.118, the O-EPA test and the inverse normal combination rejecting most often. The unit-level tests are much more powerful than the selective tests against these alternatives, in which there is no cluster structure for Panel Kmeans to exploit; the screening step recovers a large part of this power for the joint selective test. The selective tests are therefore not a substitute for unit-level tests of J-EPA; their value is in the inference on the estimated clusters.

\begin{table}[htbp]
\centering
\begin{threeparttable}
\caption{Rejection rates of tests of joint EPA when every unit has its own expected loss differential}\label{tab:supp_jepa}
\footnotesize\setlength{\tabcolsep}{3pt}\renewcommand{\arraystretch}{0.9}
\begin{tabular}{lrr cccccccc}
\toprule
Design & $d$ & $T$ & O-EPA & Split & $F_{SI,r}$ & Joint & Joint PE & Max-$t$ & $\Delta$ & Choi\\
\midrule
\multicolumn{11}{l}{A. Unconditional ($H_{i,t-1}=1$)}\\
Size & 0.000 & 50 & .080 & .063 & .057 & .074 & .067 & .026 & .077 & .104\\
Size & 0.000 & 200 & .054 & .042 & .041 & .070 & .066 & .039 & .057 & .052\\
Offsetting & 0.125 & 50 & .066 & .062 & .053 & .075 & .070 & .057 & .198 & .235\\
Offsetting & 0.125 & 200 & .052 & .146 & .067 & .120 & .192 & .320 & .893 & .800\\
Offsetting & 0.250 & 50 & .066 & .105 & .063 & .110 & .118 & .124 & .654 & .685\\
Offsetting & 0.250 & 200 & .055 & .795 & .144 & .213 & .622 & .887 & 1.00 & 1.00\\
Offsetting & 0.500 & 50 & .062 & .533 & .114 & .191 & .231 & .322 & 1.00 & .999\\
Offsetting & 0.500 & 200 & .050 & 1.00 & .361 & .486 & .996 & 1.00 & 1.00 & 1.00\\
Shifted & 0.125 & 50 & .084 & .080 & .068 & .091 & .083 & .051 & .195 & .218\\
Shifted & 0.125 & 200 & .161 & .179 & .103 & .106 & .179 & .312 & .877 & .776\\
Shifted & 0.250 & 50 & .153 & .126 & .092 & .108 & .102 & .096 & .586 & .621\\
Shifted & 0.250 & 200 & .503 & .808 & .333 & .230 & .608 & .849 & 1.00 & 1.00\\
Shifted & 0.500 & 50 & .453 & .497 & .241 & .178 & .220 & .267 & .999 & .998\\
Shifted & 0.500 & 200 & .963 & 1.00 & .847 & .446 & .991 & 1.00 & 1.00 & 1.00\\
Sparse & 0.125 & 50 & .062 & .065 & .053 & .077 & .072 & .067 & .138 & .129\\
Sparse & 0.125 & 200 & .058 & .090 & .060 & .088 & .276 & .470 & .497 & .202\\
Sparse & 0.250 & 50 & .071 & .100 & .060 & .102 & .106 & .103 & .279 & .199\\
Sparse & 0.250 & 200 & .058 & .453 & .101 & .161 & .869 & .971 & .965 & .356\\
Sparse & 0.500 & 50 & .069 & .270 & .038 & .152 & .181 & .256 & .523 & .268\\
Sparse & 0.500 & 200 & .048 & .991 & .267 & .328 & .995 & .999 & .999 & .489\\
\midrule
\multicolumn{11}{l}{B. Conditional ($H_{i,t-1}=(1,Y_{i,t-1})'$)}\\
Size & 0.000 & 50 & .089 & .073 & .077 & .087 & .079 & .022 & .067 & .118\\
Size & 0.000 & 200 & .060 & .052 & .065 & .084 & .069 & .030 & .066 & .077\\
Offsetting & 0.125 & 50 & .072 & .055 & .089 & .114 & .098 & .018 & .119 & .226\\
Offsetting & 0.125 & 200 & .068 & .076 & .081 & .107 & .134 & .193 & .723 & .683\\
Offsetting & 0.250 & 50 & .076 & .072 & .081 & .138 & .116 & .055 & .442 & .601\\
Offsetting & 0.250 & 200 & .071 & .388 & .141 & .259 & .511 & .725 & 1.00 & 1.00\\
Offsetting & 0.500 & 50 & .075 & .243 & .102 & .266 & .265 & .198 & .986 & .996\\
Offsetting & 0.500 & 200 & .087 & .994 & .345 & .512 & .983 & 1.00 & 1.00 & 1.00\\
Shifted & 0.125 & 50 & .083 & .063 & .062 & .102 & .094 & .033 & .127 & .235\\
Shifted & 0.125 & 200 & .151 & .128 & .123 & .116 & .140 & .200 & .716 & .673\\
Shifted & 0.250 & 50 & .167 & .084 & .116 & .132 & .123 & .058 & .371 & .531\\
Shifted & 0.250 & 200 & .436 & .490 & .337 & .240 & .438 & .666 & .999 & .999\\
Shifted & 0.500 & 50 & .442 & .237 & .259 & .223 & .221 & .127 & .970 & .984\\
Shifted & 0.500 & 200 & .957 & .991 & .822 & .439 & .942 & .999 & 1.00 & 1.00\\
Sparse & 0.125 & 50 & .059 & .068 & .059 & .090 & .084 & .027 & .096 & .164\\
Sparse & 0.125 & 200 & .062 & .066 & .067 & .105 & .215 & .330 & .357 & .207\\
Sparse & 0.250 & 50 & .072 & .064 & .083 & .125 & .116 & .061 & .188 & .231\\
Sparse & 0.250 & 200 & .058 & .218 & .104 & .183 & .761 & .913 & .884 & .343\\
Sparse & 0.500 & 50 & .067 & .141 & .079 & .189 & .193 & .155 & .324 & .299\\
Sparse & 0.500 & 200 & .077 & .867 & .361 & .492 & .988 & .999 & .999 & .517\\
\bottomrule
\end{tabular}
\begin{tablenotes}\footnotesize
\item Note: Joint PE: joint selective test enhanced by screening on the maximum of the unit-level statistics; Max-$t$: maximum of the unit-level Wald statistics with the Bonferroni correction; $\Delta$: standardized average of the unit-level Wald statistics; Choi: inverse normal combination of the unit-level $p$-values. Split: Split Sample test with $\gamma=0.2$.
\end{tablenotes}
\end{threeparttable}
\end{table}

\end{document}